\documentclass[aos]{imsart}
 \setattribute{journal}{name}{}

\RequirePackage[OT1]{fontenc}
\RequirePackage{amsthm,amsmath}
\RequirePackage[numbers]{natbib}
\RequirePackage[colorlinks,citecolor=blue,urlcolor=blue]{hyperref}
\usepackage{adjustbox}

\usepackage{latexsym,amsmath,amsfonts}
\usepackage{algorithmicx,algpseudocode}
\usepackage{amssymb}
\usepackage{pstricks,pst-node,pst-tree}
\usepackage{thmtools,thm-restate}
\usepackage{graphicx,epsfig,epstopdf}

\usepackage{multirow}

\usepackage{booktabs}
\usepackage{multicol}

\startlocaldefs
\numberwithin{equation}{section}
\theoremstyle{plain}

\newtheorem{theorem}{Theorem}

\newtheorem{corollary}{Corollary}
\newtheorem{lemma}{Lemma}

\newtheorem{remark}{Remark}

\def\R{\mathbb{R}} % real numbers
\def\b{\beta}
\def\Ep{{\rm E}}

\def\P{{\rm P}}

\def\e{\varepsilon}

 \usepackage[margin=4cm]{geometry}
\def\supp{{\rm supp}}

\endlocaldefs

\begin{document}
\begin{frontmatter}
\title{Pivotal Estimation via Self-Normalization for  High-Dimensional Linear Models with Errors in Variables}%\protect\thanksref{T1}}
\runtitle{Pivotal Estimation for Error in Variables}
\thankstext{T1}{First Version: 4/10/2016; Current Version: \today.}

\begin{aug}
\author{\fnms{Alexandre} \snm{Belloni}\ead[label=e1]{abn5@duke.edu}},
%\and
%\author{\fnms{Victor} \snm{Chernozhukov}\ead[label=e2]{vchern@mit.edu}},
%\and
\author{\fnms{Abhishek} \snm{Kaul}\ead[label=e3]{akaul@math.wsu.edu}},
\and
\author{\fnms{Mathieu} \snm{Rosenbaum}\ead[label=e4]{mathieu.rosenbaum@polytechnique.edu}}
%\and
%\author{\fnms{Alexandre B.} \snm{Tsybakov}\ead[label=e5]{alexandre.tsybakov@ensae.fr}}
%\thankstext{t1}{}
%%\thankstext{t2}{}
%%\thankstext{t3}{}
%\runauthor{Belloni, Chernozhukov, Kaul, Rosenbaum and Tsybakov}

%\affiliation{Duke University, MIT,  and IMPA }

\address{The Fuqua School of Business\\
Duke University\\
100 Fuqua Drive\\
Durham, NC 27708\\
\printead{e1}\\
}

%\address{Department of Economics\\ Massachusetts Institute of Technology\\
%52 Memorial Drive\\
%Cambridge, MA 02139\\
%\printead{e2}\\
%}

\address{Department of Mathematics and Statistics \\ Washington State University
\\
Pullman, WA 99164-3113\\
\printead{e3}\\
}

\address{\'Ecole Polytechnique\\
CMAP\\
91128\\
Palaiseau, cedex 05\\
France\\
\printead{e4}
}

%\address{ CREST,
%ENSAE, Universit\'e Paris-Saclay\\
%5, avenue Henry Le Chatelier\\
%91764 Palaiseau Cedex\\
%France\\
%\printead{e5}}
\end{aug}

%%%%%%
\begin{abstract}
We propose a new estimator for the high-dimensional linear regression model with measurement error in the design where the number of coefficients is potentially larger than the sample size. The main novelty of our procedure is that the choice of penalty parameters is pivotal. The estimator is based on applying a self-normalization to the constraints that characterize the estimator. Importantly, we show how to cast the computation of the estimator as the solution of a convex program with second order cone constraints. This allows the use of algorithms with theoretical guarantees and enables reliable implementation. Under sparsity assumptions, we derive $\ell_q$-rates of convergence and show that consistency can be achieved even if the number of regressors exceeds the sample size. We further provide a simple thresholded estimator that yields a provably sparse estimator with similar $\ell_2$ and $\ell_1$-rates of convergence. The thresholds are data-driven and component dependents. Finally, we also study the rates of convergence of estimators that refit the data based on a selected support with possible model selection mistakes. In addition to our finite sample theoretical results that allow for non-i.i.d. data, we also present simulations to compare the performance of the proposed estimators.
\end{abstract}

%\begin{keyword}[class=AMS]
%\kwd[Primary ]{62M05}
%\kwd{62M09}
%\kwd{62G05}
%\kwd[; secondary ]{62P20}\kwd{60J10}
%\end{keyword}


\end{frontmatter}

\section{Introduction}

In this paper, we consider the high-dimensional linear model with observation error in the design
 \begin{equation}\label{model} y_i=x_i^T\beta_0+\xi_i, \ \ \ z_i = x_i + w_i, \ \ \ i=1,\ldots, n,
 \end{equation}
where we observe the response variable $y_i$ and the $p$-dimensional vector $z_i$, and do not observe the covariates $x_i$. The scalar errors $\xi_i$ are zero-mean independent random variables and $(y_i,z_i)$ are independent across $i$. (In particular, these conditions allow for non-i.i.d. designs which increases the applicability of the estimators including the case of missing at random.) The vector $\beta_0 \in \mathbb{R}^p$ is a vector of unknown parameters to be estimated where the dimension $p$ can be much larger than the sample size $n$. We assume that $\beta_0$ is $s$-sparse, that is it has at most $s$ non-zero components. The errors in measurements $w_i$ are assumed to be zero-mean and  independent of $\xi_i$. We also assume that the error in measurement covariance matrix $\Gamma = \frac{1}{n}\sum_{i=1}^n \Ep(w_i w_i^T) $ is diagonal and admits a data-driven estimator $\hat \Gamma$ which is available in several applications as discussed below.

{} Model (\ref{model}) is motivated by many applications where the covariates may have missing values or are observed with noise. For example, in the field of genomics, the gene expression measurements from microarray data are subject to measurement error. Another example is that of microbiome data where each observation vector has a significant proportion of missing components. Many other examples arise in empirical economics and finance, see \cite{GHaus} and \cite{RT1}, and consumer surveys in marketing where random subsets of questions are selected for each consumer to reduce the length of the survey. In these settings a data-driven estimator $\hat\Gamma$ can be constructed based on auxiliary data without measurement errors \cite{reilly1995mean,chen2005measurement} or even based on $(y_i,z_i)_{i=1}^n$ alone as in the case of missing at random\footnote{We refer the reader to \cite{RT2} for a simple transformation of the data that makes the missing at random to have the additive form of (\ref{model}).}  (as one can estimate the frequencies of missing components, see, e.g. \cite{RT2,BCKRT2016b}). It has been well-documented that ignoring this measurement error leads to biased parameter estimates even in the fixed $p$ setting, see for example, \cite{GHaus}, \cite{Fuller1987}, and \cite{CRSC2006}. In the high-dimensional framework considered here it is also crucial to account for such measurement errors. In addition to potentially biased estimation, measurement errors can also impact variables selection performance and influence the choice of various penalty parameters, see \cite{SFT1}.

{} High-dimensional linear models with $p\gg n$ and measurement errors have been studied recently by \cite{BRT2014}, \cite{BRT2014b}, \cite{CChen}, \cite{CChen1}, \cite{datta2015cocolasso},\cite{KK2015}, \cite{KKCL2016}, \cite{LW}, \cite{RT1}, \cite{RT2},    \cite{RZ2015} and \cite{SFT2}. The common thread\footnote{We note that all the cited papers assume independent observations except for \cite{RZ2015} that allows for the measurement error for each covariate to be a dependent vector across observations.} of these papers is to provide estimators along with the corresponding rates of convergence in different norms. Examples of proposed estimators\footnote{These estimators were proposed under various conditions on the design matrix, relations between $s$, $p$ and $n$, and knowledge of some parameters of the problem.} include the orthogonal matching pursuit as defined in \cite{CChen}, the non-convex $\ell_1$-penalized regression studied in \cite{LW} and the conic programming estimator considered in \cite{BRT2014}.
%$$ \hat\beta^{L_1} \in \arg\min_{\|\beta\|_1 \leq b_0\sqrt{s}} \mbox{$\frac{1}{2}$}\beta^T\left(\frac{1}{n}\sum_{i=1}^n z_iz_i^T - \Gamma \right)\beta-\frac{1}{n}\sum_{i=1}^ny_iz_i^T\beta + \lambda_n\|\beta\|_1,$$
%and the conic programming estimator considered in \cite{BRT2014}
% $$\hat\beta^C\in\arg\min_{\beta, t}\left\{ \|\beta\|_1+t : \left\|\frac{1}{n}\sum_{i=1}^nz_i(y_i-z_i^T\beta)+\widehat \Gamma\beta\right\|_\infty \leq \tau_0 + \tau_1 t, \|\beta\|_2\leq t \right\}. $$
%where $\widehat \Gamma$ is an estimate of the matrix $\Gamma$.
In particular, under suitable conditions and appropriate choice of penalty parameters, some of these estimators $\tilde\beta$ can attain $\ell_q$-rates of convergence of the form
\begin{equation}\label{eq:minimaxbound}\|\tilde\beta - \beta_0\|_q \leq C (1+\|\beta_0\|_2)s^{1/q}\sqrt{\frac{\log p}{n}}, \quad 1\le q\le \infty,
 \end{equation}
 where $\|\cdot\|_q$ denotes the $\ell_q$-norm, and $C>0$ is a constant independent of $s,p$ and $n$.
It is shown in  \cite{BRT2014} that these rates are minimax optimal. The rate in (\ref{eq:minimaxbound}) highlights the impact of the errors in measurements via the $\ell_2$-norm term $\|\beta_0\|_2$, which is not present in the case where covariates are observed without error, and the fact that consistency can be achieved in high-dimensional settings even if $p\gg n$. However, estimators suggested in the literature rely on suitable choice of penalty parameters based on some specific knowledge of the model (\ref{model}). To construct these estimators, the variance of the unobserved noise $\xi_i$ and the variance parameters of the measurement noise $w_i$ should typically be known. This is for example the case for the conic programming estimator of \cite{BRT2014}. For some methods, in addition, one needs to have access to the number $s$ of non-zero components or to the $\ell_2$-norm of $\beta_0$.
For example, assuming $\|\beta_0\|_1$ is known, \cite{LW} proposed an estimator defined as the solution of a non-convex
program which can be well approximated by an iterative relaxation procedure. Assuming that the sparsity of $\beta_0$ is known and the non-zero components of $\beta_0$ are
separated from zero in a suitable way, an orthogonal matching pursuit algorithm to estimate$\beta_0$ is introduced in \cite{CChen,CChen1}.

{} In this work, we propose a new estimator of the parameter $\beta_0$ in model (\ref{model}) that achieves the optimal rates of convergence in $\ell_q$-norm under suitable conditions. Moreover, a simple thresholded version of the estimator achieves optimal sparsity, while retaining optimal convergence rates. The main novelty of our procedure is the pivotality of the penalty parameters, which makes the estimator particularly appealing for the practical applications. That is, the penalty parameters do not depend on the number of non-zero components, on the $\ell_2$-norm of $\beta_0$, the variance parameter of the errors $\xi_i$, nor on the variance parameters of the errors in the measurements $w_i$. Furthermore, our estimator is a solution of a convex optimization problem. Finally, we also study the rates of convergence of the estimator that refits the data based on a selected support with possible model selection mistakes.

\subsection*{Notation} Let $J\subseteq \{1,\ldots,p\}$ be a set of integers. We denote by $|J|$ the cardinality of $J$.  For a vector $\theta=(\theta_1,\dots,\theta_p)^T$ in $\mathbb{R}^p$, we denote by $\theta_J$ the vector in $\mathbb{R}^p$ whose $j$-th component satisfies $(\theta_J)_j = \theta_j$ if $j\in J$, and $(\theta_J)_j = 0$ otherwise. We will call $\theta_J$ the restriction of $\theta$ to $J$. We denote the $\ell_q$-norm of a vector $v\in\mathbb{R}^p$ by $\|v\|_q$. The number of non-zero components of a vector $v\in\mathbb{R}^p$ is denoted by $\|v\|_0$. For a matrix $A$, we define $\|A\|_\infty = \max_{i,j}|a_{ij}|$ as the maximum element in absolute value.
A centered random variable $\xi$ will be called  zero-mean subgaussian  with variance parameter $\sigma^2$ if
$\Ep[\exp(t\xi)]\le \exp(t^2\sigma^2/2)$ for all $t\in \mathbb{R}$. A random vector $w\in \mathbb{R}^n$ will be called  zero-mean subgaussian  with variance parameter $\sigma^2$ if all the random variables of the form $v^Tw$ where $\|v\|_2=1$ are zero-mean subgaussian  with variance parameter $\sigma^2$. For a matrix $A$, we denote by $A_{i\cdot}$ and $A_{\cdot j}$ its $i$-th row and $j$-th column, respectively.
We denote by $C,c,C',c'$ positive constants that can be different on different occurencies.

%Throughout the manuscript we say a sequence is $o(1)$ if the limit of the sequence is $0$ as $n\to\infty$.

\section{Estimator via self-normalization}\label{Sec:Pivotal}

Here we propose a pivotal estimator that does not require knowledge of typically unknown parameters. %\footnote{Previous estimator have penalty parameters dependent on the standard deviation of noise $\xi$, number of non-zero components, or $\ell_2$-norm of $\beta_0$.}
Our starting point is the moment condition that characterizes the vector of parameters $\beta_0$ in (\ref{model})
\begin{equation}\label{eq:identification} \frac{1}{n}\sum_{i=1}^n\Ep[ z_i(y_i-z_i^T\beta_0)+\Gamma\beta_0]=0,\end{equation}
where the term $\Gamma\beta_0$ corrects the bias that arises from using the noisy covariates $z$ instead of the unobserved $x$. The moment condition (\ref{eq:identification}) combined with sparsity assumptions on $\beta_0$ motivates the use of penalized methods to cope with high-dimensionality.

{} We now describe the proposed estimation procedure. Let  $(\hat\beta,\hat t,\hat u)$ be a solution of the constrained minimization  problem
\begin{eqnarray}\label{est:pivotal}
& \displaystyle \min_{\beta\in \mathbb{R}^p,t\in \mathbb{R}^p,u\in \mathbb{R}^p}   \|\beta\|_1 + \lambda_t \|t\|_\infty + \lambda_u \|u\|_\infty :  \\
\label{est:pivotal_1}
& \begin{array}{rl}
& \left| \frac{1}{n}\sum_{i=1}^n z_{ij} (y_i - z_i^T\beta) + \hat\Gamma_{jj}\beta_j \right| \leq \tau t_j + (1+\tau)b_\epsilon u_j,
 \\
& \left\{ \frac{1}{n}\sum_{i=1}^n \{ z_{ij} (y_i - z_i^T\beta)+ \hat\Gamma_{jj}\beta_j\}^2 \right\}^{1/2} \leq t_j, \ \  |\beta_j| \leq u_j,
\end{array},  \quad \forall j \leq p,
\end{eqnarray}
where  $z_i=(z_{ij})_{j=1}^p, \beta = (\beta_j)_{j=1}^p,  u = (u_j)_{j=1}^p,  t =(t_j)_{j=1}^p$,
and  $\lambda_t$, $\lambda_u$, $\tau$ are positive tuning parameters set according to Theorem 1 below. As it is standard in the literature, the statistics $(\hat\Gamma_{jj})_{j=1}^p$ are given estimators of the diagonal elements of matrix $\Gamma$ with $b_\epsilon$ being a bound on its precision satisfying, for any $n$,
\begin{equation}
\label{def:be}\P(\|\hat \Gamma - \Gamma\|_\infty > b_\epsilon)\leq \epsilon
\end{equation}
where $\epsilon\in (0,1)$ is a given number. Our methodology requires the knowledge of $b_\epsilon$. In this sense, it may appear not completely pivotal. However, in many situations, one can compute order of magnitude of $b_\epsilon$, typically $c\sqrt{\text{log}(2p/\epsilon)/n}$, and use this quantity in the definition of the estimator, see for example \cite{RT2} and the simulation experiments in Section \ref{sec:simu}.

{} Importantly, our analysis allow for non-i.i.d. settings since in some applications of interest $\Gamma = \frac{1}{n}\sum_{i=1}^n {\rm cov}(w_i)$ is estimable but each ${\rm cov}(w_i)$ might not be. We use $\hat \beta$ as an estimator of $\beta_0$ and we call it the {\it self-normalized conic estimator}.

{} The proposed method has a self-normalization feature, which is related to the square-root Lasso \cite{BCW-SqLASSO,BCW-SqLASSO2,sun2012scaled}, the STIV estimator  \cite{gautier2011high} and the self-tuned Dantzig estimator \cite{gautier2013pivotal}. The use of self normalization in high dimensional linear models have been first used in \cite{BellChenChernHans:nonGauss} via carefully constructed weights. A key point is the direct use of self-normalization in the moment condition (\ref{eq:identification}), cf. the second line of constraints in \eqref{est:pivotal_1}, instead of working with one scalar noise level as self-normalization quantity.
Similar idea was used in the context of instrumental variable regression,  cf. (9.22) in \cite{gautier2011high}, as well as in \cite{gautier2012} that deals with linear model with no measurement errors and studies a program close to (\ref{est:pivotal}) - \eqref{est:pivotal_1} with $\hat\Gamma_{jj}\equiv 0$, $u_j\equiv0$.

Importantly, (\ref{est:pivotal}) - \eqref{est:pivotal_1} is a
tractable convex optimization problem with linear and second order cone constraints, for which computationally efficient solvers exist.  In particular, we used the R package \textsf{Rmosek} for the computation of the estimator.

{} In some settings, it is of interest to work with estimators that are also sparse. However, the use of many second order constraints makes unlikely that the estimator $\hat\beta$ defined by solving (\ref{est:pivotal})-\eqref{est:pivotal_1} is sparse. In such cases, we propose a thresholded version of the self-normalized conic estimator $\hat\beta$. Consider the set of components $\hat T$ defined as
\begin{equation}\label{def:hatT} \widehat T:= \left\{ j \in \{1,\ldots,p\} : |\hat \beta_j|> \tau \frac{\Big\{\frac{1}{n}\sum_{i=1}^n\{z_{ij}(y_i-z_i^T\hat\beta)+\hat\Gamma_{jj}\hat\beta_j\}^2\Big\}^{1/2}}{\frac{1}{n}\sum_{i=1}^nz_{ij}^2} \right\} \end{equation}
where $\hat\beta_j$ are the components of $\hat\beta$.
We define the thresholded self-normalized conic estimator as $\hat\beta_{\hat T}$ (the restriction of $\hat\beta$ to $\hat T$).

\section{Main results}\label{sec:Main}

In this section we state our assumptions and main theoretical results.

\subsection{Regularity conditions}

In what follows, we consider a setting where $s$ and $p$ depend on $n$, and we state the results in the asymptotics as $n$ tends to infinity. Condition A below summarizes the assumptions on the data generating process. \\

{\bf Condition A.} {\it
(i)  The $n\times p$ matrix $X=[x_1;\ldots;x_n]^T$ is deterministic and the vector $\beta_0$ satisfies $\|\beta_0\|_0\leq s$.
 (ii) The elements of the random noise vector $\xi=(\xi_i)_{i=1}^n$ are independent zero-mean subgaussian random variables with variance parameter $\sigma^2_{\xi}\leq C$.
 (iii) The measurement error vectors $(w_i)_{i=1}^n$ are  independent zero-mean subgaussian random vectors with variance parameter $\sigma_w^2\leq C$, having zero covariances: $\Ep[w_{ij}w_{ik}]=0$ for all $1\leq j<k\leq p, \, i=1,\dots,n$. Moreover,  $(w_i)_{i=1}^n$ are independent of  $\xi=(\xi_i)_{i=1}^n$.}\\
%  (iv) The matrix  $$\Gamma = \frac{1}{n}\sum_{i=1}^n \Ep(w_i w_i^T)= \frac{1}{n}\sum_{i=1}^n{\rm diag}(\Ep[w_{i1}^2],\ldots,\Ep[w_{ip}^2])$$
%  has entries uniformly bounded from above by an absolute constant. }

{} Condition A allows for non-i.i.d. settings. Condition A(i) assumes a deterministic design\footnote{The results extends directly to a random design matrix $X$ under standard conditions on the random vectors $(x_i)_{i=1}^n$ in the high-dimensional literature, see \cite{BickelRitovTsybakov2009,van2009conditions,buhlmann2011statistics}.} and the performance of our estimator depends on the Gram matrix $\Psi = \frac{1}{n}X^TX$.  Like in problems without measurement errors, some characteristics of $\Psi$ play a key role in the analysis, see \cite{BickelRitovTsybakov2009}. In this paper, we consider $\ell_q$-sensitivity characteristics defined for $q\geq 1$ as
\begin{equation*}
 \kappa_q(s,u) = \min_{J:|J|\leq s} \big( \min_{\Delta \in C_J(u):\|\Delta\|_q=1} \|\Psi \Delta \|_\infty \big), \end{equation*}
where $C_J(u)=\{\Delta \in \mathbb{R}^p : \|\Delta_{J^c}\|_1 \leq u \|\Delta_J\|_1\}$, $u\geq 0$ and $J\subseteq \{1,\ldots,p\}$. These sensitivity characteristics generalize other well known characteristics such as the restricted eigenvalues of \cite{BickelRitovTsybakov2009}. One can find details on their properties in \cite{gautier2011high} where the notion of  sensitivity characteristic was introduced. Sensitivity characteristics have been used  previously in several papers including \cite{gautier2011high,gautier2013pivotal,RT1,RT2} and \cite{BRT2014}. For well-behaved designs that are prevalent in the literature,  we have  $\kappa_q(s,u)\geq cs^{-1/q}$ for $u\ge 1$ and $q\in [1,2]$, where $c>0$ is a constant, see \cite{gautier2011high}.  Conditions A(ii) and A(iii) are standard in the literature on high-dimensional linear regression with errors in measurements, see \cite{BRT2014,LW} among others. Moreover, our analysis relies on the quantity
$ m_2 := \max_{j=1,\ldots,p}\frac{1}{n}\sum_{i=1}^nx_{ij}^2$
which is typically uniformly bounded for many designs of interest.

For $i=1,\ldots,n$ and $j=1,\ldots,p$, we define  $$U_{ij}=z_{ij}(\xi_i-w_i^T\beta_0)+\Ep[w_{ij}^2]\beta_{0j}$$ and set $\mathcal{U}_{jk}= \{\mbox{$\frac{1}{n}\sum_{i=1}^n$}\Ep[|U_{ij}|^k]\}^{1/k} \ \ \mbox{and} \ \ \Delta_j= |\beta_{0j}|\max_{i=1,\ldots,n}|\Ep[w_{ij}^2]-\Gamma_{jj}|.$

We now state the assumptions on $\mathcal{U}_{jk}$, $\Delta_j$, $p$ and $n$. Let $\Phi(\cdot)$ denote the standard normal c.d.f.\\

{\bf Condition B.} {\it  (i) The estimator $\hat \Gamma$ is a diagonal matrix and $b_\epsilon$ satisfies (\ref{def:be}).
 For some positive sequence $\ell_n\geq 2$ tending to infinity, the following conditions hold:\\
 (ii) $\max_{1\leq j\leq p} \{\mathcal{U}_{j3} /\mathcal{U}_{j2}\}\Phi^{-1}(1-\alpha/(2pn))\leq n^{1/6}/\ell_n$, and \\(iii) ${  \max_{j:\Delta_j>0}}\{n^{-1/6}(\Delta_j/\mathcal{U}_{j2})\}\Phi^{-1}(1-\alpha/(2pn))\leq 1/\ell_n,$ and $n\geq 4\log(2pn/\alpha)$}\\

{} Condition B(i) allows for the use of data-driven estimator of $\Gamma$. Condition B(ii) is a mild moment condition and allows the application of self-normalized moderate deviation theory, see \cite{delapena,jing:etal}. In the case of i.i.d. sampling where the covariates $x_i$ are also drawn from a subgaussian distribution with bounded variance parameter, we have   $\Ep[|U_j|^3]^{1/3}\leq C(1+\|\beta_0\|_2)$ and $\Ep[|U_j|^2]^{1/2} \geq c$ if $\xi$ is independent of $z$. In fact for many designs we have $\max_{j\leq p} \{\Ep[|U_j|^3]^{1/3}/\Ep[|U_j|^2]^{1/2}\} \leq C$ so that Condition B(ii) is satisfied provided $\log^3 p = o(n)$. Condition B(iii) provides a mild sufficient condition to handle the non-i.i.d. case where the terms $\Ep[w_{ij}^2]$ change across $i$ as well. We also view Condition B(iii) a mild moment condition as it is  implied by Condition A, B(ii), $\|\beta_0\|_\infty\leq C$ and $\max_{j=1,\ldots,p}\mathcal{U}_{j3}/\mathcal{U}_{j2}\leq C$.

We also introduce a (computable) data-driven quantity
$$
H_n= \max_{j=1,\ldots,p}\Big\| \frac{1}{n}\sum_{i=1}^n (z_{ij}z_i^T-\hat\Gamma_{j\cdot})^T(z_{ij}z_i^T-\hat\Gamma_{j\cdot}) \Big\|_\infty^{1/2}.
$$
Let  $h_\epsilon$ denote its $(1-\epsilon)$-quantile, so that $ \P( H_n > h_\epsilon ) \leq \epsilon$. Here and in what follows, we assume that $\epsilon\in(0,1)$ is a fixed small number. We will not require the knowledge of $h_\epsilon$ to implement the method. We will only need that $h_\epsilon \tau s = o(1)$ as $n\to\infty$. This is implied by mild moment conditions on $z_i$'s and growth conditions on $p$ and $s$. For many designs, $h_\epsilon$ is uniformly bounded as $\epsilon\to 0$ and as the sample size grows (see Lemma \ref{lem:Hn} in the Appendix).  Finally, in order to state our theoretical results below, it will be convenient to define for any $\beta\in \R^p$ the vector $t(\beta) = (t_j(\beta))_{j=1}^p$ where
$$t_j(\beta) := \left\{\frac{1}{n}\sum_{i=1}^n \{z_{ij} (y_i - z_i^T\beta) + \hat\Gamma_{j\cdot}\beta\}^2\right\}^{1/2},~~j=1,\ldots,p.$$
%We will use $t(\beta_0)$ in the statements of our theoretical results.

\subsection{Properties of self-normalized conic estimator}

The following theorem establishes the rates of convergence of the estimator $\hat\beta$. %The analysis build upon prior literature of such conic estimators for the error-in-measurements model, see \cite{BRT2014,BRT2014b}.

\begin{theorem}\label{thm:pivotal} Let $0<\alpha<1$, $0<\varepsilon<1$, and $1\le q\le \infty$.
Set $\tau = n^{-1/2}\Phi^{-1}(1-\alpha/(2p))$, $\lambda_u = 1/4$ and $\lambda_t = 1/\{4H_n\}$. Assume that \begin{equation}\label{cond:kappaq}\kappa_q(s,3)s^{1/q} \geq 8 s \{ (1+\tau)b_\epsilon + \tau h_\epsilon + C'(1+m_2^{1/2})\sqrt{\log(2p^2/\varepsilon)/n}\}\end{equation} where $C'>0$ is a constant that depends only on $\sigma_w$ and $\sigma_\xi$. Then, under Conditions A and B, for $n$ sufficiently large, with probability at least $1-\alpha\{1+o(1)\}-2\epsilon-9\varepsilon$ we have
$$ \|\hat\beta-\beta_0\|_q \leq \frac{\tau\|t(\beta_0)\|_\infty}{c'\kappa_q(s,3)}  + \frac{(1+\|\beta_0\|_2)(1+m_2^{1/2})}{c'\kappa_q(s,3)}\sqrt{\frac{\log(2p/\varepsilon)}{n}}+ \frac{b_\epsilon  \|\beta_0\|_\infty}{c'\kappa_q(s,3)},$$
where  the constant $c'>0$ depends only on  $\sigma_w$ and $\sigma_\xi$.
\end{theorem}

{} Theorem \ref{thm:pivotal} provides a bound on the $\ell_q$-rate of convergence that depends on the critical quantities of the data generating process. Indeed, it is characterized via $\kappa_q$, $m_2$, $\|t(\beta_0)\|_\infty$,$\|\beta_0\|_2$, $\|\beta_0\|_\infty$ and $b_\epsilon$ which summarizes how good the estimate $\hat\Gamma$ is. The impact of using an estimate $\hat\Gamma$ of $\Gamma$ has a factor of $ \|\beta_0\|_\infty $ instead of $\|\beta_0\|_2$.

{} The next corollary specifies the result of Theorem \ref{thm:pivotal} for the configuration of parameters of the problem usually considered in the literature. (For example, $(x_i,w_i,\xi_i)_{i=1}^n$ are independent sub-Gaussian random vectors with bounded sub-Gaussian parameters and $\Ep[x_ix_i']$ has sparse eigenvalues bounded away from zero and from above, and $s\log^3(p/\alpha)=o(n)$.) It is described by the following conditions:  $b_\epsilon \leq C\sqrt{\log (2p/\epsilon) / n}$, $m_2 \leq C$, $\|t(\beta_0)\|_\infty \leq C(1+\|\beta_0\|_2)$ and $\kappa_q(s,3) \geq cs^{-1/q}$ for $q\in [1,2]$ with high probability. Define $\Omega_X := \{ X : \kappa_q(s,3) \geq cs^{-1/q}, \max_{1\leq j\leq p}\frac{1}{n}\sum_{i=1}^n x_{ij}^4 \leq C\}$. We have the following result.

\begin{corollary}\label{cor:pivotal}
Assume that the probability that the design matrix $X$ belongs to $\Omega_X$ tends to 1 as $n\to\infty$, and that $b_\epsilon \leq C\sqrt{\log(2p/\epsilon)/n}$. Set $\varepsilon=\epsilon$. Then, under the assumptions of Theorem \ref{thm:pivotal}, for $n$ sufficiently large, with probability at least $1-\alpha-11\varepsilon-o(1)$ we have
$$ \|\hat\beta-\beta_0\|_q  \leq C(1+\|\beta_0\|_2) s^{1/q} \sqrt{\frac{\log(c'p/(\alpha\varepsilon))}{n}}$$
where $C>0, c'\ge 1$ are constants.
\end{corollary}

{} Note that Corollary \ref{cor:pivotal} exhibits the minimax rate of convergence discussed in (\ref{eq:minimaxbound}). A major point is that the estimator $\hat\beta$ achieves the minimax rate without needing to know $\|\beta_0\|_2$, $\sigma_\xi$, $\sigma_w$ or $s$ (or invoking cross-validation) as required for the procedures previously suggested  in the literature. Note that cross-validation in the problem that we consider here remains unjustified theoretically.

We note that condition (\ref{cond:kappaq}) essentially requires $s^2\log p = o(n)$. Although this is a stronger condition for consistency than $s\log p = o(n)$ in traditional linear models, such regime is of substantial interest as the condition $s^2\log p = o(n)$ appears in some works where the procedure is not pivotal and is precisely the requirement used in recent papers to construct confidence intervals post model selection via de-biasing or orthogonal moment conditions, see for example discussions and simulations in \cite{RZ2015}. Thus the proposed estimator is well poised for such applications. In cases which condition (\ref{cond:kappaq}) is not satisfied, the remark below clarifies that we can achieve a rate of convergence that does not match the minimax rate of convergence discussed in (\ref{eq:minimaxbound}).

\begin{remark} If the condition (\ref{cond:kappaq}) on $\kappa_q(s,3)$ required in Theorem \ref{thm:pivotal} (and Corollary \ref{cor:pivotal}) is not satisfied, following the same argument as in \cite{BRT2014}, we can derive a rate of convergence that depends on $\|\beta_0\|_1$ instead of $\|\beta_0\|_2$. Specifically the result in Corollary \ref{cor:pivotal} would be
$$ \|\hat\beta-\beta_0\|_q  \leq \frac{C(1+\|\beta_0\|_1)}{\kappa_q(s,3)}\Big(\sqrt{\frac{\log(c'p)}{n}}+b_\epsilon\Big).$$
\end{remark}

\begin{remark}[Approximately Sparse Models]
The results obtained in the preceding Theorem also extend to approximate sparse models where $\beta_{0T^c}$ is non-zero but small. We refer to the appendix for the modification of the proof to cover approximately sparse models.
\end{remark}

%
%\begin{remark}
%Inspection of the proof reveals that if we can enforce the non-convex constrain $\sqrt{\frac{1}{n}\sum_{i=1}^n\{z_{ij}(y_i-z_i^T\beta)\}^2}=t_j$ so that $t(\hat\beta)=\hat t$ by construction, the requirement $\kappa_q(s,3)s^{1/q} \geq 8\mu_1 s$ is relaxed to $\kappa_2(s,3)\mu_1 <2^{-1}$
%\end{remark}

{} Next, we consider the data-driven thresholded estimator $\hat\beta_{\widehat T}$ based on $\widehat T$ defined in (\ref{def:hatT}),
and we show that it achieves the sparsity $O(s)$, while preserving the optimal $\ell_1$ and $\ell_2$ rates of convergence.

\begin{theorem}\label{thm:trim}
Let $q\in \{1,2\}$. Suppose that  $s^{1/q}\sqrt{\log(2pn/\alpha)/n} = o(1)$.% and $\vert\log(\alpha p^{-1}/\varepsilon)\vert \leq C$.
 Further\-more, assume that there exist constants $0<c<C<\infty$ such that, for any $1\leq j\leq p$,
$$c(1+\|\beta_0\|_2)^2 \leq \frac{1}{n}\sum_{i=1}^n\Ep[\{z_{ij}(\xi_i-w_i^T\beta_0)+\Gamma_{jj}\beta_{0j}\}^2] \leq C(1+\|\beta_0\|_2)^2.$$ Then, under the assumptions of Corollary \ref{cor:pivotal}, for $n$ sufficiently large, with probability at least $1-\alpha-11\varepsilon-o(1)$ we have
$$\|\hat\beta_{\widehat T}\|_0 \leq Cs \ \ \mbox{and} \ \  \|\hat\beta_{\widehat T}-\beta_0\|_q  \leq C (1+\|\beta_0\|_2)s^{1/q}\sqrt{\frac{\log(c'p/\alpha)}{n}}$$
where $C>0, c'\ge 1$ are constants.
\end{theorem}

{} Theorem \ref{thm:trim} shows that the estimator $\hat\beta_{\widehat T}$ inherits the $\ell_1$ and $\ell_2$ rates of convergence of $\hat\beta$ and is also sparse. Estimators with this additional sparsity property have been useful in many settings, see for example \cite{BCKRT2016b}.

\subsection{Refitted estimators}

Next we propose and analyze refitted estimators. The goal is still to derive estimators with good $\ell_q$-rates of convergence but to reduce the bias. Indeed the $\ell_1$-regularized estimator introduced earlier yields model selection and shrinkage. We will discuss two estimators based on a (possibly data-driven) selected support. We note that in order to still achieve a convex formulation of the problem we will solve a Dantzig selector version of the problem. Letting $\hat T$ denote a (potentially data driven) support.

The first estimator we propose parallels in spirit the so-called relaxed Lasso for linear models \cite{meinshausen2007relaxed}. However we further leverage the Dantzig selection and pivotal properties we derived so far. Let  $(\tilde\beta,\tilde t,\tilde u)$ be a solution of the constrained minimization  problem
\begin{eqnarray}\label{est:pivotal-post-relaxed}
& \displaystyle \min_{\beta\in \mathbb{R}^p,t\in \mathbb{R}^p,u\in \mathbb{R}^p}   \|\beta_{\widehat T^c}\|_1 + \lambda_t \|t\|_\infty + \lambda_u \|u\|_\infty :  \\
& \begin{array}{rl}
& \left| \frac{1}{n}\sum_{i=1}^n z_{ij} (y_i - z_i^T\beta) + \hat\Gamma_{jj}\beta_j \right| \leq \tau t_j + (1+\tau)b_\epsilon u_j,
 \\
& \left\{ \frac{1}{n}\sum_{i=1}^n \{ z_{ij} (y_i - z_i^T\beta)+ \hat\Gamma_{jj}\beta_j\}^2 \right\}^{1/2} \leq t_j, \ \  |\beta_j| \leq u_j,
\end{array},  \quad \forall j \leq p,
\end{eqnarray}

The difference relative to (\ref{est:pivotal}) pertains to the $\ell_1$-regularization term that does not include components in $\widehat T$ excluding the regularization bias from these components. In principle might still be desirable to use some components in $\widehat T^c$ but with regularization bias. The following result summarizes the performance.

\begin{theorem}\label{thm:pivotal-post-relaxed} Let $\hat k = |\widehat T|$, $0<\alpha<1$, $0<\varepsilon<1$, and $1\le q\le \infty$.
Set $\tau = n^{-1/2}\Phi^{-1}(1-\alpha/(2p))$, $\lambda_u = 1/4$ and $\lambda_t = 1/\{4H_n\}$. Assume that \begin{equation}\label{cond:kappaq}\kappa_q(\hat k + s,3)s^{1/q} \geq 8 s \{ (1+\tau)b_\epsilon + \tau h_\epsilon + C'(1+m_2^{1/2})\sqrt{\log(2p^2/\varepsilon)/n}\}\end{equation} where $C'>0$ is a constant that depends only on $\sigma_w$ and $\sigma_\xi$. Then, under Conditions A and B, for $n$ sufficiently large, with probability at least $1-\alpha\{1+o(1)\}-2\epsilon-9\varepsilon$ we have
$$ \|\tilde\beta-\beta_0\|_q \leq \frac{\tau\|t(\beta_0)\|_\infty}{c'\kappa_q(\hat k+s,3)}  + \frac{(1+\|\beta_0\|_2)(1+m_2^{1/2})}{c'\kappa_q(\hat k+s,3)}\sqrt{\frac{\log(2p/\varepsilon)}{n}}+ \frac{b_\epsilon  \|\beta_0\|_\infty}{c'\kappa_q(\hat k+s,3)},$$
where  the constant $c'>0$ depends only on  $\sigma_w$ and $\sigma_\xi$.
\end{theorem}

Theorem \ref{thm:pivotal-post-relaxed} shows that the estimator (\ref{est:pivotal-post-relaxed}) achieves good rates of convergence under similar conditions as the self-normalized conic estimator (\ref{est:pivotal}) provided the size of the support $\widehat T$ is not too large so that the sensitivity quantities $\kappa_q(\hat k + s, 3)$ are still well behaved.

Next we consider the estimator that refits restricted to $\widehat T$ defined as  
\begin{equation}\label{def:est-post}
\begin{array}{rl}
(\tilde \beta,\tilde t) \in {\displaystyle \arg\min_{\beta, t \in \mathbb{R}^{p}}}  & \|t\|_\infty\\
& \left| \frac{1}{n}\sum_{i=1}^n z_{ij} (y_i - z_i^T\beta) + \hat\Gamma_{jj}\beta_j \right| \leq t_j \ \ j\in \widehat T\\
& \beta_j = t_j = 0 \ \ j \not\in \widehat T\\
\end{array}
\end{equation}
 A key difference in this case is the likely misspecification in (\ref{def:est-post}) since $\widehat T$ might miss some of the components in the support of $\beta_0$. In turn the key argument (that $\beta_0$ is feasible in the optimization problem for (\ref{est:pivotal}) and (\ref{est:pivotal-post-relaxed}) fails.
The estimator proposed in (\ref{def:est-post}) aims to reduce the regularization bias within the selected components. When the matrix $\frac{1}{n}\sum_{i=1}^nz_{i\widehat T}z_{i\widehat T}^T - \widehat \Gamma_{\widehat T,\widehat T}$ is full rank, the optimal solution has $\tilde t=0$. However, aligned with the literature, our analysis derives results based on the properties of the maximum and minimum sparse eigenvalues of $\frac{1}{n}\sum_{i=1}^nx_{i\widehat T}x_{i\widehat T}^T$. In what follows we use
$$\phi_{\min}(k) = \min_{J:|J|\leq k} \min_{\Delta \in C_J(0),\|\Delta\|_2=1} \Delta^T\Psi\Delta, \ \ \ \phi_{\max}(k) = \min_{J:|J|\leq k} \max_{\Delta \in C_J(0),\|\Delta\|_2=1} \Delta^T\Psi\Delta.$$

\begin{theorem}\label{thm:post-est}
Let $\hat k = |\widehat T|$. Under Conditions A and B, suppose that
$$ b_\epsilon+ 2\{\sigma_w m_2^{1/2}+C_{\xi w}  \}\sqrt{\frac{2\log(2p^2/\varepsilon)}{n}}  \leq \frac{\phi_{\min}(\hat k+s)}{2\{\hat k+s\}}.$$ Then we have with probability at least $1-2\epsilon-9\varepsilon$
%$$\sqrt{(\tilde \beta - \beta_{0})^T\frac{1}{n}X^TX(\tilde \beta - \beta_{0})} \leq \frac{\sqrt{|\widehat T|+s}}{\phi_{\min}^{1/2}(|\widehat T|+s)}\sqrt{\log(p/\alpha)/n} + \|\hat\beta-\beta_0\|\phi^{1/2}_{\max}(s)$$
%$$\|\tilde \beta - \beta_{0}\|_2 \leq \frac{\{\hat k+s\}^{1/2}}{c'\phi_{\min} (\hat k+s)}\left\{ (1+\|\beta_0\|_2)(1+ m_2^{1/2})\sqrt{\frac{\log(2p/\varepsilon)}{n}}  + b_\epsilon\|\beta_0\|_\infty \right\} + \|\beta_{0\widehat T^c}\|\frac{\phi^{1/2}_{\max}(s)}{\phi^{1/2}_{\min}(\hat k+s)}$$
$$\begin{array}{rl}\|\tilde \beta - \beta_{0}\|_2 & \leq \frac{\{\hat k+s\}^{1/2}}{c'\phi_{\min} (\hat k+s)}\left\{ (1+\|\beta_0\|_2)(1+ m_2^{1/2})\sqrt{\frac{\log(2p/\varepsilon)}{n}}  + b_\epsilon\|\beta_0\|_\infty \right\} \\
& + \|\beta_{0\widehat T^c}\|\frac{\phi^{1/2}_{\max}(s)}{\phi^{1/2}_{\min}(\hat k+s)}\end{array}$$
%$$ \|\tilde \beta - \beta_{0}\|_1 \leq \{\hat k + s\}^{1/2} \|\tilde \beta - \beta_{0}\|_2 \ \ \ \ \ \ \ \ \ \ \ \ \ \ \ \ \ \ \ \ \ \ \ \ $$
where the constant $c'>0$ depends only on $\sigma_w$ and $\sigma_\xi$.
\end{theorem}

\section{Numerical experiments}\label{sec:simu}

In this section, we investigate the numerical performance of the self-normalized conic estimator and its thresholded and refitted versions. Three separate designs are illustrated in Simulation A,  Simulation B and Simulation C to follow. The first design of Simulation A considers the case where the bias correction matrix is known. This setup considers the additive EIV design of model (\ref{model}). Simulation B considers the case of an unknown bias correction matrix, which is estimated from the data. This is achieved via the model where covariates are missing at random, for this purpose we follow the framework of \cite{RT2}. We implement all estimators proposed in this paper, namely the self-normalized estimator (SN-conic), its thresholded version (SN-conic thresholded) and its refitted versions described in (\ref{est:pivotal-post-relaxed}) and (\ref{def:est-post}), in the following these refitted estimators are referred to as SN-conic refitted V1 and SN-conic refitted V2 respectively. The performance of the proposed estimators are compared with the conic estimator (Conic) of \cite{BRT2014} and the bias corrected least squares (bias cor. L.S.) estimator of \cite{LW}. We provide additional benchmarks for performance, namely the biased and the no measurement error Lasso. The biased version is obtained with observed design variables $z_i$'s and the no measurement error versions are computed with the unobserved design variables $x_i$'s. The proposed SN-conic has its penalty parameters set to $\tau=n^{-1/2}\Phi^{-1}(1-\alpha/2p),$ $\alpha=0.05$ and $\lambda_t=1,$ $\lambda_u=0.25.$  

All simulation designs consider each combination of the sample size $n\in\{300,400\}$ and the model dimension $p\in\{10,100,400,750\}.$ The unobserved variables $x_i$'s and the model errors $x_i$ are generated as independent and Gaussian r.v.'s. More precisely, we set $\xi_i\sim N(0,\sigma_{\xi}^2),$ with $\sigma_{\xi}=1$ and $x_i\sim N(0,\Sigma)$ where $I_{p\times p}$ is an identity matrix and $\Sigma$ is a $p\times p$ matrix with elements $\Sigma_{ij}=\rho^{|i-j|},$ and $\rho=0.5.$ We consider two types of coefficients $\beta_0$: (i) $\b_0=(1,1,1,1,1,1,0,\dots,0)^T_{p\times 1}$ where the first six coefficients are set to one, (ii) $\b_0=(1,\frac{1}{2},\frac{1}{3},\frac{1}{4},\frac{1}{5},\frac{1}{10},0,\dots,0)^T_{p \times 1}.$ The first vector of coefficients illustrates the case where parameters are well separated from zero, and the second case is without such a separation. The latter typically leads to model selection mistakes with high probability. The remaining specifications of each simulation design are provided below.

{\bf Simulation A:} here we consider the data generating process (\ref{model}). We set $w_i\sim N(0,\sigma^2_{w}I_{p\times p}),$ with $\sigma_w=1.$ In this design we assume $\sigma_w$ to be known in all calculations so we set $\hat\Gamma=\sigma_w^2I_{p\times p}$ so that $b_{\epsilon}=0.$ For this simulation, the conic estimator is tuned assuming that $\sigma_{\xi}=1$ is known and setting $\mu=\tau=\sqrt{\log (p/\alpha)/n}.$ Similarly to the conic estimator, the bias cor. L.S. and the Lasso estimators are also tuned assuming that $\sigma_{\xi}=1$ and we set their regularization parameters to $ 2c n^{-1/2} \Phi^{-1}(1 -\alpha/(2p))$ where $c=1.1, \alpha=0.05$ (as suggested in \cite{BCW-SqLASSO}). The secondary tuning parameter $R$ of the bias cor. L.S. estimator is set to $R=\|\beta_0\|_1+0.5$ in accordance with the results of \cite{LW}.  

{\bf Simulation B:} here we consider the case where the error in measurements represents missing data. For this purpose we follow the framework of \cite{RT2}, i.e., we observe $(y_i,\tilde z_i,\eta_i,\, i=1,\dots,n)$ where
\begin{eqnarray}
\tilde z_{ij} =x_{ij}\eta_{ij}, \quad \eta_{ij}\,\, \textrm{i.i.d Bernoulli with parameter}\,\, 1-\pi,\nonumber
\end{eqnarray}
the r.v.'s $\xi_i,$ and $x_i$ are generated as before, and $\eta_{ij}=0$ indicates that we are missing the observation $x_{ij}.$ For numerical comparisons, we set the parameter $\pi=0.25$ for each simulated repetition. This simulation design assumes the bias correction matrix to be unknown and to be estimated as described in \cite{RT2}. Accordingly for the SN-conic estimator, we set $b_{\epsilon}=c\sqrt{\log(2p/\epsilon)/{n}},$ $c=0.25,$ $\epsilon=0.05.$ The tuning parameter of the conic, bias cor. L.S. and Lasso estimators are chosen as done for Simulation A, assuming $\sigma_{\xi}=1$ is known. 

{\bf Simulation C:} this simulation is designed to provide a comparison of the proposed pivotal SN estimators with cross validated version of the comparitive methods, i.e., estimates yielded by the conic, bias cor. L.S. and lasso when the corresponding tuning parameters are chosen via cross validation. Here the data generating process is assumed to be same as that of Simulation A. The conic estimator is tuned assuming that its two tuning parameters $\mu$ and $\tau$ are set to $\mu=\tau,$ over a equally separated grid of $20$ points over $\{0.01,...,1\}.$ The bias cor. L.S. estimator and lasso are tuned similarly, with $\lambda$ over the same grid of twenty points. The secondary tuning parameter of bias cor. L.S. is set to $R=\|\beta_0\|_1+0.5$ as before. We perform a five fold cross validation for each case.    

For performance comparison, we report the following metrics which are computed based on 100 simulated repetitions: bias ($\|\Ep[\hat\b-\b_0]\|_2$), root mean squared error (RMSE, $\Ep[\|\hat\b-\b_0\|^2]^{1/2}$), prediction bias (PRb, $\|\Ep(X(\hat\b-\b_0))\|_2/\sqrt{n}$), $\ell_2$-rate (L2, $\Ep[\|\hat\b-\b_0\|_2]$), $\ell_1$-rate (L1, $\Ep[\|\hat\b-\b\|_1]$, prediction risk (PR, $\Ep[\|X(\hat\b-\b_0)\|_2]/\sqrt{n}$), false positives (FP, expected number of misidentified zero elements of $\b_0$), true positives (TP, expected number of correctly identified non-zero elements of $\b_0$), false negatives (FN,  expected number of misidentified non-zero elements of $\b_0$), and time (average computation time in seconds {\small{(CPU: Intel Xeon@ 1.9Ghz, 128GB RAM)}}). All Computations are performed in the software R. The self normalized and conic estimators are implemented using the optimization software Mosek, an interior point methods solver, wrapped through the R package Rmosek. The lasso estimator is implemented with the R package `glmnet'.

Partial results of Simulation A, B and C are reported in Table \ref{tab:adsep300}, Table \ref{tab:missep300} and Table \ref{tab:CVsepn300} respectively. The results of all remaining cases of these simulations are provided in Table \ref{tab:adsep400} - Table \ref{tab:CVunsepn400} of the supplementary materials of this article. The proposed methods, SN-conic, SN-conic thresholded, SN-conic refitted V1 and SN-conic refitted V2 provide good results at all three levels of $p$ considered in the simulations, with the performance deteriorating slightly with increase in dimension. These numerical findings support our theoretical results regarding consistency of the proposed methods. The poor performance of biased Lasso highlights the impact of disregarding errors in variables. In terms of estimation performance, the proposed self-normalized estimators uniformly outperform the Conic and the bias cor. L.S. estimators in each of the three simulation designs considered. The only metric where Conic and the bias cor. L.S. estimators seem to be marginally superior in comparison to the SN-conic estimator is the false positive rate, this issue is however addressed by the thresholded and refitted versions. It may also be of interest to note that, while the estimation performance of Conic and the bias cor. L.S. estimators improves upon cross validating tuning parameters, it is still worse than that obtained by the proposed self normalized estimators. In context of computation times, the self normalized estimators are significantly more efficient than the cross validated conic and bias cor. l.s. estimators. Note that in practice, the computation times for the latter two will be further extended due the need to tune over the corresponding secondary tuning parameters, which has not been illustrated in our simulation setup.         

Amongst the proposed self-normalized estimators, the main difference between SN-conic thresholded is in the significantly reduced number of false positives in the thresholded version. This comes at a price of a slightly higher false negative rate especially in the case where the coefficients are not well separated from zero, see Tables \ref{tab:adunsep300}, \ref{tab:adunsep400}, \ref{tab:misunsep300} and \ref{tab:misunsep400}. Another expected outcome is the reduction in bias obtained via the refitted versions, among which SN-conic refitted V2 appears to provide a significantly lower bias than other methods.

% Please add the following required packages to your document preamble:
% \usepackage{multirow}
\begin{table}[htbp]
		\begin{tabular}{cccccccccccc}
			\toprule
			\multirow{2}{*}{$n=300$}                                                                      & \multicolumn{11}{c}{$\beta=(1,1,1,1,1,1,0,....,0)^T$}                         \\ \cmidrule{2-12}
			& {\textbf {p}}   & {\textbf {Bias}} & {\textbf {RMSE}} & {\textbf {PRb}}  & {\textbf {L2}}   & {\textbf {L1}}   & {\textbf {PR}}   & {\textbf {FP}}     & {\textbf {TP}}   & {\textbf {FN}}   & {\textbf {Time}}   \\ \midrule
			\multirow{4}{*}{\textbf{SN-conic}}                                                            & 10  & 0.38 & 0.71 & 0.07 & 0.68 & 1.42 & 0.64 & 1.42   & 6.00 & 0.00 & 2.45   \\
			& 100 & 0.49 & 0.81 & 0.07 & 0.79 & 1.77 & 0.78 & 9.46   & 5.99 & 0.01 & 6.11   \\
			& 400 & 0.58 & 0.84 & 0.09 & 0.81 & 1.92 & 0.83 & 21.52  & 5.99 & 0.01 & 18.05  \\
			& 750 & 0.61 & 0.86 & 0.09 & 0.84 & 2.03 & 0.89 & 47.95  & 5.98 & 0.02 & 123.45 \\ \midrule
			\multirow{4}{*}{\textbf{\begin{tabular}[c]{@{}c@{}}SN-conic \\ (thresholded)\end{tabular}}}   & 10  & 0.38 & 0.71 & 0.07 & 0.68 & 1.40 & 0.64 & 0.24   & 5.96 & 0.04 & 2.45   \\
			& 100 & 0.49 & 0.80 & 0.07 & 0.78 & 1.61 & 0.78 & 0.44   & 5.95 & 0.05 & 6.11   \\
			& 400 & 0.59 & 0.83 & 0.09 & 0.81 & 1.68 & 0.84 & 0.52   & 5.91 & 0.09 & 18.05  \\
			& 750 & 0.61 & 0.85 & 0.09 & 0.82 & 1.73 & 0.89 & 0.41   & 5.93 & 0.07 & 123.45 \\ \midrule
			\multirow{4}{*}{\textbf{\begin{tabular}[c]{@{}c@{}}SN-conic \\ (refitted V1)\end{tabular}}}   & 10  & 0.33 & 0.59 & 0.06 & 0.56 & 1.14 & 0.55 & 1.08   & 5.98 & 0.02 & 1.58   \\
			& 100 & 0.38 & 0.66 & 0.07 & 0.64 & 1.45 & 0.68 & 6.53   & 5.96 & 0.04 & 4.91   \\
			& 400 & 0.48 & 0.76 & 0.08 & 0.73 & 1.72 & 0.75 & 21.60  & 5.98 & 0.02 & 16.00  \\
			& 750 & 0.46 & 0.73 & 0.08 & 0.70 & 1.74 & 0.75 & 41.54  & 5.96 & 0.04 & 125.84 \\ \midrule
			\multirow{4}{*}{\textbf{\begin{tabular}[c]{@{}c@{}}SN-conic \\ (refitted V2)\end{tabular}}}   & 10  & 0.20 & 2.18 & 0.14 & 1.35 & 2.91 & 0.94 & 0.24   & 5.96 & 0.04 & 6.05   \\
			& 100 & 0.19 & 1.96 & 0.13 & 1.42 & 3.10 & 1.04 & 0.44   & 5.95 & 0.05 & 8.93   \\
			& 400 & 0.21 & 1.57 & 0.11 & 1.31 & 2.90 & 0.94 & 0.52   & 5.91 & 0.09 & 22.66  \\
			& 750 & 0.47 & 3.24 & 0.22 & 1.73 & 3.83 & 1.22 & 0.41   & 5.93 & 0.07 & 55.08  \\ \midrule
			\multirow{4}{*}{\textbf{Conic}}                                                               & 10  & 0.57 & 0.95 & 0.09 & 0.91 & 1.92 & 0.90 & 0.32   & 5.98 & 0.02 & 0.44   \\
			& 100 & 0.69 & 1.04 & 0.10 & 1.01 & 2.09 & 1.03 & 0.59   & 5.93 & 0.07 & 0.47   \\
			& 400 & 0.75 & 1.10 & 0.11 & 1.06 & 2.24 & 1.09 & 4.15   & 5.87 & 0.13 & 1.45   \\
			& 750 & 0.76 & 1.05 & 0.11 & 1.02 & 2.17 & 1.09 & 6.18   & 5.87 & 0.13 & 3.97   \\ \midrule
			\multirow{4}{*}{\textbf{Bias Cor. L.S.}}                                                      & 10  & 0.40 & 1.06 & 0.09 & 0.99 & 2.09 & 0.83 & 0.41   & 5.84 & 0.16 & 0.10   \\
			& 100 & 0.50 & 1.21 & 0.10 & 1.14 & 2.61 & 0.99 & 2.91   & 5.66 & 0.34 & 0.13   \\
			& 400 & 0.56 & 1.29 & 0.10 & 1.20 & 3.01 & 1.03 & 5.68   & 5.58 & 0.42 & 0.41   \\
			& 750 & 0.53 & 1.26 & 0.11 & 1.17 & 2.91 & 1.03 & 6.45   & 5.61 & 0.39 & 0.93   \\ \midrule
			\multirow{4}{*}{\textbf{\begin{tabular}[c]{@{}c@{}}Lasso \\ (biased)\end{tabular}}}        & 10  & 0.92 & 0.96 & 0.14 & 0.95 & 2.37 & 1.33 & 3.69   & 6.00 & 0.00 & 0.17   \\
			& 100 & 0.95 & 1.02 & 0.12 & 1.02 & 3.23 & 1.37 & 68.14  & 6.00 & 0.00 & 0.53   \\
			& 400 & 0.99 & 1.09 & 0.14 & 1.09 & 4.48 & 1.39 & 150.97 & 6.00 & 0.00 & 0.81   \\
			& 750 & 1.01 & 1.13 & 0.14 & 1.13 & 5.13 & 1.39 & 187.89 & 6.00 & 0.00 & 1.04   \\ \midrule
			\multirow{4}{*}{\textbf{\begin{tabular}[c]{@{}c@{}}Lasso \\ (no meas error)\end{tabular}}} & 10  & 0.21 & 0.28 & 0.03 & 0.28 & 0.57 & 0.32 & 3.50   & 6.00 & 0.00 & 0.14   \\
			& 100 & 0.26 & 0.32 & 0.04 & 0.32 & 0.65 & 0.38 & 6.21   & 6.00 & 0.00 & 0.39   \\
			& 400 & 0.30 & 0.35 & 0.05 & 0.35 & 0.73 & 0.43 & 8.01   & 6.00 & 0.00 & 0.61   \\
			& 750 & 0.30 & 0.37 & 0.05 & 0.36 & 0.75 & 0.44 & 11.26  & 6.00 & 0.00 & 0.93   \\ \bottomrule
		\end{tabular}
	\vspace{3mm}
	\caption{\label{tab:adsep300} Simulation A: numerical results at $n=300$ under the additive EIV setup, with a known bias correction matrix and separated regression coefficients.}
\end{table}%

\begin{table}[htbp]
		\begin{tabular}{cccccccccccc}
			\toprule
			\multirow{2}{*}{$n=300$}                                                                  & \multicolumn{11}{c}{$\beta=(1,1,1,1,1,1,0,....,0)^T$}                                                                                                         \\\cmidrule{2-12} 
			& \textbf{p} & \textbf{Bias} & \textbf{RMSE} & \textbf{PRb} & \textbf{L2} & \textbf{L1} & \textbf{PR} & \textbf{FP} & \textbf{TP} & \textbf{FN} & \textbf{Time} \\ \midrule
			\multirow{4}{*}{\textbf{SN-conic}}                                                          & 10         & 0.23          & 0.47          & 0.04         & 0.44        & 0.94        & 0.40        & 0.71        & 6.00        & 0.00        & 0.89          \\
			& 100        & 0.30          & 0.52          & 0.05         & 0.50        & 1.14        & 0.50        & 5.35        & 6.00        & 0.00        & 3.89          \\
			& 400        & 0.34          & 0.56          & 0.06         & 0.54        & 1.21        & 0.55        & 11.40       & 6.00        & 0.00        & 16.57         \\
			& 750        & 0.37          & 0.56          & 0.05         & 0.54        & 1.22        & 0.57        & 5.81        & 6.00        & 0.00        & 110.99        \\ \midrule
			\multirow{4}{*}{\textbf{\begin{tabular}[c]{@{}c@{}}SN-conic \\ (thresholded)\end{tabular}}} & 10         & 0.23          & 0.47          & 0.04         & 0.44        & 0.92        & 0.40        & 0.12        & 6.00        & 0.00        & 0.89          \\
			& 100        & 0.30          & 0.51          & 0.04         & 0.49        & 1.03        & 0.49        & 0.13        & 6.00        & 0.00        & 3.89          \\
			& 400        & 0.34          & 0.55          & 0.06         & 0.53        & 1.08        & 0.54        & 0.02        & 6.00        & 0.00        & 16.57         \\
			& 750        & 0.37          & 0.55          & 0.05         & 0.53        & 1.09        & 0.57        & 0.04        & 6.00        & 0.00        & 110.99        \\ \midrule
			\multirow{4}{*}{\textbf{\begin{tabular}[c]{@{}c@{}}SN-conic \\ (refitted V1)\end{tabular}}} & 10         & 0.16          & 0.38          & 0.04         & 0.35        & 0.72        & 0.33        & 0.41        & 6.00        & 0.00        & 0.75          \\
			& 100        & 0.20          & 0.39          & 0.04         & 0.37        & 0.79        & 0.38        & 3.01        & 6.00        & 0.00        & 3.63          \\
			& 400        & 0.22          & 0.38          & 0.04         & 0.37        & 0.82        & 0.40        & 6.65        & 6.00        & 0.00        & 16.56         \\
			& 750        & 0.24          & 0.42          & 0.04         & 0.40        & 0.89        & 0.43        & 32.62       & 6.00        & 0.00        & 113.54        \\ \midrule
			\multirow{4}{*}{\textbf{\begin{tabular}[c]{@{}c@{}}SN-conic \\ (refitted V2)\end{tabular}}} & 10         & 0.06          & 0.44          & 0.03         & 0.40        & 0.85        & 0.30        & 0.12        & 6.00        & 0.00        & 0.50          \\
			& 100        & 0.04          & 0.45          & 0.03         & 0.42        & 0.87        & 0.33        & 0.13        & 6.00        & 0.00        & 3.45          \\
			& 400        & 0.02          & 0.45          & 0.03         & 0.42        & 0.88        & 0.31        & 0.02        & 6.00        & 0.00        & 18.27         \\
			& 750        & 0.04          & 0.42          & 0.03         & 0.39        & 0.80        & 0.30        & 0.04        & 6.00        & 0.00        & 50.62         \\ \midrule
			\multirow{4}{*}{\textbf{Conic}}                                                             & 10         & 0.51          & 0.69          & 0.08         & 0.67        & 1.39        & 0.76        & 0.08        & 6.00        & 0.00        & 0.21          \\
			& 100        & 0.60          & 0.76          & 0.08         & 0.74        & 1.57        & 0.88        & 0.39        & 6.00        & 0.00        & 0.26          \\
			& 400        & 0.64          & 0.80          & 0.10         & 0.79        & 1.66        & 0.94        & 2.40        & 6.00        & 0.00        & 1.20          \\
			& 750        & 0.67          & 0.80          & 0.09         & 0.79        & 1.67        & 0.96        & 1.91        & 6.00        & 0.00        & 3.53          \\ \midrule
			\multirow{4}{*}{\textbf{Bias Cor. L.S.}}                                                    & 10         & 0.43          & 0.63          & 0.07         & 0.61        & 1.27        & 0.66        & 0.08        & 6.00        & 0.00        & 0.88          \\
			& 100        & 0.54          & 0.71          & 0.07         & 0.70        & 1.47        & 0.81        & 0.05        & 6.00        & 0.00        & 3.41          \\
			& 400        & 0.59          & 0.77          & 0.09         & 0.75        & 1.57        & 0.88        & 0.01        & 6.00        & 0.00        & 12.27         \\
			& 750        & 0.62          & 0.77          & 0.09         & 0.76        & 1.59        & 0.90        & 0.04        & 6.00        & 0.00        & 22.73         \\ \midrule
			\multirow{4}{*}{\textbf{\begin{tabular}[c]{@{}c@{}}Lasso \\ (biased)\end{tabular}}}         & 10         & 0.52          & 0.58          & 0.08         & 0.57        & 1.31        & 0.74        & 3.60        & 6.00        & 0.00        & 0.04          \\
			& 100        & 0.56          & 0.62          & 0.08         & 0.62        & 1.52        & 0.81        & 34.06       & 6.00        & 0.00        & 0.24          \\
			& 400        & 0.58          & 0.65          & 0.09         & 0.65        & 1.72        & 0.84        & 61.87       & 6.00        & 0.00        & 0.57          \\
			& 750        & 0.60          & 0.66          & 0.08         & 0.66        & 1.85        & 0.85        & 84.13       & 6.00        & 0.00        & 0.97          \\ \midrule
			\multirow{4}{*}{\textbf{\begin{tabular}[c]{@{}c@{}}Lasso \\ (no meas error)\end{tabular}}}  & 10         & 0.21          & 0.28          & 0.04         & 0.27        & 0.56        & 0.32        & 3.39        & 6.00        & 0.00        & 0.05          \\
			& 100        & 0.27          & 0.33          & 0.04         & 0.33        & 0.68        & 0.40        & 5.72        & 6.00        & 0.00        & 0.21          \\
			& 400        & 0.29          & 0.35          & 0.05         & 0.34        & 0.72        & 0.43        & 8.26        & 6.00        & 0.00        & 0.45          \\
			& 750        & 0.31          & 0.37          & 0.04         & 0.36        & 0.76        & 0.45        & 11.00       & 6.00        & 0.00        & 0.83          \\ \bottomrule
		\end{tabular}
	\vspace{3mm}
	\caption{\label{tab:missep300} Simulation B: numerical results at $n=300$ under the covariates missing at random setup, with an estimated bias correction and separated regression coefficients.}
\end{table}%

\begin{table}[htbp]
		\begin{tabular}{clllllllllll}
		\toprule
			\multirow{2}{*}{$n=300$}                                                                & \multicolumn{11}{c}{$\beta=(1,1,1,1,1,1,0,...,0)^T$}                                                                                                                                                                                                                                                                                                                                     \\ \cmidrule{2-12} 
			& \multicolumn{1}{c}{\textbf{p}} & \multicolumn{1}{c}{\textbf{Bias}} & \multicolumn{1}{c}{\textbf{RMSE}} & \multicolumn{1}{c}{\textbf{PRb}} & \multicolumn{1}{c}{\textbf{L2}} & \multicolumn{1}{c}{\textbf{L1}} & \multicolumn{1}{c}{\textbf{PR}} & \multicolumn{1}{c}{\textbf{FP}} & \multicolumn{1}{c}{\textbf{TP}} & \multicolumn{1}{c}{\textbf{FN}} & \multicolumn{1}{c}{\textbf{Time}} \\ \midrule
			\multirow{4}{*}{\textbf{Bias Cor. L.S.}}                                                & 10                             & 0.28                              & 1.18                              & 0.08                             & 1.08                            & 2.41                            & 0.83                            & 1.24                            & 5.80                            & 0.20                            & 4.54                              \\
			& 100                            & 0.39                              & 1.24                              & 0.10                             & 1.18                            & 2.97                            & 0.98                            & 5.35                            & 5.70                            & 0.30                            & 18.84                             \\
			& 400                            & 0.57                              & 1.39                              & 0.11                             & 1.31                            & 3.44                            & 1.10                            & 7.59                            & 5.51                            & 0.49                            & 84.76                             \\
			& 750                            & 0.63                              & 1.48                              & 0.12                             & 1.41                            & 3.66                            & 1.20                            & 7.98                            & 5.39                            & 0.61                            & 223.66                            \\ \midrule
			\multirow{4}{*}{\textbf{Conic}}                                                         & 10                             & 0.36                              & 1.04                              & 0.08                             & 0.87                            & 1.87                            & 0.72                            & 1.12                            & 6.00                            & 0.00                            & 11.18                             \\
			& 100                            & 0.40                              & 0.91                              & 0.07                             & 0.87                            & 2.47                            & 0.80                            & 10.44                           & 5.99                            & 0.01                            & 45.59                             \\
			& 400                            & 0.52                              & 0.88                              & 0.08                             & 0.86                            & 2.43                            & 0.85                            & 20.68                           & 5.98                            & 0.02                            & 313.35                            \\
			& 750                            & 0.58                              & 0.92                              & 0.09                             & 0.90                            & 2.54                            & 0.90                            & 21.31                           & 5.98                            & 0.02                            & 852.59                            \\ \midrule
			\multirow{4}{*}{\textbf{\begin{tabular}[c]{@{}c@{}}Lasso \\ (biased)\end{tabular}}}     & 10                             & 1.10                              & 1.13                              & 0.15                             & 1.13                            & 2.70                            & 1.62                            & 0.74                            & 6.00                            & 0.00                            & 0.14                              \\
			& 100                            & 1.19                              & 1.22                              & 0.16                             & 1.22                            & 2.95                            & 1.76                            & 1.49                            & 6.00                            & 0.00                            & 0.32                              \\
			& 400                            & 1.23                              & 1.26                              & 0.18                             & 1.25                            & 3.04                            & 1.83                            & 1.43                            & 6.00                            & 0.00                            & 1.04                              \\
			& 750                            & 1.27                              & 1.30                              & 0.19                             & 1.30                            & 3.15                            & 1.88                            & 1.95                            & 5.99                            & 0.01                            & 1.17                              \\ \midrule
			\multirow{4}{*}{\textbf{\begin{tabular}[c]{@{}c@{}}Lasso\\ (no meas err)\end{tabular}}} & 10                             & 0.19                              & 0.27                              & 0.03                             & 0.26                            & 0.53                            & 0.29                            & 0.19                            & 6.00                            & 0.00                            & 0.12                              \\
			& 100                            & 0.23                              & 0.29                              & 0.03                             & 0.29                            & 0.60                            & 0.34                            & 0.36                            & 6.00                            & 0.00                            & 0.28                              \\
			& 400                            & 0.25                              & 0.32                              & 0.04                             & 0.32                            & 0.68                            & 0.38                            & 0.68                            & 6.00                            & 0.00                            & 0.62                              \\
			& 750                            & 0.28                              & 0.33                              & 0.04                             & 0.33                            & 0.70                            & 0.40                            & 0.54                            & 6.00                            & 0.00                            & 0.81                              \\ \bottomrule
		\end{tabular}
	\vspace{3mm}
	\caption{\label{tab:CVsepn300} Simulation C: numerical results at $n=300$ for comparative methods tuned via cross validation, under the additive EIV setup with separated regression coefficients.}%
\end{table}%

\section{Conclusion}

In this paper, we have introduced a new estimator for the high-dimensional linear regression model with measurement error in the design. Compared to previous literature, the main interest of our procedure is that it is pivotal, in the sense that it does not require the knowledge of typically unknown quantities (for example variance of the noise in the signal, sparsity of the vector of interest or its $l_1$ norm) to tune regularization parameters. Furthermore it still achieves optimal rates of convergence and is valid for non-i.i.d. data. Note also that we obtain a more adaptive regularization as it depends on the empirical second moment of the moment conditions, in opposition to \cite{BRT2014} that had to use constants and triangle inequalities to derive feasible choices of constants. This estimator is solution of a convex program with second order cone constraints, allowing for easy implementation and fast computations. We also consider thresholded and refitted versions of the estimator enabling us to obtain respectively sparsity guarantees and bias reduction in the estimated coefficients. An extensive simulation study confirms good performance of our estimators in various settings.

\newpage

\begin{appendix}
\section{Main proofs}

We set $T=\supp(\beta_0)$ with $|T| \leq s$. We begin by stating a technical lemma.

\begin{lemma}\label{lemma:auxPhi}
For $a\geq 1$ and $\gamma>0$, we have
$$1-\Phi(a/\{1+\gamma\}) \leq \{1-\Phi(a)\}\exp(2a^2\gamma).$$\end{lemma}

{} The next lemma deals with $(\hat\beta,\hat t,\hat u)$ defined as the solution to (\ref{est:pivotal})-(\ref{est:pivotal_1}).

\begin{lemma}\label{lem:pivotal}
Under Conditions A and B, for $\tau = n^{-1/2}\Phi^{-1}(1-\alpha/(2p))$, we have that the triple $(\beta_0, t(\beta_0),|\beta_0|)$ is feasible in the pivotal conic programming (\ref{est:pivotal}) with probability at least $1-\alpha\{1+o(1)\}-\epsilon$ for $n$ sufficiently large. On that event, we have
 $$ \|\hat t \|_\infty - \|t(\beta_0)\|_\infty \leq \frac{1+\lambda_u}{\lambda_t}\|\hat\beta-\beta_0\|_1, \ \ \mbox{and} $$$$\|\hat u \|_\infty - \|\beta_0\|_\infty \leq \frac{1}{\lambda_u}\|\hat\beta-\beta_0\|_1+\frac{\lambda_t}{\lambda_u}H_n\|\hat\beta-\beta_0\|_1.$$
In addition, using  $\lambda_u = 1/4$ and $\lambda_t = 1/\{4H_n\}$,  we have $\hat \beta-\beta_0 \in C_T(3)$.
\end{lemma}
\begin{proof}[Proof of Lemma \ref{lem:pivotal}]
Let $\Gamma_{i,jj}=\Ep[w_{ij}^2]$, $\Gamma_{jj}=\frac{1}{n}\sum_{i=1}^n\Ep[w_{ij}^2]$, and $|\beta_0|=(|\beta_{0j}|)_{j=1}^p$. Recall that we defined $t_j(\beta) = \{\frac{1}{n}\sum_{i=1}^n \{z_{ij} (y_i - z_i^T\beta) + \hat\Gamma_{j\cdot}\beta\}^2\}^{1/2}$ and $U_{ij}=z_{ij}(y_i-z_i^T\beta_0)+\Gamma_{i,jj}\beta_{0j}$. For $j=1,\ldots,p$, $i=1,\ldots,n$, we set $\bar U_{ij} = z_{ij}(y_i-z_i^T\beta_0)+\Gamma_{jj}\beta_{0j}$ (note that the optimization problem has $\hat\Gamma_{jj}$ instead of $\Gamma_{i,jj}$ or $\Gamma_{jj}$). Remark also that by definition $\sum_{i=1}^nU_{ij}=\sum_{i=1}^n\bar U_{ij}$. We now show that the triplet $(\beta_0, t(\beta_0),|\beta_0|)$ is feasible with high probability. Indeed, the probability that the triplet $(\beta_0, t(\beta_0),|\beta_0|)$ violates any constraint satisfies
{\small $$ \begin{array}{ll}
\displaystyle \P\Big(\exists j: \Big| \frac{1}{n}\sum_{i=1}^n z_{ij}(y_i-z_i^T\beta_0)+\hat \Gamma_{jj}\beta_{0j} \Big|
> \tau t_j(\beta_0) + (1+\tau)b_\epsilon|\beta_{0j}|\Big) \\
=\displaystyle \P\Big(\exists j: \Big| \frac{1}{n}\sum_{i=1}^n \bar U_{ij} + (\hat\Gamma_{jj}-\Gamma_{jj})\beta_{0j} \Big|
> \tau t_j(\beta_0) + (1+\tau)b_\epsilon|\beta_{0j}|\Big) \\
\displaystyle \leq \P\Big(\exists j: \Big| \frac{1}{n}\sum_{i=1}^n U_{ij}+ (\hat\Gamma_{jj}-\Gamma_{jj})\beta_{0j}\Big|
> \tau \Big\|\frac{\bar U_{\cdot j}}{\sqrt{n}}\Big\|_2 - \tau|(\hat\Gamma_{jj}-\Gamma_{jj})\beta_{0j}|+ (1+\tau) b_\epsilon|\beta_{0j}|\Big) \\
\leq \displaystyle \P\Big(\exists j: \Big| \frac{1}{n}\sum_{i=1}^n U_{ij} \Big|
> \tau \Big\|\frac{\bar U_{\cdot j}}{\sqrt{n}}\Big\|_2\Big) + \epsilon \\
\leq \displaystyle \P\Big(\exists j: \Big| \frac{1}{n}\sum_{i=1}^n U_{ij} \Big|
> \frac{\tau}{1+n^{-1/3}} \Big\|\frac{ U_{\cdot j}}{\sqrt{n}}\Big\|_2\Big) + \P\Big(\exists j : (1+n^{-1/3})\Big\|\frac{ \bar U_{\cdot j}}{\sqrt{n}}\Big\|_2 \leq \Big\|\frac{ U_{\cdot j}}{\sqrt{n}}\Big\|_2\Big)+\epsilon,\end{array}$$}\\
where we used that $ \|\bar U_{\cdot j}/\sqrt{n}\|_2 -  t_j(\beta_0)| \leq |\hat\Gamma_{jj}-\Gamma_{jj}|\cdot  |\beta_{0j}|$ by definition  and $\max_{j=1,\ldots,p}|\hat\Gamma_{jj}-\Gamma_{jj}|\leq b_\epsilon$ with probability $1-\epsilon$ from Condition B.

{} We now bound each term in the last display separately. By Condition A, note that $U_{ij}$ is a zero-mean random variable. %By Condition D, we have $\Ep[|U_{ij}|^3]^{ \leq C(1+ \|\beta_0\|_2^3 )$ and $\Ep[U_{ij}^2]\geq c(1+\|\beta_0\|_2^2)$ uniformly in $j=1,\ldots,p$.
Therefore, applying Lemma 7.4 in \cite{delapena} together with Condition B that implies $$\sqrt{n}\tau \max_{1\leq j\leq p}\Big\{\Big(\frac{1}{n}\sum_{i=1}^n\Ep[|U_{ij}|^3]\Big)^{1/3}\Big/\Big(\frac{1}{n}\sum_{i=1}^n\Ep[|U_{ij}|^2]\Big)^{1/2}\Big\} \leq  n^{1/6}/\ell_n,$$
where $\ell_n\to\infty$,  we have
{\footnotesize \begin{equation}\label{auxTerm1} \begin{array}{ll}
&  \displaystyle \P\Big( \Big| \frac{1}{\sqrt{n}}\sum_{i=1}^n U_{ij} \Big|
> \frac{\sqrt{n}\tau}{(1+n^{-1/3})} \Big\|\frac{U_{\cdot j}}{\sqrt{n}}\Big\|_2  \Big)  \leq \{1-\Phi(\sqrt{n}\tau/\{1+n^{-1/3}\})\}\big( 1 + \frac{A}{\ell_n^3} \big)\\
& \leq \{1-\Phi(\sqrt{n}\tau)\}\exp\left( 2n^{-1/3} \{\Phi^{-1}(1-\frac{\alpha}{2p})\}^2 \right)\left( 1 + \frac{A}{\ell_n^3} \right)\\
& \leq \{1-\Phi(\sqrt{n}\tau)\}\exp\left( 2/\ell_n^2 \right)\left( 1 + \frac{A}{\ell_n^3} \right) =\frac{\alpha}{2p}\exp\left( 2/\ell_n^2 \right)\big( 1 + \frac{A}{\ell_n^3} \big), \end{array}\end{equation}}
for some universal constant $A>0$. Here, we have used Lemma \ref{lemma:auxPhi}, and again Condition B that implies $n^{-1/3}\{\Phi^{-1}(1-\frac{\alpha}{2p})\}^2\leq 1/\ell_n^2$.

{} To bound the last term, set  $\Sigma_j^2:= \frac{1}{n}\sum_{i=1}^n(\Gamma_{i,jj}-\Gamma_{jj})^2$, $j=1,\ldots,p$. Note that
$$\P\Big(\exists j : (1+n^{-1/3})\Big\|\frac{ \bar U_{\cdot j}}{\sqrt{n}}\Big\|_2 \leq \Big\|\frac{ U_{\cdot j}}{\sqrt{n}}\Big\|_2\Big)$$
is smaller than
$$\P\Big(\exists j : n^{-1/3}\Big\|\frac{  U_{\cdot j}}{\sqrt{n}}\Big\|_2^2 + (1+n^{-1/3})^2\Sigma_j^2\beta_{0j}^2 + \frac{2(1+n^{-1/3})^2\beta_{0j}}{n}\sum_{i=1}^n U_{ij}(\Gamma_{jj}-\Gamma_{i,jj}) < 0 \Big).$$
Since $0 \leq \Gamma_{i,jj} \leq C$ by Condition A, we deduce that $\{U_{ij}(\Gamma_{i,jj}-\Gamma_{jj}):i=1,\ldots,n\}$ satisfies the moderate deviation condition for self-normalized sums since $\{U_{ij}:i=1,\ldots,n\}$ satisfies it by Condition B. Therefore we get
{\small $$ \begin{array}{ll}
&  \displaystyle \P\Big(\exists j :  \Big|\frac{1}{n}\sum_{i=1}^n U_{ij}(\Gamma_{i,jj}-\Gamma_{jj})\Big| >\frac{\Phi^{-1}(1-\frac{\alpha}{2pn})}{n^{1/2}}\Big\{ \frac{1}{n}\sum_{i=1}^n U_{ij}^2(\Gamma_{i,jj}-\Gamma_{jj})^2\Big\}^{1/2}  \Big) \\
\\
& \displaystyle \leq \frac{\alpha}{n}(1+A/\ell_n^3).
 \end{array}$$}
Note that if $\Big\{ \frac{1}{n}\sum_{i=1}^n U_{ij}^2(\Gamma_{i,jj}-\Gamma_{jj})^2\Big\}^{1/2}=0$ the result is trivial. Furthermore we have
$$ \Big\{ \frac{1}{n}\sum_{i=1}^n U_{ij}^2(\Gamma_{i,jj}-\Gamma_{jj})^2\Big\}^{1/2} \leq \Big\{ \frac{1}{n}\sum_{i=1}^n |U_{ij}|^2\Big\}^{1/2}\max_{1\leq i\leq n}|\Gamma_{i,jj}-\Gamma_{jj}|$$
and by Condition B, for all $j$ such that $\Delta_j := |\beta_{0j}| \max_{1\leq i\leq n}|\Gamma_{i,jj}-\Gamma_{jj}|>0$,  $$\Delta_j \Phi^{-1}(1-\alpha/(2pn)) \leq \frac{n^{1/6}}{\ell_n} \Big\{\frac{1}{n}\sum_{i=1}^n \Ep[U_{ij}^2]\Big\}^{1/2}.$$
Hence,
{\footnotesize \begin{equation}\label{auxTerm2} \begin{array}{ll}
&  \displaystyle \P\Big(\exists j : (1+n^{-1/3})\Big\|\frac{ \bar U_{\cdot j}}{\sqrt{n}}\Big\|_2 \leq \Big\|\frac{ U_{\cdot j}}{\sqrt{n}}\Big\|_2\Big) \\
& \displaystyle \leq \P\Big(\exists j : n^{-1/3}\Big\|\frac{  U_{\cdot j}}{\sqrt{n}}\Big\|_2^2 < 2(1+n^{-1/3})^2\frac{\Phi^{-1}(1-\frac{\alpha}{2pn})}{n^{1/2}}\Big\|\frac{ U_{\cdot j}}{\sqrt{n}}\Big\|_2 \Delta_j  \Big) +\frac{\alpha}{n}\{1+A/\ell_n^3\}\\
& \displaystyle \leq \P\Big(\exists j: \Delta_j>0 \ \mbox{and} \  \Big\|\frac{  U_{\cdot j}}{\sqrt{n}}\Big\|_2 < \frac{1}{\ell_n^2} \Big\{ \frac{1}{n}\sum_{i=1}^n \Ep[U_{ij}^2]\Big\}^{1/2} \Big) +\frac{\alpha}{n}\{1+A/\ell_n^3\}.
 \end{array}\end{equation}}

{} Next we bound the first term of the RHS of (\ref{auxTerm2}). For any $j$ such that $\Delta_j> 0 $:
$$\begin{array}{rl}
 \P( \frac{1}{n}\sum_{i=1}^n U_{ij}^2 < \frac{1}{\ell_n^4}\frac{1}{n}\sum_{i=1}^n\Ep[U_{ij}^2] )& \leq \exp( -\frac{1}{2}(1-\ell_n^{-4})^2n) \\
 & \leq \exp( -\frac{1}{4}n)\\
 & \leq \exp(-\log(2pn/\alpha)) \\
 &\leq \frac{\alpha}{2pn}
 \end{array}$$ where we applied  Lemma \ref{lemma:LBsubG} with $t=(3/4)\sum_{i=1}^n\Ep[U_{ij}^2]$, and used Condition B(iii).

By the union bound, this implies \begin{equation}\label{auxBound3}\P\Big(\exists j : \Delta_j > 0 \ \mbox{and} \ \Big\|\frac{  U_{\cdot j}}{\sqrt{n}}\Big\|_2 < \frac{1}{\ell_n^2} \Big\{ \frac{1}{n}\sum_{i=1}^n \Ep[U_{ij}^2]\Big\}^{1/2} \Big)\leq \frac{\alpha}{2n}\end{equation}

{} Combining (\ref{auxTerm1}), (\ref{auxTerm2}),  (\ref{auxBound3}), and using the convergence $\ell_n \to \infty$ and the union bound,  we find that the triplet $(\beta_0,t(\beta_0),|\beta_0|)$ is feasible with probability at least $1-\alpha\{1+o(1)\}-\epsilon$.

{} If the triplet $(\beta_0,t(\beta_0),|\beta_0|)$ is feasible for the problem above, it follows that
  \begin{equation}\label{opt:pivotal}\|\hat \beta\|_1 + \lambda_t \|\hat t\|_\infty + \lambda_u \|\hat u\|_\infty \leq  \|\beta_0\|_1 + \lambda_t \|t(\beta_0)\|_\infty + \lambda_u \|\beta_0\|_\infty.\end{equation}
By (\ref{opt:pivotal}), and the inequalities $\|t(\hat\beta)\|_\infty \leq \|\hat t\|_\infty$ and $\|\hat\beta\|_\infty \leq \|\hat u\|_\infty$ from the definition of the estimator, we have that
$$\begin{array}{l}
 \|\hat t \|_\infty - \|t(\beta_0)\|_\infty \leq \frac{1+\lambda_u}{\lambda_t}\|\hat\beta-\beta_0\|_1,\\
 \|\hat u \|_\infty - \|\beta_0\|_\infty \leq \frac{1}{\lambda_u}\|\beta_0-\hat\beta\|_1+\frac{\lambda_t}{\lambda_u}\{\|t(\beta_0)\|_\infty-\|t(\hat\beta)\|_\infty\}.
 \end{array}$$
Next, since $|\{\frac{1}{n}\sum_{i=1}^na_i^2\}^{1/2} -\{\frac{1}{n}\sum_{i=1}^nb_i^2\}^{1/2}|\leq \{\frac{1}{n}\sum_{i=1}^n(a_i-b_i)^2\}^{1/2}$, we have
%$$
%\begin{array}{rl}
%|t_j(\beta_0)-t_j(\hat\beta)|^2 & \leq \frac{1}{n}\sum_{i=1}^n\{ (z_{ij}z_i^T-\hat\Gamma_{j\cdot}) (\hat\beta-\beta_0)\}^2\\
%& = (\hat\beta-\beta_0)^T \frac{1}{n}\sum_{i=1}^n (z_{ij}z_i^T-\hat\Gamma_{j\cdot})(z_{ij}z_i^T-\hat\Gamma_{j\cdot})^T (\hat\beta-\beta_0)\\
%& \leq \|\hat\beta-\beta_0\|_1^2 \| \frac{1}{n}\sum_{i=1}^n (z_{ij}z_i^T-\hat\Gamma_{j\cdot})(z_{ij}z_i^T-\hat\Gamma_{j\cdot})^T \|_\infty \end{array}$$
%Thus for $H_n= \max_{j=1,\ldots,p}\| \frac{1}{n}\sum_{i=1}^n (z_{ij}z_i^T-\hat\Gamma_{j\cdot})(z_{ij}z_i^T-\hat\Gamma_{j\cdot})^T \|_\infty^{1/2}$, we have
%\begin{equation}\label{bound:t}\|t(\beta_0)-t(\hat\beta)\|_\infty \leq H_n \|\hat\beta-\beta_0\|_1 \end{equation}
%and  the result on $\|\hat u \|_\infty$ follows by (\ref{bound:t}).
$$
\begin{array}{rl}
|t_j(\beta_0)-t_j(\hat\beta)|^2 & \leq \frac{1}{n}\sum_{i=1}^n\{ (z_{ij}z_i^T-\hat\Gamma_{j\cdot}) (\hat\beta-\beta_0)\}^2\\
& = (\hat\beta-\beta_0)^T \frac{1}{n}\sum_{i=1}^n (z_{ij}z_i^T-\hat\Gamma_{j\cdot})^T(z_{ij}z_i^T-\hat\Gamma_{j\cdot}) (\hat\beta-\beta_0)\\
& \leq \|\hat\beta-\beta_0\|_1^2 \| \frac{1}{n}\sum_{i=1}^n (z_{ij}z_i^T-\hat\Gamma_{j\cdot})^T(z_{ij}z_i^T-\hat\Gamma_{j\cdot}) \|_\infty \end{array}.$$
Thus for $H_n= \max_{j=1,\ldots,p}\| \frac{1}{n}\sum_{i=1}^n (z_{ij}z_i^T-\hat\Gamma_{j\cdot})^T(z_{ij}z_i^T-\hat\Gamma_{j\cdot}) \|_\infty^{1/2}$, we obtain
\begin{equation}\label{bound:t}\|t(\beta_0)-t(\hat\beta)\|_\infty \leq H_n \|\hat\beta-\beta_0\|_1 \end{equation}
and the inequality on $\|\hat u \|_\infty$ stated in the lemma follows.

{} We now establish the last claim of the lemma. From (\ref{opt:pivotal}), and the inequalities $\|t(\hat\beta)\|_\infty \leq \|\hat t\|_\infty$, $\|\hat\beta\|_\infty \leq \|\hat u\|_\infty$ we get
\begin{equation}\label{eq:aux:approx} \begin{array}{rl}
\|\hat \beta\|_1 & \leq \|\beta_0\|_1 + \lambda_t \{\|t(\beta_0)\|_\infty - \|t(\hat\beta)\|_\infty\}+\lambda_u\{\|\beta_0\|_\infty-\|\hat\beta\|_\infty\}\\
 & \leq \|\beta_0\|_1 + \lambda_t \|t(\beta_0)-t(\hat\beta)\|_\infty+\lambda_u\|\hat\beta-\beta_0\|_1.
% & = \|\beta_0\|_1 + \lambda_t \|\{\frac{1}{n}Z'Z-\hat \Gamma\}(\beta_0-\hat\beta)\|_\infty+\lambda_u\|\beta_0-\hat\beta\|_1\\
%  & \leq \|\beta_0\|_1 + \lambda_t \|\frac{1}{n}Z'Z-\hat \Gamma\|_\infty \|\beta_0-\hat\beta\|_1+\lambda_u\|\beta_0-\hat\beta\|_1\\
\end{array}\end{equation}
Setting $\lambda_t = \frac{1}{4H_n}$, $\lambda_u=1/4$, and using the fact that $\|\hat \beta\|_1=\|\hat \beta_{T}\|_1+\|\hat \beta_{T^c}\|_1$, we obtain
$ \frac{1}{2}\|\hat \beta_{T^c}\|_1 \leq \frac{3}{2}\|\beta_0-\hat\beta_T\|_1.$
\end{proof}

\begin{proof}[Proof of Theorem \ref{thm:pivotal}]
Set $Z=[z_1;\ldots;z_n]^T$ and $W=[w_1;\ldots;w_n]^T$. By the triangle inequality,
{\small \begin{equation}\label{eq:techone}\begin{array}{c}
\left\|\frac{1}{n}X^TX(\hat\beta-\beta_0)\right\|_\infty
\leq \left\|\frac{1}{n}Z^T(Y-Z\hat\beta)+\hat\Gamma \hat\beta\right\|_\infty + \left\|(\frac{1}{n}Z^TW-\Gamma)\hat\beta \right\|_\infty\\
 + \left\|(\hat \Gamma - \Gamma)\hat\beta\right\|_\infty + \left\|\frac{1}{n} Z^T\xi\right\|_\infty + \left\|\frac{1}{n}W^TX(\hat\beta-\beta_0)\right\|_\infty.\end{array}\end{equation}}\\

%$$\begin{array}{ll} \kappa_q(s,3)\|\hat\beta-\beta_0\|_q  \leq \big\|\frac{1}{n}X^TX(\hat\beta-\beta_0)\big\|_\infty \\
%\leq \tau\|t(\beta_0)\|_\infty + \tau_0 + \{ (1+\tau)b_\epsilon +\tau_\infty\} \|\beta_0\|_\infty + \tau_2\|\beta_0\|_2+\tilde \mu_1\|\hat\beta-\beta_0\|_1,
%\end{array}$$
%where $\tilde \mu_1=  \tau_1+ 5\tau H_n + 5(1+\tau)b_\epsilon $.

We now bound separately the terms on the RHS in (\ref{eq:techone}).

As shown at the end of this proof, the second term in (\ref{eq:techone}) is bounded with probability at least $1-6\varepsilon$ as follows:
\begin{equation}\label{a6a}
\begin{array}{rl} \left\|(\frac{1}{n}Z^TW-\Gamma)\hat\beta \right\|_\infty & \leq \{\delta_1(\varepsilon)+\delta_4(\varepsilon)+\delta_5(\varepsilon)\}\|\hat\beta-\beta_0\|_1 \\
 &+ \{\delta_1'(\varepsilon)+\delta_4'(\varepsilon)\}\|\beta_0\|_2+\delta_5(\varepsilon)\|\beta_0\|_\infty,
 \end{array}
 \end{equation}
where the quantities $\delta_i(\varepsilon)$ are defined in Appendix \ref{stocterm}. By Condition B, the third term in (\ref{eq:techone}) is  bounded with probability at least $1-\epsilon$  as follows:
$$ \begin{array}{rl}
\big\|(\hat \Gamma - \Gamma)\hat\beta\big\|_\infty & \leq \big\|(\hat \Gamma - \Gamma)\beta_0\big\|_\infty + \big\|(\hat \Gamma - \Gamma)(\hat\beta-\beta_0)\big\|_\infty \\
& \leq b_\epsilon \|\beta_0\|_\infty + b_\epsilon \|\hat\beta-\beta_0\|_\infty.
 \end{array}$$
Lemma \ref{lem2} provides, with probability at least $1-2\varepsilon$,  the following bound on the fourth term in (\ref{eq:techone}) :
$$ \begin{array}{rl}
\left\|\frac{1}{n} Z^T\xi\right\|_\infty  & \leq \left\|\frac{1}{n} X^T\xi\right\|_\infty+ \left\|\frac{1}{n} W^T\xi\right\|_\infty \leq \delta_2(\varepsilon)+\delta_3(\varepsilon). \end{array}
 $$
Finally the last term in (\ref{eq:techone}) is bounded, with probability at least $1-\varepsilon$, again via Lemma \ref{lem2}:
$$\big\|\mbox{$\frac{1}{n}$}W^TX(\hat\beta-\beta_0)\big\|_\infty \leq \big\|\mbox{$\frac{1}{n}$}X^TW\big\|_\infty\|\hat\beta-\beta_0\|_1 \leq \delta_1(\varepsilon)\|\hat\beta-\beta_0\|_1.$$
Therefore, with probability at least $1-9\varepsilon-\epsilon$ we have
\begin{equation}\label{midterm}\begin{array}{rl}\left\|\frac{1}{n}X^TX(\hat\beta-\beta_0)\right\|_\infty
& \leq \left\|\frac{1}{n}Z^T(Y-Z\hat\beta)+\hat\Gamma \hat\beta\right\|_\infty \\
 & + \tau_0 + \tau_\infty \|\beta_0\|_\infty + \tau_2\|\beta_0\|_2+\tau_1\|\hat\beta-\beta_0\|_1,\end{array}\end{equation}
where
$$\begin{array}{lll}
\tau_0 & := \delta_2(\varepsilon)+\delta_3(\varepsilon)   & \leq \{\sigma_\xi m_2^{1/2}+ C_{\xi w}\}\sqrt{\frac{2\log(2p/\varepsilon)}{n}} \\
\tau_\infty & := b_\epsilon+\delta_5(\varepsilon)         & \leq  b_\epsilon + C_{\xi w} \sqrt{\frac{2\log(2p/\varepsilon)}{n}}\\
\tau_2 & := \delta_1'(\varepsilon)+\delta_4'(\varepsilon) & \leq \{\sigma_wm_2^{1/2}+C_{\xi w}\} \sqrt{\frac{2\log(2p/\varepsilon)}{n}} \\
\tau_1 & := 2\delta_1(\varepsilon)+\delta_4(\varepsilon)+\delta_5(\varepsilon)+b_\epsilon & \leq b_\epsilon+ 2\{\sigma_w m_2^{1/2}+C_{\xi w}  \}\sqrt{\frac{2\log(2p^2/\varepsilon)}{n}}.
\end{array}$$
Here, $C_{\xi w}$ is a positive constant depending only on $\sigma_\xi$ and $\sigma_w$, and the bounds hold for $n$ large enough under the condition $C_{\xi w}\log(p/\varepsilon) =o(n)$.

{} Next, we bound the first term in (\ref{eq:techone}). By the feasibility of $(\hat\beta,\hat t, \hat u)$ in (\ref{est:pivotal}) we have
$$\begin{array}{rl}
 \left\|\frac{1}{n}Z^T(Y-Z\hat\beta)+\hat\Gamma \hat\beta\right\|_\infty & \leq \tau \|\hat t\|_\infty + (1+\tau)b_\epsilon \|\hat u\|_\infty.
 \end{array}
 $$
By Lemma \ref{lem:pivotal} and the choices $\lambda_t = 1/\{4H_n\}$ and $\lambda_u=1/4$, with probability $1-\alpha\{1+o(1)\}-\epsilon$ we have  $\hat\beta-\beta_0 \in C_T(3)$ and the bounds on $\|\hat t\|_\infty$ and $\|\hat u\|_\infty$ apply, so that
\begin{equation}\label{thm1lastpiece}\begin{array}{ll}
 \big\|\frac{1}{n}Z^T(Y-Z\hat\beta)+\hat\Gamma \hat\beta \big\|_\infty
 \leq \tau\|t(\beta_0)\|_\infty + \tau\frac{1+\lambda_u}{\lambda_t}\|\hat\beta-\beta_0\|_1  \\
 + (1+\tau)b_\epsilon \left\{ \|\beta_0\|_\infty +  \left(\frac{1}{\lambda_u}+\frac{\lambda_t}{\lambda_u}H_n\right)\|\hat\beta-\beta_0\|_1\right\}\\
 =\tau\|t(\beta_0)\|_\infty + 5\tau H_n\|\hat\beta-\beta_0\|_1  + (1+\tau)b_\epsilon \left\{ \|\beta_0\|_\infty +  5\|\hat\beta-\beta_0\|_1\right\}.
 \end{array}
 \end{equation}
Next, on the event $\hat\beta-\beta_0 \in C_T(3)$ we  bound the LHS of (\ref{midterm}) from below via the $\ell_q$-sensitivity. Plugging that lower bound and (\ref{thm1lastpiece}) in (\ref{midterm}) we find
 \begin{equation}\label{eq:extraneed}\begin{array}{ll} \kappa_q(s,3)\|\hat\beta-\beta_0\|_q  \leq \big\|\frac{1}{n}X^TX(\hat\beta-\beta_0)\big\|_\infty \\
 \leq \tau\|t(\beta_0)\|_\infty + \tau_0 + \{ (1+\tau)b_\epsilon +\tau_\infty\} \|\beta_0\|_\infty + \tau_2\|\beta_0\|_2+\tilde \mu_1\|\hat\beta-\beta_0\|_1,
\end{array}\end{equation}
where $\tilde \mu_1=  \tau_1+ 5\tau H_n + 5(1+\tau)b_\epsilon $. Note that
$$\tilde \mu_1 \leq \tau_1+ 5\tau H_\epsilon + 5(1+\tau)b_\epsilon \leq  (1+\tau)b_\epsilon + \tau h_\epsilon + C'(1+m_2^{1/2})\sqrt{\log(2p^2/\varepsilon)/n}$$ with probability $1-\epsilon$ where $C' = \sigma_w\vee C_{w\xi}$ is bounded by a constant since $\sigma_w \vee \sigma_\xi\leq C$ under Condition A. Moreover, since $\hat\beta-\beta_0\in C_T(3)$ we have $$\|\hat\beta-\beta_0\|_1 \leq 4\|(\hat\beta-\beta_0)_T\|_1 \leq 4s^{1-1/q}\|(\hat\beta-\beta_0)_T\|_q\leq 4s^{1-1/q}\|\hat\beta-\beta_0\|_q.$$ Thus under the condition of the theorem on $\kappa_q(s,3)$, we have with probability $1-\alpha\{1+o(1)\}-2\epsilon-9\varepsilon$ that
$$  \frac{\kappa_q(s,3)}{2}\|\hat\beta-\beta_0\|_q \leq \tau\|t(\beta_0)\|_\infty + \tau_0 + \{ (1+\tau)b_\epsilon +\tau_\infty\} \|\beta_0\|_\infty + \tau_2\|\beta_0\|_2.$$
The result follows by noticing that $(1+\tau)b_\epsilon \leq 2b_\epsilon \leq 2\tau_\infty$ for large enough $n$.

{\it Proof of \eqref{a6a}.}  We have
$$ \begin{array}{rl} \big\|(\frac{1}{n}Z^TW-\Gamma)\hat\beta \big\|_\infty & \leq \big\|(\frac{1}{n}Z^TW-\Gamma)\beta_0 \big\|_\infty+ \big\|(\frac{1}{n}Z^TW-\Gamma)(\hat\beta-\beta_0) \big\|_\infty\\
& \leq  \big\|(\frac{1}{n}Z^TW-\Gamma)\beta_0 \big\|_\infty + \|\frac{1}{n}Z^TW-\Gamma\|_\infty\|\hat\beta-\beta_0 \|_1\\
& \leq \big\|(\frac{1}{n}W^TW-\Gamma)\beta_0 \big\|_\infty + \big\|\frac{1}{n}X^TW\beta_0 \big\|_\infty \\
& +  \|\frac{1}{n}X^TW\|_\infty\|\hat\beta-\beta_0 \|_1+ \|\frac{1}{n}W^TW-\Gamma\|_\infty\|\hat\beta-\beta_0 \|_1. \end{array}$$
By Lemma \ref{lem2} we get,  with probability at least $1-3\varepsilon$,
$$\begin{array}{rl}
& \|\frac{1}{n}X^TW\|_\infty \leq \delta_1(\varepsilon)\end{array},$$$$\begin{array}{rl}
\|\frac{1}{n}W^TW-\Gamma\|_\infty &\leq  \|\frac{1}{n}W^TW-\frac{1}{n}{\rm Diag}(W^TW)\|_\infty +  \|\frac{1}{n}{\rm Diag}(W^TW)-\Gamma\|_\infty\\
& \leq  \delta_4(\varepsilon)+\delta_5(\varepsilon).\end{array}$$
Finally, Lemma \ref{lem3a} and Lemma \ref{lem2} yield that, with probability at least $1-3\varepsilon$,
$$
\begin{array}{rl}
& \|\frac{1}{n}X^TW\beta_0\|_\infty \leq \delta_1'(\varepsilon)\|\beta_0\|_2\end{array},$$$$ \begin{array}{rl}
\left\|(\frac{1}{n}W^TW-\Gamma)\beta_0 \right\|_\infty  & \leq  \left\|\frac{1}{n}(W^TW-{\rm Diag}(W^TW))\beta_0 \right\|_\infty  \\ & + \left\|(\frac{1}{n}{\rm Diag}(W^TW)-\Gamma)\beta_0 \right\|_\infty\\
& \leq  \delta_4'(\varepsilon)\|\beta_0\|_2 + \|\frac{1}{n}{\rm Diag}(W^TW)-\Gamma\|_\infty \|\beta_0 \|_\infty \\
& \leq   \delta_4'(\varepsilon)\|\beta_0\|_2 + \delta_5(\varepsilon)\|\beta_0 \|_\infty.
\end{array}$$
\vspace{-.2in}\end{proof}

\begin{proof}[Proof of Corollary \ref{cor:pivotal}]
By Theorem \ref{thm:pivotal} with probability $1-\alpha\{1+o(1)\}-11\varepsilon$ we have
$$ \|\hat\beta-\beta_0\|_q \leq \frac{\tau\|t(\beta_0)\|_\infty}{c'\kappa_q(s,3)}  + \frac{(1+\|\beta_0\|_2)(1+m_2^{1/2})}{c'\kappa_q(s,3)}\sqrt{\frac{\log(2p/\varepsilon)}{n}}+ \frac{b_\varepsilon  \|\beta_0\|_\infty}{c'\kappa_q(s,3)}.$$
Under the additional condition $X \in \Omega_X$, we have by Lemma \ref{lem:Hn} that $P( H_n \leq C ) \geq 1-o(1)$. Therefore we have that with probability $1-\alpha\{1+o(1)\}-11\varepsilon-o(1)$
$$ \|\hat\beta-\beta_0\|_q \leq Cs^{1/q}\Big\{\tau\|t(\beta_0)\|_\infty + (1+\|\beta_0\|_2)\sqrt{\frac{\log(2p/\varepsilon)}{n}}+ b_\varepsilon  \|\beta_0\|_\infty\Big\}$$
since $m_2^{1/2} \leq \{\max_{j\leq p}\frac{1}{n}\sum_{i=1}^n x_{ij}^4 \}^{1/4} \leq C^{1/4}$ when  $X \in \Omega_X$.
Using the triangle inequality, we obtain
$$\begin{array}{rl}
\|t(\beta_0)\|_\infty & = \max_{j\leq p} \left\{ \frac{1}{n}\sum_{i=1}^n\{z_{ij}(\xi_i-w_i^T\beta_0)+\hat \Gamma_{jj}\beta_{0j}\}^2\right\}^{1/2} \\
& \leq \max_{j\leq p} \left\{ \frac{1}{n}\sum_{i=1}^n\{z_{ij}(\xi_i-w_i^T\beta_0)\}^2\right\}^{1/2}+|\hat \Gamma_{jj}\beta_{0j}| \\
& \leq_{(i)} \max_{j\leq p} \left\{ \frac{1}{n}\sum_{i=1}^n\{z_{ij}(\xi_i-w_i^T\beta_0)\}^2\right\}^{1/2}\\
& +|\Gamma_{jj}\beta_{0j}| + b_\varepsilon\|\beta_0\|_\infty \\
& \leq_{(ii)} \max_{j\leq p} \left\{ \frac{1}{n}\sum_{i=1}^nz_{ij}^4\right\}^{1/4}\left\{ \frac{1}{n}\sum_{i=1}^n(\xi_i-w_i^T\beta_0)^4\right\}^{1/4}\\
& +|\Gamma_{jj}\beta_{0j}| + b_\varepsilon\|\beta_0\|_\infty \\
& \leq_{(iii)} \max_{j\leq p} \left\{ \frac{1}{n}\sum_{i=1}^nx_{ij}^4\right\}^{1/4}\left\{ \frac{1}{n}\sum_{i=1}^n(\xi_i-w_i^T\beta_0)^4\right\}^{1/4}\\
&+\max_{j\leq p} \left\{ \frac{1}{n}\sum_{i=1}^nw_{ij}^4\right\}^{1/4}\left\{ \frac{1}{n}\sum_{i=1}^n(\xi_i-w_i^T\beta_0)^4\right\}^{1/4}\\
& +|\Gamma_{jj}\beta_{0j}| + b_\varepsilon\|\beta_0\|_\infty,
\end{array}$$
where (i) follows from the inequality $\|\hat\Gamma-\Gamma\|_\infty \leq b_\varepsilon$ which holds with probability $1-\varepsilon$ by Condition B, (ii) follows from the Cauchy-Schwarz inequality, and (iii) from the triangle inequality. On the event $X \in \Omega_X$, we have $\max_{j\leq p} \left\{ \frac{1}{n}\sum_{i=1}^nx_{ij}^4\right\}^{1/4}\leq C$. Note that $w_{ij}$, $i=1,\ldots,n$, $j=1,\ldots,p$, are $\sigma_w$-subgaussian random variables with $\sigma_w \leq C$. Therefore,  by Lemmas \ref{lem:subgaussian} and \ref{lem:m2bound},
$\max_{i\leq n,j\leq p}\frac{1}{n}\sum_{i=1}^n\Ep[w_{ij}^4] \leq C'$,  $\Ep[\max_{i\leq n,j\leq p}|w_{ij}|^4] \leq C'\log^2(pn)$, and $$
\Ep\left[\max_{j\leq p} \left|\frac{1}{n}\sum_{i=1}^nw_{ij}^4 -\Ep[w_{ij}^4]\right|\right] \leq C'\frac{\log(p)}{n}\log^2(pn)+C'\sqrt{\frac{\log(p)}{n}\log^2(pn)}
\leq o(1)$$
where we have used the relation $\log^3(2p) = o(n)$ following from Condition B(ii). Then using Markov's inequality and the fact that $w_i$'s are subgaussian, with probability $1-o(1)$ we get $\max_{j\leq p} \big\{ \frac{1}{n}\sum_{i=1}^nw_{ij}^4\big\}^{1/4} \leq C''$.

{} Next note that $\left\{ \frac{1}{n}\sum_{i=1}^n(\xi_i-w_i^T\beta_0)^4\right\}^{1/4}\leq C'(1+\|\beta_0\|_2)$ with probability $1-o(1)$. Indeed, each of the random variables $\tilde \xi_i :=\xi_i-w_i^T\beta_0$ is subgaussian with parameter bounded by $C(1+\|\beta_0\|_2)$. Thus we have $${\rm Var}\left(\frac{1}{n}\sum_{i=1}^n\tilde \xi_i^4\right) \leq \frac{1}{n^2}\sum_{i=1}^n\Ep[\tilde \xi_i^8] \leq C'(1+\|\beta_0\|_2)^8/n.$$ Therefore, using Markov's inequality, we get with probability at least $1-n^{-1/2}$,
$$ \begin{array}{rl}
 \frac{1}{n}\sum_{i=1}^n\tilde \xi_i^4 & \leq \left|\frac{1}{n}\sum_{i=1}^n\tilde \xi_i^4- \frac{1}{n}\sum_{i=1}^n\Ep[\tilde \xi_i^4]\right|+ C(1+\|\beta_0\|_2)^4 \\
& \leq n^{1/4}C'(1+\|\beta_0\|_2)^4/\sqrt{n} + C(1+\|\beta_0\|_2)^4\\
& \leq  C''(1+\|\beta_0\|_2)^4.\end{array}$$
Thus with probability $1-\varepsilon-o(1)$ we have
$$ \|t(\beta_0)\|_\infty \leq C(1+\|\beta_0\|_2) + (C+b_\varepsilon)\|\beta_0\|_\infty.$$
Since $\tau = n^{-1/2}\Phi^{-1}(1-\alpha/(2p)) \leq \sqrt{2\log(2p/\alpha)/n}$, $b_\varepsilon\leq C\sqrt{\log(2p/\varepsilon)/n}$ and $\|\beta_0\|_\infty\leq \|\beta_0\|_2$, the result follows.\end{proof}

\begin{proof}[Proof of Theorem \ref{thm:pivotal-post-relaxed}]
First note that the feasibility constrains in (\ref{est:pivotal-post-relaxed}) and (\ref{est:pivotal}) are the same. Therefore, by Lemma  \ref{lem:pivotal} the triple $(\beta_0,t(\beta_0),|\beta_0|)$ is feasible with probability $1-\alpha\{1+o(1)\}-\epsilon$. In that event  only the last result of Lemma \ref{lem:pivotal} requires modification for the estimator (\ref{est:pivotal-post-relaxed}). Next we will show that
$\tilde \beta - \beta_0 \in C_{T\cup\widehat T}(3)$.

From (\ref{opt:pivotal}), and the inequalities $\|t(\tilde\beta)\|_\infty \leq \|\hat t\|_\infty$, $\|\tilde\beta_{\widehat T^c}\|_\infty \leq \|\hat u_{\widehat T^c}\|_\infty$ we get
\begin{equation}\label{eq:aux:approx} \begin{array}{rl}
\|\hat \beta_{\widehat T^c}\|_1 & \leq \|\beta_{0\widehat T^c}\|_1 + \lambda_t \{\|t(\beta_0)\|_\infty - \|t(\tilde\beta)\|_\infty\}+\lambda_u\{\|\beta_{0\widehat T^c}\|_\infty-\|\tilde\beta_{\widehat T^c}\|_\infty\}\\
 & \leq \|\beta_{0\widehat T^c}\|_1 + \lambda_tH_n\|\beta_0-\tilde\beta\|_1+\lambda_u\|\hat\beta-\beta_0\|_1\\
 & \leq \|\beta_{0\widehat T^c}\|_1 + \lambda_t \|t(\beta_0)-t(\hat\beta)\|_\infty+\lambda_u\|\hat\beta-\beta_0\|_1.
% & = \|\beta_0\|_1 + \lambda_t \|\{\frac{1}{n}Z'Z-\hat \Gamma\}(\beta_0-\hat\beta)\|_\infty+\lambda_u\|\beta_0-\hat\beta\|_1\\
%  & \leq \|\beta_0\|_1 + \lambda_t \|\frac{1}{n}Z'Z-\hat \Gamma\|_\infty \|\beta_0-\hat\beta\|_1+\lambda_u\|\beta_0-\hat\beta\|_1\\
\end{array}\end{equation}
Setting $\lambda_t = \frac{1}{4H_n}$, $\lambda_u=1/4$, and using the fact that $\|\tilde \beta_{\widehat T^c}\|_1=\|\tilde \beta_{T\cup\widehat{T}}\|_1+\|\tilde \beta_{(T\cup\widehat{T})^c}\|_1$, we have
$$\|\hat \beta_{(T\cup\widehat T)^c}\|_1 \leq \|\beta_{0\widehat T^c}-\tilde \beta_{T\cup\widehat T}\|_1 + \frac{1}{2}\|\tilde \beta - \beta_0\|_1$$
so that
$ \frac{1}{2}\|\tilde \beta_{(T\cup\widehat T)^c}\|_1 \leq \frac{3}{2}\|\beta_{0\widehat T}-\hat\beta_{T\cup \widehat T}\|_1.$

The rest of the proof follows as in Theorem \ref{thm:pivotal} with $\kappa_q(s+\hat k,3)$ replacing $\kappa_q(s,3)$.
\end{proof}

\begin{proof}[Proof of Theorem \ref{thm:post-est}]
Set $Z=[z_1;\ldots;z_n]^T$ and $W=[w_1;\ldots;w_n]^T$.

We have that
\begin{equation} \begin{array}{rl}
(\tilde \beta - \beta_{0})^T\frac{1}{n}X^TX(\tilde \beta - \beta_{0})& = (\tilde \beta - \beta_{0})^T\frac{1}{n}X^T_{\widehat T}X(\tilde \beta - \beta_{0})\\
&+ (\tilde \beta - \beta_{0\widehat T^c})^T\frac{1}{n}X^T_{\widehat T^c}X(\tilde \beta - \beta_{0})\\
&\leq \|\tilde \beta - \beta_{0}\|_1 \|\frac{1}{n}X^T_{\widehat T}X(\tilde \beta - \beta_{0})\|_\infty \\
&+ \beta_{0\widehat T^c}^T\frac{1}{n}X^TX(\tilde \beta - \beta_{0})\\
\end{array}
\end{equation}
It follows that
$$ \left|\beta_{0\widehat T^c}^T\mbox{$\frac{1}{n}$}X^TX(\tilde \beta - \beta_{0}) \right|\leq \left\{\beta_{0\widehat T^c}^T\mbox{$\frac{1}{n}$}X^TX\beta_{0\widehat T^c}\right\}^{1/2}\left\{(\tilde \beta - \beta_{0})^T\mbox{$\frac{1}{n}$}X^TX(\tilde \beta - \beta_{0})\right\}^{1/2}$$
where $\{\beta_{0\widehat T^c}^T\frac{1}{n}X^TX\beta_{0\widehat T^c}\}^{1/2}\leq \|\beta_{0\widehat T^c}\|\phi_{\max}^{1/2}(s)\leq\|\hat\beta-\beta_0\|\phi_{\max}^{1/2}(s)$.

By the triangle inequality,
{\small \begin{equation}\label{eq:main-post-est-proof} \begin{array}{c}
\left\|\frac{1}{n}X^T_{\widehat T}X(\tilde\beta-\beta_0)\right\|_\infty \leq
 \left\|\frac{1}{n}Z^T_{\widehat T}(Y-Z\tilde\beta)+\hat\Gamma_{\widehat T} \tilde\beta\right\|_\infty + \left\|(\frac{1}{n}Z^T_{\widehat T}W-\Gamma)\tilde\beta \right\|_\infty\\
 + \left\|(\hat \Gamma_{\widehat T} - \Gamma_{\widehat T})\tilde\beta\right\|_\infty   + \left\|\frac{1}{n} Z^T_{\widehat T}\xi\right\|_\infty + \left\|\frac{1}{n}W^T_{\widehat T}X(\tilde\beta-\beta_0)\right\|_\infty.\end{array}\end{equation}}\\

Following the proof of Theorem \ref{thm:pivotal} and (\ref{midterm}), with probability at least $1-9\varepsilon-\epsilon$ we have
\begin{equation}\label{eq:extra-post-main-again}\begin{array}{rl}\left\|\frac{1}{n}X^T_{\widehat T}X(\tilde\beta-\beta_0)\right\|_\infty
& \leq \left\|\frac{1}{n}Z^T_{\widehat T}(Y-Z\tilde\beta)+\hat\Gamma_{\widehat T} \tilde\beta\right\|_\infty \\
 & + \tau_0 + \tau_\infty \|\beta_0\|_\infty + \tau_2\|\beta_0\|_2+\tau_1\|\tilde\beta-\beta_0\|_1,\end{array}\end{equation}
where
$$\begin{array}{lll}
\tau_0 & := \delta_2(\varepsilon)+\delta_3(\varepsilon)   & \leq \{\sigma_\xi m_2^{1/2}+ C_{\xi w}\}\sqrt{\frac{2\log(2p/\varepsilon)}{n}} \\
\tau_\infty & := b_\epsilon+\delta_5(\varepsilon)         & \leq  b_\epsilon + C_{\xi w} \sqrt{\frac{2\log(2p/\varepsilon)}{n}}\\
\tau_2 & := \delta_1'(\varepsilon)+\delta_4'(\varepsilon) & \leq \{\sigma_wm_2^{1/2}+C_{\xi w}\} \sqrt{\frac{2\log(2p/\varepsilon)}{n}} \\
\tau_1 & := 2\delta_1(\varepsilon)+\delta_4(\varepsilon)+\delta_5(\varepsilon)+b_\epsilon & \leq b_\epsilon+ 2\{\sigma_w m_2^{1/2}+C_{\xi w}  \}\sqrt{\frac{2\log(2p^2/\varepsilon)}{n}}.
\end{array}$$
Here, $C_{\xi w}$ is a positive constant depending only on $\sigma_\xi$ and $\sigma_w$, and the bounds hold for $n$ large enough under the condition $C_{\xi w}\log(p/\varepsilon) =o(n)$.

To control the first term in the RHS of (\ref{eq:main-post-est-proof}) note that for any $\Delta \in C_{\widehat T}(0) = \{ \tilde \Delta \in  \mathbb{R}^p : \|\tilde \Delta_{\widehat T^c}\|_1 = 0\}$ (with support in $\widehat T$), we have
{\small $$\begin{array}{rl}
\Delta^T\left\{ \frac{1}{n}\sum_{i=1}^nz_{i\widehat T}z_{i\widehat T}^T - \widehat \Gamma_{\widehat T,\widehat T}\right\}\Delta & \geq \Delta^T\left\{\frac{1}{n}\sum_{i=1}^nx_{i\widehat T}x_{i\widehat T}^T\right\} \Delta \\
& - \|\Delta\|_1^2 \left\{ \|\frac{1}{n}W^TW-\Gamma\|_\infty + \| \Gamma - \widehat \Gamma\|_\infty + 2\|W^TX\|_\infty \right\}\\
 & \geq_{(1)} \phi_{\min}(\hat k+s)\|\Delta\|_2^2 - \hat k \|\Delta\|_2^2 \tilde \mu_1\\
&\geq_{(2)} \frac{1}{2}\phi_{\min}(\hat k+s)\|\Delta\|_2^2 \\
\end{array}
$$}
%$$\begin{array}{rl}
%\left\|\left\{ \frac{1}{n}\sum_{i=1}^nz_{i\widehat T}z_{i\widehat T}^T - \widehat \Gamma_{\widehat T,\widehat T}\right\}\Delta \right\|_\infty & \geq \left\|\left\{\frac{1}{n}\sum_{i=1}^nx_{i\widehat T}x_{i\widehat T}^T\right\} \Delta \right\|_\infty\\
%& - \|\Delta\|_1 \left\{ \|\frac{1}{n}W^TW-\Gamma\|_\infty + \| \Gamma - \widehat \Gamma\|_\infty + 2\|W^TX\|_\infty \right\}\\
% & \geq_{(1)} \kappa_1(\hat k,0)\|\Delta\|_1 - \|\Delta\|_1 \tilde \mu_1\\
%&\geq_{(2)} \frac{1}{2}\kappa_1(\hat k,0)\|\Delta\|_1 \\
%\end{array}
%$$
 where (1) holds with probability $1-\epsilon-4\varepsilon$ and (2) holds by $\tilde \mu_1 \hat k \leq \frac{1}{2}\phi_{\min}(\hat k)$ (implied by our condition since $\tilde \mu_1 \leq  (1+\tau)b_\epsilon + \tau h_\epsilon + C'(1+m_2^{1/2})\sqrt{\log(2p^2/\varepsilon)/n}$).  Note that we can assume $\phi_{\min}(\hat k) > 0$ (otherwise the result is trivial) so that the derivation above implies that $\left\{ \frac{1}{n}\sum_{i=1}^nz_{i\widehat T}z_{i\widehat T}^T - \widehat \Gamma_{\widehat T,\widehat T}\right\}$ has full rank. In turn implies
$$ \begin{array}{rl}
\left\|\frac{1}{n}Z^T_{\widehat T}(Y-Z\tilde\beta)+\hat\Gamma_{\widehat T} \tilde\beta\right\|_\infty & = \left\|\frac{1}{n}Z^T_{\widehat T}Y-\left\{\frac{1}{n}Z^T_{\widehat T}Z_{\widehat T} - \hat\Gamma_{\widehat T}\right\} \tilde\beta\right\|_\infty \\
& = \min_{ \tilde \beta \in C_{\widehat T}(0)} \left\|\frac{1}{n}Z^T_{\widehat T}Y-\left\{\frac{1}{n}Z^T_{\widehat T}Z_{\widehat T} - \hat\Gamma_{\widehat T}\right\} \tilde\beta\right\|_\infty \\
& = 0.\end{array}
 $$
Therefore, by (\ref{eq:extra-post-main-again}) and $\left\|\frac{1}{n}Z^T_{\widehat T}(Y-Z\tilde\beta)+\hat\Gamma_{\widehat T} \tilde\beta\right\|_\infty=0$, we have that with probability $1-\epsilon-9\varepsilon$ that
\begin{equation}\label{eq:extraneed2}\begin{array}{rl} \big\|\frac{1}{n}X^T_{\widehat T}X(\tilde\beta-\beta_0)\big\|_\infty & \leq \tau_0 + \tau_\infty \|\beta_0\|_\infty + \tau_2\|\beta_0\|_2+\tau_1\|\tilde\beta-\beta_0\|_1\\
& =: \varphi_{\beta_{0}} + \tau_1\|\tilde\beta-\beta_0\|_1,
\end{array}\end{equation}

Letting $u_n^2 := (\tilde \beta - \beta_{0})^T\frac{1}{n}X^TX(\tilde \beta - \beta_{0})$, we have that

$$\begin{array}{rl}
 u_n^2 & \leq \|\tilde \beta - \beta_0\|_1 \|\frac{1}{n}X_{\widehat T}^TX(\tilde \beta-\beta_0)\|_\infty + \beta_{0\widehat T^c}^T\frac{1}{n}X^TX(\tilde \beta - \beta_0) \\
 & \leq \frac{\{\hat k+ s\}^{1/2}}{\phi_{\min}^{1/2}(\hat k+s)} u_n \|\frac{1}{n}X_{\widehat T}^TX(\tilde \beta-\beta_0)\|_\infty + \|\beta_{0\widehat T^c}\|\phi^{1/2}_{\max}(s) u_n\\
 & \leq_{(1)} u_n \frac{\sqrt{\hat k+s}}{\phi_{\min}^{1/2}(\hat k+s)}\left\{ \varphi_{\beta_0} + u_n\tau_1 \frac{\sqrt{\hat k+s}}{\phi_{\min}^{1/2}(|\widehat T|+s)}\right\} + \|\beta_{0\widehat T^c}\|\phi^{1/2}_{\max}(s) u_n\\
 \end{array}
 $$
 where (1) follows from (\ref{eq:extraneed2}).
The result follows under the condition  $\tau_1 \{\hat k+s\} \leq \frac{1}{2}\phi_{\min}(\hat k+s)$ and noting that $u_n \geq \|\tilde \beta - \beta_0\|_2 \phi_{\min}^{1/2}(\hat k+s)$.

%The second inequality follows from $\tilde \beta - \beta_0$ having at most $\hat k + s$ non-zero components.
\end{proof}

\section{Auxiliary Lemmas}

\begin{lemma}\label{lem:Hn}
Under Conditions A and B, if $\max_{j\leq p}\frac{1}{n}\sum_{i=1}^n x_{ij}^4 \leq C$, we have that $P(H_n \leq C') = 1-o(1)$.
\end{lemma}
\begin{proof}[Proof of Lemma \ref{lem:Hn}]
We have
$$
\begin{array}{rl}
H_n^2 & = \max_{1\leq j,k,\ell \leq p} \left|\frac{1}{n}\sum_{i=1}^n(z_{ij}z_{ik}-\hat\Gamma_{jk})(z_{ij}z_{i\ell}-\hat\Gamma_{j\ell})\right| \\
&  \leq \max_{1\leq j,k,\ell \leq p} \left|\frac{1}{n}\sum_{i=1}^nz_{ij}^2z_{ik}z_{i\ell}\right| + \|\hat\Gamma\|_\infty^2\\
& +2\|\hat \Gamma\|_\infty  \max_{1\leq j,k \leq p} \left|\frac{1}{n}\sum_{i=1}^nz_{ij}z_{ik}\right|. \\
\end{array}
$$

{} Note that
$$
\begin{array}{rl}
\left|\frac{1}{n}\sum_{i=1}^nz_{ij}^2z_{ik}z_{i\ell}\right|& \leq \frac{1}{n}\sum_{i=1}^nz_{ij}^2\frac{( z_{ik}^2+z_{i\ell}^2)}{2}\\
& \leq \frac{1}{4}\left|\frac{1}{n}\sum_{i=1}^nz_{ij}^4+z_{ik}^4\right|+\frac{1}{4}\left|\frac{1}{n}\sum_{i=1}^nz_{ij}^4+z_{i\ell}^4\right| \\
&\leq \max_{j\leq p}\frac{1}{n}\sum_{i=1}^nz_{ij}^4.\\
\end{array}
$$
Moreover,  $$\max_{j\leq p}\frac{1}{n}\sum_{i=1}^nz_{ij}^4 \leq 8\max_{j\leq p}\frac{1}{n}\sum_{i=1}^nx_{ij}^4 + 8\max_{j\leq p}\frac{1}{n}\sum_{i=1}^nw_{ij}^4 \leq 8C + 8\max_{j\leq p}\frac{1}{n}\sum_{i=1}^nw_{ij}^4.$$

{} Since $w_{ij}$ are $\sigma_w$-subgaussian random variables with $\sigma_w \leq C$, Lemma \ref{lem:subgaussian} yields $M_4 = \Ep[\max_{i\leq n, j\leq p}w_{ij}^4] \leq C\log^2(pn)$. This and
Lemma \ref{lem:m2bound} imply
$$
\begin{array}{rl}
\Ep[{\displaystyle\max_{j\leq p}}\frac{1}{n}\sum_{i=1}^n(w_{ij}^4-\Ep[w_{ij}^4])] & \leq \frac{CM_4\log(2p)}{n} + \sqrt{\frac{CM_4\log(2p)}{n}} {\displaystyle \max_{j\leq p}}\left(\frac{1}{n}\sum_{i=1}^n\Ep[w_{ij}^4]\right)^{1/2}\\
& \leq \frac{C'\log^3(pn)}{n} + \sqrt{\frac{C'\log^3(pn)}{n}} = o(1)
\end{array}
$$ where the last equality follows from the relation $\Phi^{-1}(1-\alpha/(2p))=o(n^{-1/6})$.
The result now follows since $\|\hat\Gamma\|_\infty$ is bounded:  $\|\hat\Gamma\|_\infty\leq b_\epsilon + \|\Gamma\|_\infty$.
\end{proof}

\subsection{Bounds on the stochastic error terms}\label{stocterm}

{} The following technical lemmas were proved in \cite{BRT2014} and \cite{RT2} and are stated here for completeness. For a square matrix $A$, we denote
by $\text{Diag}\{A\}$ the matrix with the same dimensions as $A$,
the same diagonal elements, and all off-diagonal elements
equal to zero.

\begin{lemma}[Lemma 1 in \cite{RT2}]\label{lem2} Let $0<\e<1$ and assume Condition A holds. Then, with probability at least $1-\varepsilon$ (for each event),
\begin{eqnarray*}
&&\left\|\mbox{$\frac{1}{n}$} X^TW\right\|_\infty \le \delta_1(\varepsilon),\quad \left\| \mbox{$\frac{1}{n}$} X^T\xi\right\|_\infty \le \delta_2(\varepsilon),\quad \left\|\mbox{$\frac{1}{n}$} W^T\xi\right\|_\infty \le \delta_3(\varepsilon),\\
&& \left\|\mbox{$\frac{1}{n}$}(W^TW-\text{\rm Diag}\{W^TW\})\right\|_\infty \le \delta_4(\varepsilon), \quad\left\|\mbox{$\frac{1}{n}$} \text{\rm Diag}\{W^TW\}-\Gamma\right\|_\infty \le \delta_5(\varepsilon),
\end{eqnarray*}
where $m_2 := \max_{1\leq j\leq p}\frac{1}{n}\sum_{i=1}^n X_{ij}^2$,
\begin{eqnarray*}
&&
\delta_1(\varepsilon)=\sigma_w
\sqrt{\frac{2m_2\log(2p^2/\varepsilon)}{n}}, \quad \delta_2(\varepsilon)=\sigma_\xi
\sqrt{\frac{2m_2\log(2p/\varepsilon)}{n}},\\
&& \delta_3(\varepsilon)=\delta_5(\varepsilon)=\varpi
(\varepsilon,2p), \quad \delta_4(\varepsilon)=\varpi
(\varepsilon,p(p-1)),
\end{eqnarray*}
and for
an integer $N$,
$
\varpi (\varepsilon,N) = \max\left(\gamma_0
\sqrt{\frac{2\log(N/\varepsilon)}{n}},\
\frac{2\log(N/\varepsilon)}{t_0n}\right),
$
where $\gamma_0, t_0$ are positive constants depending only on $\sigma_\xi, \sigma_w$.
\end{lemma}

\begin{lemma}[Lemma 2 in \cite{BRT2014}] \label{lem3a}
Let $0<\e<1$, $\theta^*\in \R^p$ and assume that Condition A holds. Then, with probability at least $1-\varepsilon$,
$
\left\|\mbox{$\frac{1}{n}$} X^TW\theta^*\right\|_{\infty} \leq \delta_1'(\e)\|\theta^*\|_2,
$
where $\delta_1'(\e)= \sigma_w \sqrt{\frac{2m_2\log(2p/\e)}{n}}$.
In addition, %$(\log (2p/\e))/n \le 1$ then,
with probability at least $1-\varepsilon$,
\begin{align*}
&\left\|\mbox{$\frac{1}{n}$} (W^TW-{\rm Diag}\{W^TW\})\theta^*\right\|_{\infty}\leq\delta_4'(\e)\|\theta^*\|_2,
\end{align*}
where %$\delta_1'(\e)= \sigma_{*} \sqrt{\frac{2m_2\log(2p/\e)}{n}}$ and
$
\delta_4'(\e) =\max\left(\gamma_2
\sqrt{\frac{2\log(2p/\varepsilon)}{n}},\
\frac{2\log(2p/\varepsilon)}{t_2n}\right),
$
and $\gamma_2, t_2$ are positive constants depending only on $\sigma_w$.
\end{lemma}

%The following lemma recalls some results for subgaussian random variables.

\begin{lemma}\label{lem:subgaussian}
(1) If $X$ is a centered subgaussian random variable with parameter $\gamma$, it follows that for any $k>0$
$ \Ep[ |X|^k] \leq k2^{k/2}\gamma^k\Gamma( k/2 )$
and for $p\geq 1$ we have $\{\Ep[ |X|^k]\}^{1/k} \leq C\gamma \sqrt{k}$.
(2) If $X_j, j=1,\ldots,N$, is a collection of centered subgaussian variables with parameter $\gamma$, then for $k\geq 1$ we have
$ \Ep\left[ \max_{j\leq N}|X_j|^k \right] \leq \gamma^k \log^{k/2}( NC_k ) $ for some constant $C_k$ that depends only on $k$.\\
\end{lemma}

\begin{lemma}\label{lem:m2bound}(e.g.,\cite{BRT2014})
Let $X_i, i=1,\ldots,n,$ be independent random vectors in $\mathbb{R}^p$, $p\geq 3$. Define $\bar m_k := \max_{j\leq p}\frac{1}{n}\sum_{i=1}^n\mathbb{E}[|X_{ij}|^k]$ and $M_{k} \geq \mathbb{E}[ {\displaystyle \max_{i\leq n}}\|X_i\|_\infty^k]$. Then
{\small $$\mathbb{E}\Big[\max_{j\leq p}\frac{1}{n}\Big|\sum_{i=1}^n|X_{ij}|^k-\mathbb{E}[|X_{ij}|^k]\Big|\Big] \leq 2C^2 \frac{\log p}{n}M_{k}+2C\sqrt{\frac{\log p}{n}}M_{k}^{1/2}\bar m_k^{1/2} $$}
for some universal constant $C$.%\leq 12$.
\end{lemma}

\begin{lemma}\label{lemma:LBsubG}
Let $X_1,\ldots,X_n$ be independent non-negative random variables. Then
$$\P\left( \sum_{i=1}^n (X_i - \Ep[X_i]) \leq -t \right)\leq \exp\left( - \frac{t^2}{2\sum_{i=1}^n\Ep[X_i]}\right).  $$
\end{lemma}

\vspace{-.2in}

%The following lemma is a version of a result in \cite{RudelsonVershynin2008}, see \cite{belloni2015uniformly} for a proof of a more general result.
%\begin{lemma}\label{thm:RV34}
%Let $(X_{i})$, $i=1,\ldots, n$, be independent (across i) random vectors such that $X_{i} \in \mathbb{R}^p$ with $p\geq 3$ and $(\Ep[ \max_{1\leq i\leq n}\|X_{i}\|_\infty^2])^{1/2} \leq K$. Furthermore, for $k\geq 1$, define
%$$
%\delta_n:= \frac{K \sqrt{k}}{\sqrt n}\left(\log^{1/2} p + (\log k) (\log^{1/2}p) (\log^{1/2} n) \right),
%$$
%%where $A$ is a universal constant.
%Then,
%$$
%\Ep\left[ \sup_{\|\theta\|_0\leq k, \|\theta\| =1} \left| \frac{1}{n}\sum_{i=1}^n\{ (\theta'X_i)^2 - \Ep[(\theta'X_i)^2\} ]\right|\right] \lesssim \delta_n^2 + \delta_n \sup_{\|\theta\|_0\leq k, \|\theta\| =1} \sqrt{\frac{1}{n}\sum_{i=1}^n\Ep[(\theta'X_i)^2]}
%$$
%up-to a universal constant.
%\end{lemma}

{\footnotesize \bibliography{EIVbib-Pivotalv2}

\begin{thebibliography}{10}

\bibitem{BRT2014}
{Alexandre Belloni}, {Mathieu Rosenbaum}, and {Alexandre~B. Tsybakov}.
\newblock Linear and conic programming approaches to high-dimensional
  errors-in-variables models.
\newblock {\em Journal of the Royal Statistical Society, Series B},
  79:939--956, 2017.

\bibitem{BCKRT2016b}
{Alexandre Belloni}, {Victor Chernozhukov}, and {Abhishek Kaul}.
\newblock Confidence bands for coefficients in high dimensional linear models
  with error-in-variables.
\newblock {\em Arxiv 1703.00469}, 2017.

\bibitem{BellChenChernHans:nonGauss}
Alexandre Belloni, Daniel Chen, Victor Chernozhukov, and Christian Hansen.
\newblock Sparse models and methods for optimal instruments with an application
  to eminent domain.
\newblock {\em Econometrica}, 80:2369--2429, 2012.
\newblock Arxiv, 2010.

\bibitem{belloni2016quantile}
Alexandre Belloni, Mingli Chen, and Victor Chernozhukov.
\newblock Quantile graphical models: prediction and conditional independence
  with applications to financial risk management.
\newblock {\em arXiv preprint arXiv:1607.00286}, 2016.

\bibitem{BCW-SqLASSO}
Alexandre Belloni, Victor Chernozhukov, and Lie Wang.
\newblock Square-root-lasso: Pivotal recovery of sparse signals via conic
  programming.
\newblock {\em Biometrika}, 98(4):791--806, 2011.

\bibitem{BCW-SqLASSO2}
Alexandre Belloni, Victor Chernozhukov, and Lie Wang.
\newblock Pivotal estimation via square-root lasso in nonparametric regression.
\newblock {\em The Annals of Statistics}, 42(2):757--788, 2014.

\bibitem{BRT2014b}
Alexandre Belloni, Mathieu Rosenbaum, and Alexandre~B. Tsybakov.
\newblock An \{$L_1, L_2, L_\infty$\}-approach to high-dimensional
  errors-in-variables models.
\newblock {\em Electronic Journal of Statistics}, 10(2):1729--1750, 2016.

\bibitem{BickelRitovTsybakov2009}
Peter~J. Bickel, Ya’acov Ritov, and Alexandre~B. Tsybakov.
\newblock Simultaneous analysis of {L}asso and {D}antzig selector.
\newblock {\em The Annals of Statistics}, 37(4):1705--1732, 2009.

\bibitem{buhlmann2011statistics}
Peter B{\"u}hlmann and Sara Van De~Geer.
\newblock {\em Statistics for high-dimensional data: methods, theory and
  applications}.
\newblock Springer Science \& Business Media, 2011.

\bibitem{chen2005measurement}
Xiaohong Chen, Han Hong, and Elie Tamer.
\newblock Measurement error models with auxiliary data.
\newblock {\em The Review of Economic Studies}, 72(2):343--366, 2005.

\bibitem{CChen}
Yudong Chen and Constantine Caramanis.
\newblock Orthogonal matching pursuit with noisy and missing data: Low and
  high-dimensional results.
\newblock {\em arXiv:1206.0823}, 2012.

\bibitem{CChen1}
Yudong Chen and Constantine Caramanis.
\newblock Noisy and missing data regression: Distribution-oblivious support
  recovery.
\newblock {\em Proc. of International Conference on Machine Learning (ICML)},
  2013.

\bibitem{datta2015cocolasso}
Abhirup Datta and Hui Zou.
\newblock Cocolasso for high-dimensional error-in-variables regression.
\newblock {\em arXiv preprint arXiv:1510.07123}, 2015.

\bibitem{delapena}
Victor~H. de~la Pe{\~n}a, Tze~Leung Lai, and Qi-Man Shao.
\newblock {\em Self-normalized processes}.
\newblock Probability and its Applications (New York). Springer-Verlag, Berlin,
  2009.
\newblock Limit theory and statistical applications.

\bibitem{Fuller1987}
Wayne~A. Fuller.
\newblock {\em Measurement Error Models}.
\newblock Wiley \& Sons, Inc. New York, 1987.

\bibitem{gautier2011high}
Eric Gautier and Alexandre~B. Tsybakov.
\newblock High-dimensional instrumental variables regression and confidence
  sets.
\newblock {\em arXiv preprint arXiv:1105.2454v4}, 2011.

\bibitem{gautier2012}
Eric Gautier and Alexandre~B. Tsybakov.
\newblock Pivotal uniform inference in high-dimensional regression with random
  design in wide classes of models via linear programming. {U}npublished
  manuscript.
\newblock 2012.

\bibitem{gautier2013pivotal}
Eric Gautier and Alexandre~B. Tsybakov.
\newblock Pivotal estimation in high-dimensional regression via linear
  programming.
\newblock In {\em Empirical Inference}, pages 195--204. Springer, 2013.

\bibitem{GHaus}
Zvi Griliches and Jerry~A Hausman.
\newblock Errors in variables in panel data.
\newblock {\em Journal of econometrics}, 31(1):93--118, 1986.

\bibitem{jing:etal}
Bing-Yi Jing, Qi-Man Shao, and Qiying Wang.
\newblock Self-normalized {C}ram{\'e}r-type large deviations for independent
  random variables.
\newblock {\em The Annals of Probability}, 31(4):2167--2215, 2003.

\bibitem{KK2015}
Abhishek Kaul and Hira~L Koul.
\newblock Weighted $\ell_1$-penalized corrected quantile regression for high
  dimensional measurement error models.
\newblock {\em Journal of Multivariate Analysis}, 140:72--91, 2015.

\bibitem{KKCL2016}
Abhishek Kaul, Hira~L Koul, Akshita Chawla, and Soumendra~N Lahiri.
\newblock Two stage non-penalized corrected least squares for high dimensional
  linear models with measurement error or missing covariates.
\newblock {\em arXiv preprint arXiv:1605.03154}, 2016.

\bibitem{LW}
Po-Ling Loh and Martin~J. Wainwright.
\newblock High-dimensional regression with noisy and missing data: Provable
  guarantees with nonconvexity.
\newblock {\em The Annals of Statistics}, 40(3):1637--1664, 2012.

\bibitem{meinshausen2007relaxed}
Nicolai Meinshausen.
\newblock Relaxed lasso.
\newblock {\em Computational Statistics \& Data Analysis}, 52(1):374--393,
  2007.

\bibitem{CRSC2006}
Leonard A.~Stefanski Raymond J.~Carroll, David~Ruppert and Ciprian Crainiceanu.
\newblock {\em Measurement Error in Nonlinear Models: A Modern Perspective}.
\newblock Chapman \& Hall, New York, 2006.

\bibitem{reilly1995mean}
Marie Reilly and Margaret~Sullivan Pepe.
\newblock A mean score method for missing and auxiliary covariate data in
  regression models.
\newblock {\em Biometrika}, 82(2):299--314, 1995.

\bibitem{RT1}
Mathieu Rosenbaum and Alexandre~B. Tsybakov.
\newblock Sparse recovery under matrix uncertainty.
\newblock {\em The Annals of Statistics}, 38(5):2620--2651, 2010.

\bibitem{RT2}
Mathieu Rosenbaum and Alexandre~B. Tsybakov.
\newblock Improved matrix uncertainty selector.
\newblock In {\em From Probability to Statistics and Back: High-Dimensional
  Models and Processes--A Festschrift in Honor of Jon A. Wellner}, pages
  276--290. Institute of Mathematical Statistics, 2013.

\bibitem{RZ2015}
Mark Rudelson and Shuheng Zhou.
\newblock High dimensional errors-in-variables models with dependent
  measurements.
\newblock {\em arXiv preprint arXiv:1502.02355}, 2015.

\bibitem{SFT2}
{\O}ystein S{\o}rensen, Arnoldo Frigessi, and Magne Thoresen.
\newblock Covariate selection in high-dimensional generalized linear models
  with measurement error.
\newblock {\em arXiv preprint arXiv:1407.1070}, 2014.

\bibitem{SFT1}
{\O}ystein S{\o}rensen, Arnoldo Frigessi, and Magne Thoresen.
\newblock Measurement error in lasso: Impact and correction.
\newblock {\em Statistica Sinica}, 25(2):809--829, 2015.

\bibitem{sun2012scaled}
Tingni Sun and Cun-Hui Zhang.
\newblock Scaled sparse linear regression.
\newblock {\em Biometrika}, 99(4):879--898, 2012.

\bibitem{van2009conditions}
Sara~A Van De~Geer and Peter B{\"u}hlmann.
\newblock On the conditions used to prove oracle results for the lasso.
\newblock {\em Electronic Journal of Statistics}, 3:1360--1392, 2009.

\end{thebibliography}
\bibliographystyle{plain}
}
\end{appendix}

\newpage

\section{\large Supplementary Material: Deferred Proofs}

\begin{proof}[Proof of Lemma \ref{lemma:auxPhi}]
Note that $\log ( 1 - \Phi(t) )$ is a concave function so that the supergradient inequality yields
$$ \log\{1-\Phi(a/\{1+\gamma\})\} \leq \log\{1-\Phi(a)\}+ \frac{\phi(a)}{1-\Phi(a)}\frac{\gamma}{1+\gamma}a,$$ where $\phi$ denotes the Gaussian density function.
The result follows by noting that $\frac{\phi(t)}{1-\Phi(t)}\leq \{t^2+1\}/t \leq 2t$ if $t\geq 1$, and exponentiating both sides of the inequality.\end{proof}

\begin{proof}[Proof of Theorem \ref{thm:trim}]
Recall that $t_j(\beta) = \{\frac{1}{n}\sum_{i=1}^n \{z_{ij}(y_i-z_i^T\beta)+\hat \Gamma_{jj}\beta_{j}\}^2\}^{1/2}$ and define $\tilde t_j(\beta) = \{\frac{1}{n}\sum_{i=1}^n \{z_{ij}(y_i-z_i^T\beta)+ \Gamma_{jj}\beta_{j}\}^2\}^{1/2}$. Then we can write the threshold in the $j$th component of the estimator as $\bar v_j = \tau t_j(\hat\beta) /\frac{1}{n}\sum_{i=1}^nz_{ij}^2$, $j=1,\ldots,p.$ Further, note that $\Ep[ \tilde t_j^2(\beta_0)] = \frac{1}{n}\sum_{i=1}^n \Ep[\{z_{ij}(\xi_i-w_i^T\beta_0)+\Gamma_{jj}\beta_{0j}\}^2]$.\\

We first derive upper and lower bounds on $t_j(\hat\beta)$  by controlling the value
$$ |t_j(\hat\beta) - \{\Ep[\tilde t_j^2(\beta_0)]\}^{1/2}| \leq |t_j(\hat\beta) -  t_j(\beta_0)|+|t_j(\beta_0) - \tilde t_j(\beta_0)|+|\tilde t_j(\beta_0)-\{\Ep[\tilde t_j^2(\beta_0)]\}^{1/2}| $$
via triangle inequality and using the bracketing
\begin{equation}\label{brac}
c(1+\|\beta_0\|_2)^2 \leq \Ep[\tilde t_j^2(\beta_0)] \leq C(1+\|\beta_0\|_2)^2
\end{equation}
that holds by assumption.

By (\ref{bound:t}) we have with probability $1-o(1)$
$$ |t_j(\hat\beta)- t_j(\beta_0)| \leq H_n\|\hat\beta-\beta_0\|_1 \leq C\|\hat\beta-\beta_0\|_1 = o(1+\|\beta_0\|_2)$$
since $H_n \leq C$ with probability $1-o(1)$ by Lemma \ref{lem:Hn} when $X \in \Omega_X$, Condition A and B hold, and $\|\hat\beta-\beta_0\|_1 = o(1+\|\beta_0\|_2)$ by Theorem \ref{thm:pivotal} with $\alpha = \log n$ under the assumed condition $s^{1/q}\sqrt{\log(2pn/\alpha)/n} = o(1)$ for $q\in\{1,2\}$.

{} Moreover, we have  $$|t_j(\beta_0)-\tilde t_j(\beta_0)|\leq|\hat\Gamma_{jj}-\Gamma_{jj}|\ |\beta_{0j}| \leq b_\varepsilon \|\beta_0\|_2=o(\|\beta_0\|_2)$$ under our conditions that imply $b_\varepsilon = o(1)$.

{} Next, Lemma \ref{lem:m2bound} implies
$$
\begin{array}{rl}
\Ep\left[ \max_{j\leq p}|\tilde t_j^2(\beta_0)- \Ep[\tilde t_j^2(\beta_0)]| \right] & \leq \frac{CM_2 \log(3p)}{n} \\
& + \sqrt{\frac{CM_2\log(2p)}{n}}  \max_{j\leq p}\{\Ep[\tilde t_j^2(\beta_0)]\}^{1/2}\\ \end{array}
$$
where $M_2:= \Ep[\max_{i\leq n, j\leq p} |z_{ij}(\xi_i-w_i^T\beta_0)+\Gamma_{jj}\beta_{0j}|^2]$. The quantity $M_2$ satisfies
\begin{equation}\label{bound:M2}\begin{array}{rl}
M_2 & \leq  2\Ep[\max_{i\leq n, j\leq p} |x_{ij}(\xi_i-w_i^T\beta_0)+\Gamma_{jj}\beta_{0j}|^2]\\
& +2\Ep[\max_{i\leq n, j\leq p} |w_{ij}(\xi_i-w_i^T\beta_0)+\Gamma_{jj}\beta_{0j}|^2]\\
& \leq 4\max_{i\leq n, j\leq p} |x_{ij}|^2\Ep[\max_{i\leq n} |\xi_i-w_i^T\beta_0|^2]+8|\Gamma_{jj}\beta_{0j}|^2 \\
& +4\Ep[\max_{i\leq n, j\leq p} |w_{ij}(\xi_i-w_i^T\beta_0)|^2]\\
& \leq C\max_{i\leq n, j\leq p} |x_{ij}|^2 (1+\|\beta_0\|_2)^2\log(n)+C\|\beta_{0}\|_2^2 \\
& +4\Ep[\max_{i\leq n, j\leq p} |w_{ij}|^4]^{1/2}\Ep[\max_{i\leq n}|\xi_i-w_i^T\beta_0|^4]^{1/2}\\
& \leq C'n^{1/2}(1+\|\beta_0\|_2)^2\log(n) + C\log(pn)(1+\|\beta_0\|_2)^2\log(n)\\
\end{array}\end{equation}
where we used the inequalities $\max_{i\leq n, j\leq p} |x_{ij}| \leq n^{1/4} \max_{i\leq n, j\leq p} \left(\frac{1}{n}\sum_{i=1}^n x_{ij}^4\right)^{1/4} \leq C n^{1/4}$
and Lemma \ref{lem:subgaussian}. Finally, note that
$$\begin{array}{ll}
\displaystyle \max_{j\leq p}|\tilde t_j(\beta_0)- \{\Ep[\tilde t_j^2(\beta_0)]\}^{1/2}|  \leq \displaystyle  \max_{j\leq p}\frac{|\tilde t_j^2(\beta_0)- \Ep[\tilde t_j^2(\beta_0)]|}{\{\Ep[\tilde t_j^2(\beta_0)]\}^{1/2}}\\
\displaystyle \leq \max_{j\leq p}\frac{|\tilde t_j^2(\beta_0)- \Ep[\tilde t_j^2(\beta_0)]|}{c'(1+\|\beta_0\|_2)}\\
= \displaystyle O_P\Big( \frac{C'n^{1/2} + C\log(pn)}{n} \Big) (1+\|\beta_0\|_2)\log(n) \log(pn) \\
+ \displaystyle O_P\Big( \frac{C'n^{1/2} + C\log(pn)}{n} \Big)^{1/2}(1+\|\beta_0\|_2)\log^{1/2}(n) \log^{1/2}(pn) \\
= (1+\|\beta_0\|_2)o_P(1)  \end{array}$$
where we used Markov's inequality,  (\ref{bound:M2}) and the fact that  $\log^2(n)\log^2(pn) = o(n)$, which is due to the relation  $\Phi^{-1}(1-\alpha/(2pn))=o(n^{1/6})$ in Condition B(ii).

Thus, uniformly over $j\in \{1,\ldots,p\}$, we have $$|t_j(\hat\beta)-\{\Ep[\tilde t_j^2(\beta_0)]\}^{1/2}|=(1+\|\beta_0\|_2)o_P(1).
$$
This implies that $|t_j(\hat\beta)|$ satisfies, with high probability, the same bracketing as $|\{\Ep[\tilde t_j^2(\beta_0)]\}^{1/2}|$, cf. \eqref{brac}.
Since $\frac{1}{n}\sum_{i=1}^nz_{ij}^2$ is bounded away from zero and from above by constants uniformly in $j$ with probability $1-o(1)$, we have $\min_{j\leq p} \bar v_j \geq c \tau (1+\|\beta_0\|_2)$ and $\max_{j\leq p} \bar v_j \leq C\tau (1+\|\beta_0\|_2)$. Applying \eqref{b2} in Lemma \ref{Lemma:Bound2nNormSecond} given below with $\nu_{\min}=\min_{j\leq p} \bar v_j$ and Corollary \ref{cor:pivotal} with $q=1$ we get
$$\begin{array}{rl}
 \|\hat\beta_{\hat T}\|_0&   \le s + \frac{\|\hat\beta-\beta_0\|_1}{c\tau(1+\|\beta_0\|_2)} \leq s + \frac{C(1+\|\beta_0\|_2)s\sqrt{\log(c'p/(\alpha\varepsilon))/n}}{c\tau(1+\|\beta_0\|_2)} \leq C' s \end{array}
 $$
where we have used  the fact that $\tau = n^{-1/2}\Phi^{-1}(1-\alpha/(2p))\ge c\sqrt{\log(p/\alpha)/n}$.\\

{} Similarly, to prove the bounds on the $\ell_1$ and $\ell_2$ errors of the thresholded estimator, we use inequalities \eqref{b1} and \eqref{b3} in Lemma \ref{Lemma:Bound2nNormSecond} below with $\nu_{\max}=\max_{j\leq p} \bar v_j$, the bounds of Corollary \ref{cor:pivotal} with $q\in\{1,2\}$, and the fact that $\tau \le \sqrt{2\log(2p/\alpha)/n}$. \end{proof}

The following lemma provides bounds for general thresholded estimators (see also a related lemma in \cite{belloni2016quantile}).

\begin{lemma}\label{Lemma:Bound2nNormSecond}
Let $\hat\beta, \beta_0 \in \mathbb{R}^p$ be such that $\|\beta_0\|_0\leq s$.
Denote by $\hat \beta^{\nu}=(\hat \beta^{\nu}_1,\dots,\hat \beta^{\nu}_p)$ the vector obtained by thresholding the components $\hat \beta_j$ of $\hat\beta$ as follows: $\hat\beta^\nu_j=\hat \beta_j 1\{|\hat \beta_{j}|\geq \nu_j\}$ where $\nu_j$ are positive numbers. Then,
\begin{eqnarray}\label{b1}
&\|\hat \beta^\nu - \beta_0\|_{1} & \leq \|\hat \beta - \beta_0 \|_{1}+s \nu_{\max},\\
\label{b2} &\qquad \|\hat\beta^\nu\|_0  &\leq s + \|\hat \beta - \beta_0 \|_{1}/\nu_{\min},\\
\label{b3}&\|\hat \beta^\nu-\beta_0\|_2 & \leq  \|\hat \beta-\beta_0\|_2  +  2\sqrt{s}\nu_{\max} + \frac{2\|\hat\beta-\beta_0\|_{1}}{\sqrt{s}}\end{eqnarray}
where $\nu_{\max}=\max_{j\leq p} \nu_j$ and $\nu_{\min}=\min_{j\leq p} \nu_j$.
\end{lemma}
\begin{proof}[Proof of Lemma \ref{Lemma:Bound2nNormSecond}]
Let $T = \supp(\beta_0)$. The bound \eqref{b1} follows from the chain of inequalities
$$\begin{array}{rl}
 \|\hat \beta^\nu - \beta_0\|_{1} &= \|(\hat \beta^\nu - \beta_0)_{T}\|_{1} + \|(\hat \beta^\nu)_{T^c}\|_{1} \\
&\leq  \|(\hat \beta^\nu - \hat\beta)_{T}\|_{1} + \|(\hat \beta - \beta_0)_{T}\|_{1}+ \|(\hat \beta^\nu)_{T^c}\|_{1} \\
& \leq s\nu_{\max}  + \|(\hat \beta - \beta_0)_{T}\|_{1}+ \|(\hat \beta)_{T^c}\|_{1}\\
& = s\nu_{\max}   + \|\hat \beta - \beta_0\|_{1}. \end{array}$$
To prove the bound \eqref{b2}, set $\hat T = \supp(\hat\beta^\nu)$, and note that
$$\|\hat\beta-\beta_0\|_{1} \geq \|\hat \beta_{T^c}\|_{1} \ge \|\hat \beta_{T^c\cap \hat T }\|_1 \ge
\nu_{\min}  |T^c\cap \hat T|
\ge \nu_{\min} ( | \hat T|-|T|)\ge \nu_{\min} ( \|\hat\beta^\nu\|_{0}-s).$$
We now show \eqref{b3}. By  the triangle inequality,
\begin{equation}\label{b4} \|\hat\beta^\nu-\beta_0\|_2 \leq \|\hat\beta^\nu-\hat\beta\|_2 + \|\hat\beta-\beta_0\|_2.\end{equation}
Without loss of generality assume that the order of the components is such that $|\hat\beta^\nu_j -\hat\beta_j|$ is non-increasing in $j$. Let $T_1$ be the set of  indices $j$ corresponding to the $s$ largest values of $|\hat\beta^\nu_j-\hat\beta_j|$. Similarly, define $T_k$ as the set of  indices corresponding to the $s$ largest values of $|\hat\beta^\nu_j-\hat\beta_j|$ outside $\cup_{m=1}^{k-1}T_m$. Therefore, $\hat\beta^\nu-\hat\beta=\sum_{k=1}^{\lceil p/s\rceil } (\hat\beta^\nu-\hat\beta)_{T_k}$. Moreover,
$\|(\hat\beta^\nu-\hat\beta)_{T_k}\|_2 \leq \|(\hat\beta^\nu-\hat\beta)_{T_{k-1}}\|_1/\sqrt{s}$ in view of the monotonicity of the components. Thus,
$$ \begin{array}{rl}
 \|\hat\beta^\nu-\hat\beta\|_2 & = \|\sum_{k=1}^{\lceil p/s \rceil} (\hat\beta^\nu-\hat\beta)_{T_k}\|_2\\
& \leq \| (\hat\beta^\nu -\hat\beta )_{T_1}\|_2  + \sum_{k\geq 2}\| (\hat\beta^\nu -\hat\beta )_{T_k}\|_2 \\
 & \leq  \| (\hat\beta^\nu -\hat\beta )_{T_1}\|_2+ \sum_{k\geq 2}\|(\hat\beta^\nu -\hat\beta )_{T_k}\|_2\\
&  \leq   \nu_{\max}\sqrt{s}  +  \sum_{k\geq 1}\|(\hat\beta^\nu -\hat\beta )_{T_k}\|_{1}/\sqrt{s}\\
& =   \nu_{\max}\sqrt{s} +  \|\hat\beta^\nu -\hat\beta \|_{1}/\sqrt{s} \\
&  \leq  2\sqrt{s}\nu_{\max} + 2\|\hat\beta -\beta_0\|_{1}/\sqrt{s}  \\ \end{array}$$
where we have used the bound $|\hat\beta^\nu_j-\hat\beta_j|\le \nu_j$ valid for all $j$, and then \eqref{b1}. Inequality \eqref{b3} follows by combining the last display with \eqref{b4}.
\end{proof}

\section{Approximate Sparse Models}

In what follows we modify the argument in the proof of Theorem \ref{thm:pivotal} to cover approximate sparse models were $\|\beta_{0T^c}\|_q$ is small but not zero (as in exactly sparse models).

\begin{proof}
By (\ref{eq:aux:approx}), $\lambda = 1/(4H_n)$ and $\lambda_u = 1/4$, we have
$$ \|\hat\beta_{T^c}\|_1 \leq \|\beta_{0T}-\hat\beta_T\|_1 +\|\beta_{0T^c}\|_1 +\frac{1}{2}\|\beta_0-\hat\beta\|_1$$
First suppose that $\|\beta_{0T}-\hat\beta_T\|_1 \leq 2\|\beta_{0T^c}\|_1$. Then
$$ \frac{1}{2}\|\hat\beta_{T^c}\|_1 \leq \frac{3}{2}\|\beta_{0T}-\hat\beta_T\|_1 +\frac{3}{2}\|\beta_{0T^c}\|_1 \leq 3\|\beta_{0T^c}\|_1$$
and we have $\|\hat\beta - \beta_0\|_1 \leq 5\|\beta_{0T^c}\|_1$.
Otherwise, suppose that $\|\beta_{0T}-\hat\beta_T\|_1 \geq 2\|\beta_{0T^c}\|_1$. Then
$$\begin{array}{rl}
 \|\hat\beta_{T^c}-\beta_{0T^c}\|_1 & \leq \|\beta_{0T}-\hat\beta_T\|_1 +2\|\beta_{0T^c}\|_1 +\frac{1}{2}\|\beta_0-\hat\beta\|_1\\
 & \leq (2+\frac{1}{2})\|\beta_{0T}-\hat\beta_T\|_1+ \frac{1}{2}\|\beta_{0T^c}-\hat\beta_{T^c}\|_1\\
 \end{array}$$
and we have that $\hat\beta-\beta_0 \in \Delta_T(5)$. In turn we have
the same bound as in Theorem \ref{thm:pivotal} with $\kappa_1(s,5)$ replacing $\kappa_q(s,3)$. Therefore we can conclude that with probability $1-\alpha-o(1)$
$$\|\hat\beta-\beta_0\|_1 \leq 5\|\beta_{0T^c}\|_1+ C(1+\|\beta_0\|_2)s\sqrt{\frac{\log(c'p/(\alpha\varepsilon))}{n}}.$$
Moreover, we have
$$\begin{array}{rl}
 (\hat\beta-\beta_0)\mbox{$\frac{1}{n}$}X^TX (\hat\beta-\beta_0) & \leq \| \hat\beta-\beta_0\|_1 \| \mbox{$\frac{1}{n}$}X^TX (\hat\beta-\beta_0)\|_\infty\\
 &\leq C\|\beta_{0T^c}\|_1 (1+\|\beta_0\|_2)\sqrt{\frac{\log(c'p/(\alpha\varepsilon))}{n}}\\
 & +
 (1+\|\beta_0\|_2)^2s\frac{\log(c'p/(\alpha\varepsilon))}{n} \end{array}
 $$
\end{proof}

\section{\large Supplementary Material: Numerical and Optimization Issues}

~\\

The pivotality of the self-normalized estimator is achieved by the introduction of $p$ second order cone constraints which is more computationally demanding than the $2p$ linear constraints (associated with the near zero score condition).  In order to solve the optimization problem defined in (\ref{est:pivotal}), it is convenient to formulate it as
$$
\begin{array}{rl}
{\displaystyle \min_{\theta,\nu}} & c^T\theta\\
s.t. & A\theta+\nu = b\\
& (\theta,\nu) \in \mathbb{R}^d\times \mathcal{K},\\
\end{array}
$$
where $\theta$ is a vector that contains the positive and negative parts of $\beta$ and auxiliary variables, and $\mathcal{K}$ is the cartesian product of non-negative cones and second order cones. The introduction of residual variables $\varepsilon_i=y_i-z_i^T(\beta^+-\beta^-)$ proves to be helpful in the implementation to further exploit sparsity in the design matrix in the $p$ second order constraints. Indeed the additional residual variables allow us to write the second order constraints for $j=1,\ldots,p$ as
$$ \Big\{ \frac{1}{n}\sum_{i=1}^n \{z_{ij} \varepsilon_i +\hat \Gamma_{jj}(\beta_j^+-\beta_j^-)\}^2 \Big\}^{1/2} \leq t_j$$ instead of $$\Big\{ \frac{1}{n}\sum_{i=1}^n \{z_{ij} (y_i - z_i^T(\beta^+-\beta^-)) +\hat\Gamma_{jj}(\beta_j^+-\beta_j^-)\}^2 \Big\}^{1/2} \leq t_j. \ \ $$
The difference in these representations come from what multiplies $\beta^+-\beta^-$. In the first formulation, ${\rm diag}(z_{1j},z_{2j},\ldots,z_{pj})$ multiplies $\varepsilon$ while $e\Gamma_{j\cdot}$ multiplies $\beta^+-\beta^-$. In the second formulation we have $-z_{\cdot j}z_i^T+e\Gamma_{j\cdot}$ multiplying $\beta^+-\beta^-$ where $e$ is the $n$-vector of ones. These three matrices have $n$ rows but they are sparse in the first formulation and typically dense in the second formulation. Since these matrices are formed $p$ times (one for each $j$), this has non-negligible consequences on the software performance.

\section{\large Supplementary Material: Additional Simulations}

~\\

% Please add the following required packages to your document preamble:
% \usepackage{multirow}
% Table generated by Excel2LaTeX from sheet 'Additive_separated'

%%%%%%%%%%%%%%%%%NEW BELOW 6/23/19 %%%%%%%%%%%%%%%%%%%%%%%%%%%%%%%%%%%%%%%%%%%%%%%

% Please add the following required packages to your document preamble:
% \usepackage{multirow}
\begin{table}[htbp]
	\begin{tabular}{cccccccccccc}
		\toprule
		\multirow{2}{*}{$n=400$}                                                                      & \multicolumn{11}{c}{$\beta=(1,1,1,1,1,1,0,....,0)^T$}                                                                                                         \\ \cmidrule{2-12} 
		& \textbf{p} & \textbf{Bias} & \textbf{RMSE} & \textbf{PRb} & \textbf{L2} & \textbf{L1} & \textbf{PR} & \textbf{FP} & \textbf{TP} & \textbf{FN} & \textbf{Time} \\ \midrule
		\multirow{4}{*}{\textbf{SN-conic}}                                                            & 10         & 0.29          & 0.61          & 0.06         & 0.60        & 1.26        & 0.55        & 1.50        & 6.00        & 0.00        & 2.08          \\
		& 100        & 0.38          & 0.69          & 0.06         & 0.66        & 1.54        & 0.65        & 8.54        & 6.00        & 0.00        & 6.15          \\
		& 400        & 0.46          & 0.74          & 0.07         & 0.72        & 1.74        & 0.74        & 27.69       & 5.99        & 0.01        & 23.11         \\
		& 750        & 0.49          & 0.75          & 0.09         & 0.72        & 1.76        & 0.75        & 42.23       & 5.98        & 0.02        & 56.41         \\ \midrule
		\multirow{4}{*}{\textbf{\begin{tabular}[c]{@{}c@{}}SN-conic \\ (thresholded)\end{tabular}}}   & 10         & 0.29          & 0.61          & 0.06         & 0.60        & 1.23        & 0.55        & 0.21        & 6.00        & 0.00        & 2.08          \\
		& 100        & 0.38          & 0.68          & 0.06         & 0.65        & 1.39        & 0.65        & 0.45        & 5.99        & 0.01        & 6.15          \\
		& 400        & 0.46          & 0.73          & 0.07         & 0.70        & 1.49        & 0.74        & 0.53        & 5.98        & 0.02        & 23.11         \\
		& 750        & 0.49          & 0.73          & 0.09         & 0.71        & 1.50        & 0.75        & 0.48        & 5.97        & 0.03        & 56.41         \\ \midrule
		\multirow{4}{*}{\textbf{\begin{tabular}[c]{@{}c@{}}SN-conic \\ (refitted V1)\end{tabular}}}   & 10         & 0.25          & 0.49          & 0.05         & 0.47        & 0.96        & 0.46        & 1.05        & 6.00        & 0.00        & 1.90          \\
		& 100        & 0.32          & 0.57          & 0.06         & 0.54        & 1.24        & 0.57        & 6.77        & 6.00        & 0.00        & 4.67          \\
		& 400        & 0.37          & 0.64          & 0.07         & 0.61        & 1.54        & 0.65        & 28.94       & 5.99        & 0.01        & 22.09         \\
		& 750        & 0.39          & 0.66          & 0.08         & 0.64        & 1.60        & 0.67        & 39.73       & 5.98        & 0.02        & 53.56         \\ \midrule
		\multirow{4}{*}{\textbf{\begin{tabular}[c]{@{}c@{}}SN-conic \\ (refitted V2)\end{tabular}}}   & 10         & 0.09          & 1.08          & 0.08         & 0.86        & 1.83        & 0.64        & 0.21        & 6.00        & 0.00        & 6.18          \\
		& 100        & 0.17          & 0.99          & 0.07         & 0.89        & 1.94        & 0.68        & 0.45        & 5.99        & 0.01        & 9.77          \\
		& 400        & 0.14          & 1.34          & 0.10         & 1.00        & 2.16        & 0.77        & 0.53        & 5.98        & 0.02        & 27.57         \\
		& 750        & 0.12          & 1.20          & 0.08         & 1.03        & 2.20        & 0.75        & 0.48        & 5.97        & 0.03        & 69.25         \\ \midrule
		\multirow{4}{*}{\textbf{Conic}}                                                               & 10         & 0.47          & 0.78          & 0.08         & 0.76        & 1.56        & 0.76        & 0.29        & 5.99        & 0.01        & 0.46          \\
		& 100        & 0.55          & 0.85          & 0.08         & 0.82        & 1.73        & 0.87        & 0.66        & 5.98        & 0.02        & 0.50          \\
		& 400        & 0.62          & 0.92          & 0.09         & 0.89        & 1.89        & 0.95        & 4.29        & 5.97        & 0.03        & 1.85          \\
		& 750        & 0.63          & 0.91          & 0.10         & 0.88        & 1.87        & 0.94        & 5.43        & 5.97        & 0.03        & 4.50          \\ \midrule
		\multirow{4}{*}{\textbf{Bias Cor. L.S.}}                                                      & 10         & 0.34          & 0.83          & 0.08         & 0.79        & 1.64        & 0.69        & 0.50        & 5.97        & 0.03        & 0.12          \\
		& 100        & 0.42          & 0.89          & 0.08         & 0.85        & 1.89        & 0.78        & 2.42        & 5.95        & 0.05        & 0.11          \\
		& 400        & 0.46          & 1.04          & 0.09         & 0.97        & 2.35        & 0.89        & 5.59        & 5.85        & 0.15        & 0.38          \\
		& 750        & 0.47          & 1.10          & 0.10         & 1.03        & 2.55        & 0.91        & 7.17        & 5.81        & 0.19        & 1.08          \\ \midrule
		\multirow{4}{*}{\textbf{\begin{tabular}[c]{@{}c@{}}Lasso \\ (biased)\end{tabular}}}        & 10         & 0.89          & 0.92          & 0.13         & 0.91        & 2.30        & 1.29        & 3.71        & 6.00        & 0.00        & 0.16          \\
		& 100        & 0.92          & 0.97          & 0.13         & 0.97        & 3.04        & 1.33        & 69.09       & 6.00        & 0.00        & 0.88          \\
		& 400        & 0.95          & 1.04          & 0.14         & 1.04        & 4.32        & 1.36        & 163.68      & 6.00        & 0.00        & 1.37          \\
		& 750        & 0.97          & 1.07          & 0.14         & 1.07        & 5.01        & 1.35        & 235.91      & 6.00        & 0.00        & 1.63          \\ \midrule
		\multirow{4}{*}{\textbf{\begin{tabular}[c]{@{}c@{}}Lasso \\(no meas error)\end{tabular}}} & 10         & 0.18          & 0.24          & 0.03         & 0.24        & 0.49        & 0.28        & 3.32        & 6.00        & 0.00        & 0.12          \\
		& 100        & 0.22          & 0.27          & 0.03         & 0.26        & 0.56        & 0.33        & 9.54        & 6.00        & 0.00        & 0.64          \\
		& 400        & 0.25          & 0.30          & 0.04         & 0.30        & 0.62        & 0.37        & 7.75        & 6.00        & 0.00        & 1.11          \\
		& 750        & 0.26          & 0.31          & 0.04         & 0.31        & 0.64        & 0.38        & 17.37       & 6.00        & 0.00        & 1.44          \\ \bottomrule
	\end{tabular}
\vspace{3mm}
\caption{\label{tab:adsep400}  Simulation A: numerical results at $n=400$ under the additive EIV setup, with a known bias correction matrix and separated regression coefficients.}
\end{table}%

% Please add the following required packages to your document preamble:
% \usepackage{multirow}
\begin{table}[htbp]
	\begin{tabular}{cccccccccccc}
		\toprule
		\multirow{2}{*}{$n=300$}                                                                    & \multicolumn{11}{c}{$\beta=(1,1/2,1/3,1/4,1/5,1/10,0,....,0)^T$}                                                                                              \\ \cmidrule{2-12} 
		& \textbf{p} & \textbf{Bias} & \textbf{RMSE} & \textbf{PRb} & \textbf{L2} & \textbf{L1} & \textbf{PR} & \textbf{FP} & \textbf{TP} & \textbf{FN} & \textbf{Time} \\ \midrule
		\multirow{4}{*}{\textbf{SN-conic}}                                                          & 10         & 0.20          & 0.38          & 0.03         & 0.36        & 0.76        & 0.35        & 0.79        & 5.41        & 0.59        & 2.25          \\
		& 100        & 0.28          & 0.43          & 0.04         & 0.42        & 0.97        & 0.42        & 7.54        & 5.09        & 0.91        & 5.45          \\
		& 400        & 0.32          & 0.45          & 0.04         & 0.45        & 1.11        & 0.47        & 25.51       & 5.10        & 0.90        & 16.91         \\
		& 750        & 0.33          & 0.49          & 0.05         & 0.48        & 1.19        & 0.51        & 20.45       & 4.70        & 1.30        & 111.80        \\ \midrule
		\multirow{4}{*}{\textbf{\begin{tabular}[c]{@{}c@{}}SN-conic \\ (thresholded)\end{tabular}}} & 10         & 0.21          & 0.39          & 0.04         & 0.38        & 0.78        & 0.37        & 0.08        & 4.39        & 1.61        & 2.25          \\
		& 100        & 0.29          & 0.44          & 0.04         & 0.43        & 0.93        & 0.45        & 0.25        & 3.81        & 2.19        & 5.45          \\
		& 400        & 0.34          & 0.47          & 0.05         & 0.46        & 0.99        & 0.50        & 0.37        & 3.72        & 2.28        & 16.91         \\
		& 750        & 0.35          & 0.50          & 0.05         & 0.49        & 1.03        & 0.53        & 0.36        & 3.50        & 2.50        & 111.80        \\ \midrule
		\multirow{4}{*}{\textbf{\begin{tabular}[c]{@{}c@{}}SN-conic \\ (refitted V1)\end{tabular}}} & 10         & 0.22          & 0.40          & 0.03         & 0.39        & 0.80        & 0.34        & 0.67        & 4.64        & 1.36        & 1.93          \\
		& 100        & 0.28          & 0.46          & 0.04         & 0.45        & 1.05        & 0.41        & 5.10        & 4.18        & 1.82        & 5.11          \\
		& 400        & 0.33          & 0.50          & 0.04         & 0.49        & 1.21        & 0.46        & 17.96       & 4.23        & 1.77        & 15.44         \\
		& 750        & 0.32          & 0.53          & 0.05         & 0.52        & 1.32        & 0.49        & 25.88       & 4.05        & 1.95        & 113.16        \\ \midrule
		\multirow{4}{*}{\textbf{\begin{tabular}[c]{@{}c@{}}SN-conic \\ (refitted V2)\end{tabular}}} & 10         & 0.16          & 1.26          & 0.09         & 0.68        & 1.40        & 0.51        & 0.08        & 4.39        & 1.61        & 6.83          \\
		& 100        & 0.14          & 0.71          & 0.06         & 0.62        & 1.33        & 0.48        & 0.25        & 3.81        & 2.19        & 7.59          \\
		& 400        & 0.14          & 0.62          & 0.05         & 0.59        & 1.27        & 0.47        & 0.37        & 3.72        & 2.28        & 22.83         \\
		& 750        & 0.18          & 0.67          & 0.05         & 0.63        & 1.36        & 0.50        & 0.36        & 3.50        & 2.50        & 46.95         \\ \midrule
		\multirow{4}{*}{\textbf{Conic}}                                                             & 10         & 0.33          & 0.44          & 0.05         & 0.44        & 0.90        & 0.50        & 0.06        & 4.62        & 1.38        & 0.45          \\
		& 100        & 0.38          & 0.48          & 0.06         & 0.48        & 1.00        & 0.57        & 0.37        & 4.12        & 1.88        & 0.42          \\
		& 400        & 0.42          & 0.51          & 0.06         & 0.51        & 1.07        & 0.62        & 5.41        & 4.03        & 1.97        & 2.43          \\
		& 750        & 0.42          & 0.53          & 0.06         & 0.52        & 1.08        & 0.63        & 6.94        & 3.96        & 2.04        & 3.68          \\ \midrule
		\multirow{4}{*}{\textbf{Bias Cor. L.S.}}                                                    & 10         & 0.35          & 0.51          & 0.06         & 0.50        & 1.04        & 0.58        & 0.00        & 4.08        & 1.92        & 0.06          \\
		& 100        & 0.43          & 0.54          & 0.06         & 0.54        & 1.13        & 0.67        & 0.03        & 3.56        & 2.44        & 0.07          \\
		& 400        & 0.48          & 0.58          & 0.07         & 0.57        & 1.20        & 0.74        & 0.02        & 3.45        & 2.55        & 0.30          \\
		& 750        & 0.50          & 0.60          & 0.08         & 0.59        & 1.24        & 0.76        & 0.01        & 3.36        & 2.64        & 0.45          \\ \midrule
		\multirow{4}{*}{\textbf{\begin{tabular}[c]{@{}c@{}}Lasso \\ (biased)\end{tabular}}}         & 10         & 0.59          & 0.61          & 0.07         & 0.61        & 1.12        & 0.74        & 3.19        & 5.92        & 0.08        & 0.18          \\
		& 100        & 0.61          & 0.63          & 0.08         & 0.63        & 1.26        & 0.78        & 29.70       & 5.69        & 0.31        & 0.52          \\
		& 400        & 0.63          & 0.65          & 0.08         & 0.65        & 1.43        & 0.80        & 61.57       & 5.62        & 0.38        & 1.29          \\
		& 750        & 0.64          & 0.66          & 0.08         & 0.66        & 1.56        & 0.81        & 82.42       & 5.48        & 0.52        & 1.53          \\ \midrule
		\multirow{4}{*}{\textbf{\begin{tabular}[c]{@{}c@{}}Lasso \\ (no meas error)\end{tabular}}}  & 10         & 0.20          & 0.27          & 0.03         & 0.26        & 0.55        & 0.31        & 2.64        & 5.95        & 0.05        & 0.12          \\
		& 100        & 0.24          & 0.29          & 0.03         & 0.28        & 0.60        & 0.36        & 7.97        & 5.43        & 0.57        & 0.37          \\
		& 400        & 0.27          & 0.31          & 0.04         & 0.31        & 0.65        & 0.40        & 18.08       & 5.32        & 0.68        & 0.96          \\
		& 750        & 0.28          & 0.33          & 0.04         & 0.32        & 0.69        & 0.41        & 5.72        & 4.99        & 1.01        & 0.89          \\ \bottomrule
	\end{tabular}
\vspace{3mm}
\caption{\label{tab:adunsep300}  Simulation A: numerical results at $n=300$ under the additive EIV setup, with a known bias correction matrix and unseparated regression coefficients.}
\end{table}%

% Please add the following required packages to your document preamble:
% \usepackage{multirow}
\begin{table}[htbp]
	\begin{tabular}{cccccccccccc}
	\toprule
		\multirow{2}{*}{$n=400$}                                                                    & \multicolumn{11}{c}{$\beta=(1,1/2,1/3,1/4,1/5,1/10,0,....,0)^T$}                                                                                              \\ \cmidrule{2-12} 
		& \textbf{p} & \textbf{Bias} & \textbf{RMSE} & \textbf{PRb} & \textbf{L2} & \textbf{L1} & \textbf{PR} & \textbf{FP} & \textbf{TP} & \textbf{FN} & \textbf{Time} \\ \midrule
		\multirow{4}{*}{\textbf{SN-conic}}                                                          & 10         & 0.20          & 0.35          & 0.03         & 0.34        & 0.72        & 0.32        & 0.91        & 5.44        & 0.56        & 2.66          \\
		& 100        & 0.21          & 0.38          & 0.04         & 0.37        & 0.87        & 0.37        & 7.24        & 5.27        & 0.73        & 6.60          \\
		& 400        & 0.26          & 0.41          & 0.05         & 0.40        & 1.01        & 0.41        & 14.87       & 5.14        & 0.86        & 23.24         \\
		& 750        & 0.29          & 0.40          & 0.04         & 0.39        & 1.02        & 0.43        & 39.04       & 4.91        & 1.09        & 51.53         \\ \midrule
		\multirow{4}{*}{\textbf{\begin{tabular}[c]{@{}c@{}}SN-conic \\ (thresholded)\end{tabular}}} & 10         & 0.21          & 0.36          & 0.04         & 0.35        & 0.72        & 0.34        & 0.12        & 4.50        & 1.50        & 2.66          \\
		& 100        & 0.22          & 0.39          & 0.04         & 0.38        & 0.83        & 0.39        & 0.41        & 4.29        & 1.71        & 6.60          \\
		& 400        & 0.27          & 0.42          & 0.05         & 0.41        & 0.90        & 0.43        & 0.46        & 4.06        & 1.94        & 23.24         \\
		& 750        & 0.30          & 0.41          & 0.05         & 0.40        & 0.88        & 0.45        & 0.42        & 3.88        & 2.12        & 51.53         \\ \midrule
		\multirow{4}{*}{\textbf{\begin{tabular}[c]{@{}c@{}}SN-conic \\ (refitted V1)\end{tabular}}} & 10         & 0.22          & 0.38          & 0.03         & 0.37        & 0.78        & 0.32        & 0.60        & 4.75        & 1.25        & 2.29          \\
		& 100        & 0.22          & 0.40          & 0.03         & 0.39        & 0.91        & 0.35        & 5.87        & 4.61        & 1.39        & 5.13          \\
		& 400        & 0.26          & 0.46          & 0.04         & 0.44        & 1.11        & 0.40        & 20.81       & 4.48        & 1.52        & 22.43         \\
		& 750        & 0.28          & 0.45          & 0.04         & 0.44        & 1.12        & 0.42        & 28.89       & 4.25        & 1.75        & 50.62         \\ \midrule
		\multirow{4}{*}{\textbf{\begin{tabular}[c]{@{}c@{}}SN-conic \\ (refitted V2)\end{tabular}}} & 10         & 0.07          & 0.52          & 0.04         & 0.47        & 1.00        & 0.36        & 0.12        & 4.50        & 1.50        & 6.95          \\
		& 100        & 0.12          & 0.61          & 0.05         & 0.56        & 1.22        & 0.45        & 0.41        & 4.29        & 1.71        & 9.72          \\
		& 400        & 0.11          & 0.58          & 0.05         & 0.54        & 1.18        & 0.44        & 0.46        & 4.06        & 1.94        & 26.41         \\
		& 750        & 0.16          & 0.64          & 0.05         & 0.57        & 1.26        & 0.46        & 0.42        & 3.88        & 2.12        & 67.47         \\ \midrule
		\multirow{4}{*}{\textbf{Conic}}                                                             & 10         & 0.30          & 0.41          & 0.05         & 0.40        & 0.84        & 0.46        & 0.07        & 4.69        & 1.31        & 0.39          \\
		& 100        & 0.33          & 0.43          & 0.05         & 0.42        & 0.89        & 0.51        & 0.36        & 4.56        & 1.44        & 0.49          \\
		& 400        & 0.36          & 0.45          & 0.06         & 0.45        & 0.95        & 0.55        & 2.25        & 4.39        & 1.61        & 1.54          \\
		& 750        & 0.37          & 0.45          & 0.06         & 0.44        & 0.94        & 0.56        & 7.01        & 4.25        & 1.75        & 4.48          \\ \midrule
		\multirow{4}{*}{\textbf{Bias Cor. L.S.}}                                                    & 10         & 0.32          & 0.46          & 0.05         & 0.45        & 0.93        & 0.51        & 0.03        & 4.29        & 1.71        & 0.06          \\
		& 100        & 0.37          & 0.49          & 0.06         & 0.48        & 1.02        & 0.60        & 0.01        & 3.97        & 2.03        & 0.09          \\
		& 400        & 0.42          & 0.51          & 0.07         & 0.51        & 1.07        & 0.65        & 0.00        & 3.92        & 2.08        & 0.19          \\
		& 750        & 0.43          & 0.51          & 0.07         & 0.51        & 1.07        & 0.67        & 0.05        & 3.75        & 2.25        & 0.58          \\ \midrule
		\multirow{4}{*}{\textbf{\begin{tabular}[c]{@{}c@{}}Lasso \\ (biased)\end{tabular}}}         & 10         & 0.57          & 0.59          & 0.07         & 0.59        & 1.06        & 0.72        & 3.15        & 5.95        & 0.05        & 0.13          \\
		& 100        & 0.59          & 0.60          & 0.08         & 0.60        & 1.19        & 0.75        & 42.56       & 5.84        & 0.16        & 0.75          \\
		& 400        & 0.61          & 0.63          & 0.08         & 0.62        & 1.36        & 0.78        & 56.98       & 5.68        & 0.32        & 0.84          \\
		& 750        & 0.61          & 0.63          & 0.08         & 0.63        & 1.44        & 0.78        & 67.57       & 5.65        & 0.35        & 1.62          \\ \midrule
		\multirow{4}{*}{\textbf{\begin{tabular}[c]{@{}c@{}}Lasso \\ (no meas error)\end{tabular}}}  & 10         & 0.18          & 0.24          & 0.03         & 0.23        & 0.49        & 0.28        & 2.90        & 5.99        & 0.01        & 0.12          \\
		& 100        & 0.21          & 0.25          & 0.03         & 0.25        & 0.52        & 0.32        & 6.70        & 5.43        & 0.57        & 0.55          \\
		& 400        & 0.24          & 0.28          & 0.04         & 0.27        & 0.59        & 0.35        & 11.18       & 5.29        & 0.71        & 0.73          \\
		& 750        & 0.24          & 0.28          & 0.04         & 0.28        & 0.60        & 0.36        & 8.91        & 5.28        & 0.72        & 1.40          \\ \bottomrule
	\end{tabular}
  \vspace{3mm}
\caption{\label{tab:adunsep400}  Simulation A: numerical results at $n=400$ under the additive EIV setup, with a known bias correction matrix and unseparated regression coefficients.}
\end{table}%

% Please add the following required packages to your document preamble:
% \usepackage{multirow}

% Please add the following required packages to your document preamble:
% \usepackage{multirow}
\begin{table}[htbp]
	\begin{tabular}{cccccccccccc}
		\toprule
		\multirow{2}{*}{$n=400$}                                                                    & \multicolumn{11}{c}{$\beta=(1,1,1,1,1,1,0,....,0)^T$}                                                                                                         \\ \cmidrule{2-12} 
		& \textbf{p} & \textbf{Bias} & \textbf{RMSE} & \textbf{PRb} & \textbf{L2} & \textbf{L1} & \textbf{PR} & \textbf{FP} & \textbf{TP} & \textbf{FN} & \textbf{Time} \\ \midrule
		\multirow{4}{*}{\textbf{SN-conic}}                                                          & 10         & 0.17          & 0.37          & 0.04         & 0.36        & 0.73        & 0.33        & 0.86        & 6.00        & 0.00        & 0.86          \\
		& 100        & 0.26          & 0.44          & 0.05         & 0.43        & 0.96        & 0.42        & 4.68        & 6.00        & 0.00        & 4.65          \\
		& 400        & 0.28          & 0.47          & 0.05         & 0.46        & 1.06        & 0.47        & 8.42        & 6.00        & 0.00        & 22.07         \\
		& 750        & 0.31          & 0.48          & 0.05         & 0.46        & 1.10        & 0.48        & 13.75       & 6.00        & 0.00        & 54.42         \\ \midrule
		\multirow{4}{*}{\textbf{\begin{tabular}[c]{@{}c@{}}SN-conic \\ (thresholded)\end{tabular}}} & 10         & 0.17          & 0.37          & 0.04         & 0.35        & 0.71        & 0.33        & 0.02        & 6.00        & 0.00        & 0.86          \\
		& 100        & 0.26          & 0.44          & 0.05         & 0.42        & 0.87        & 0.42        & 0.10        & 6.00        & 0.00        & 4.65          \\
		& 400        & 0.28          & 0.46          & 0.05         & 0.45        & 0.91        & 0.46        & 0.06        & 6.00        & 0.00        & 22.07         \\
		& 750        & 0.31          & 0.47          & 0.05         & 0.46        & 0.93        & 0.48        & 0.07        & 6.00        & 0.00        & 54.42         \\ \midrule
		\multirow{4}{*}{\textbf{\begin{tabular}[c]{@{}c@{}}SN-conic \\ (refitted V1)\end{tabular}}} & 10         & 0.13          & 0.28          & 0.03         & 0.27        & 0.53        & 0.26        & 0.55        & 6.00        & 0.00        & 0.77          \\
		& 100        & 0.18          & 0.33          & 0.04         & 0.32        & 0.68        & 0.34        & 3.37        & 6.00        & 0.00        & 4.28          \\
		& 400        & 0.20          & 0.37          & 0.04         & 0.35        & 0.79        & 0.37        & 7.07        & 6.00        & 0.00        & 22.05         \\
		& 750        & 0.22          & 0.36          & 0.04         & 0.35        & 0.82        & 0.38        & 12.86       & 6.00        & 0.00        & 56.02         \\ \midrule
		\multirow{4}{*}{\textbf{\begin{tabular}[c]{@{}c@{}}SN-conic \\ (refitted V2)\end{tabular}}} & 10         & 0.03          & 0.35          & 0.03         & 0.33        & 0.67        & 0.25        & 0.02        & 6.00        & 0.00        & 0.53          \\
		& 100        & 0.05          & 0.39          & 0.03         & 0.36        & 0.75        & 0.28        & 0.10        & 6.00        & 0.00        & 3.77          \\
		& 400        & 0.02          & 0.40          & 0.03         & 0.38        & 0.78        & 0.29        & 0.06        & 6.00        & 0.00        & 23.91         \\
		& 750        & 0.04          & 0.39          & 0.03         & 0.36        & 0.75        & 0.28        & 0.07        & 6.00        & 0.00        & 66.64         \\ \midrule
		\multirow{4}{*}{\textbf{Conic}}                                                             & 10         & 0.42          & 0.55          & 0.07         & 0.54        & 1.13        & 0.64        & 0.03        & 6.00        & 0.00        & 0.20          \\
		& 100        & 0.53          & 0.65          & 0.08         & 0.64        & 1.34        & 0.77        & 0.03        & 6.00        & 0.00        & 0.33          \\
		& 400        & 0.55          & 0.69          & 0.08         & 0.68        & 1.42        & 0.82        & 0.36        & 6.00        & 0.00        & 1.32          \\
		& 750        & 0.58          & 0.70          & 0.08         & 0.69        & 1.44        & 0.84        & 6.59        & 6.00        & 0.00        & 3.95          \\ \midrule
		\multirow{4}{*}{\textbf{Bias Cor. L.S.}}                                                    & 10         & 0.35          & 0.50          & 0.06         & 0.49        & 1.01        & 0.55        & 0.03        & 6.00        & 0.00        & 0.84          \\
		& 100        & 0.47          & 0.61          & 0.07         & 0.60        & 1.24        & 0.69        & 0.02        & 6.00        & 0.00        & 3.54          \\
		& 400        & 0.50          & 0.65          & 0.08         & 0.64        & 1.33        & 0.75        & 0.02        & 6.00        & 0.00        & 11.87         \\
		& 750        & 0.53          & 0.67          & 0.07         & 0.65        & 1.36        & 0.78        & 0.03        & 6.00        & 0.00        & 23.69         \\ \midrule
		\multirow{4}{*}{\textbf{\begin{tabular}[c]{@{}c@{}}Lasso \\ (biased)\end{tabular}}}         & 10         & 0.48          & 0.52          & 0.07         & 0.52        & 1.20        & 0.70        & 3.62        & 6.00        & 0.00        & 0.04          \\
		& 100        & 0.53          & 0.58          & 0.08         & 0.58        & 1.42        & 0.77        & 40.99       & 6.00        & 0.00        & 0.52          \\
		& 400        & 0.54          & 0.60          & 0.08         & 0.60        & 1.58        & 0.79        & 73.97       & 6.00        & 0.00        & 0.94          \\
		& 750        & 0.56          & 0.61          & 0.08         & 0.61        & 1.68        & 0.80        & 113.00      & 6.00        & 0.00        & 1.71          \\ \midrule
		\multirow{4}{*}{\textbf{\begin{tabular}[c]{@{}c@{}}Lasso \\ (no meas error)\end{tabular}}}  & 10         & 0.18          & 0.23          & 0.03         & 0.23        & 0.48        & 0.27        & 3.36        & 6.00        & 0.00        & 0.05          \\
		& 100        & 0.23          & 0.28          & 0.04         & 0.28        & 0.58        & 0.34        & 10.71       & 6.00        & 0.00        & 0.54          \\
		& 400        & 0.25          & 0.30          & 0.04         & 0.29        & 0.62        & 0.36        & 8.73        & 6.00        & 0.00        & 0.73          \\
		& 750        & 0.27          & 0.31          & 0.04         & 0.31        & 0.65        & 0.39        & 17.29       & 6.00        & 0.00        & 1.33          \\ \bottomrule
	\end{tabular}
\vspace{3mm}
\caption{\label{tab:missep400} Simulation B: numerical results at $n=400$ under the covariates missing at random setup, with an estimated bias correction and separated regression coefficients.}
\end{table}%
% Please add the following required packages to your document preamble:
% \usepackage{multirow}

\begin{table}[htbp]
	\begin{tabular}{cccccccccccc}
	\toprule
		\multirow{2}{*}{$n=300$}                                                                    & \multicolumn{11}{c}{$\beta=(1,1/2,1/3,1/4,1/5,1/10,0,....,0)^T$}                                                                                              \\ \cmidrule{2-12} 
		& \textbf{p} & \textbf{Bias} & \textbf{RMSE} & \textbf{PRb} & \textbf{L2} & \textbf{L1} & \textbf{PR} & \textbf{FP} & \textbf{TP} & \textbf{FN} & \textbf{Time} \\ \midrule
		\multirow{4}{*}{\textbf{SN-conic}}                                                          & 10         & 0.14          & 0.31          & 0.03         & 0.29        & 0.62        & 0.27        & 0.59        & 5.48        & 0.52        & 0.99          \\
		& 100        & 0.19          & 0.32          & 0.04         & 0.31        & 0.71        & 0.32        & 4.09        & 5.37        & 0.63        & 3.68          \\
		& 400        & 0.22          & 0.34          & 0.03         & 0.33        & 0.79        & 0.35        & 9.24        & 5.20        & 0.80        & 16.36         \\
		& 750        & 0.24          & 0.36          & 0.04         & 0.36        & 0.87        & 0.38        & 6.50        & 5.06        & 0.94        & 112.47        \\ \midrule
		\multirow{4}{*}{\textbf{\begin{tabular}[c]{@{}c@{}}SN-conic \\ (thresholded)\end{tabular}}} & 10         & 0.15          & 0.32          & 0.03         & 0.31        & 0.63        & 0.28        & 0.02        & 4.70        & 1.30        & 0.99          \\
		& 100        & 0.21          & 0.34          & 0.04         & 0.33        & 0.69        & 0.35        & 0.07        & 4.29        & 1.71        & 3.68          \\
		& 400        & 0.24          & 0.36          & 0.04         & 0.35        & 0.74        & 0.38        & 0.01        & 4.08        & 1.92        & 16.36         \\
		& 750        & 0.27          & 0.39          & 0.04         & 0.38        & 0.81        & 0.42        & 0.04        & 3.81        & 2.19        & 112.47        \\ \midrule
		\multirow{4}{*}{\textbf{\begin{tabular}[c]{@{}c@{}}SN-conic \\ (refitted V1)\end{tabular}}} & 10         & 0.18          & 0.37          & 0.03         & 0.35        & 0.76        & 0.27        & 0.33        & 4.77        & 1.23        & 0.75          \\
		& 100        & 0.24          & 0.41          & 0.03         & 0.40        & 0.91        & 0.32        & 3.37        & 4.49        & 1.51        & 3.57          \\
		& 400        & 0.26          & 0.43          & 0.04         & 0.42        & 1.00        & 0.34        & 6.61        & 4.35        & 1.65        & 16.14         \\
		& 750        & 0.27          & 0.47          & 0.04         & 0.46        & 1.12        & 0.38        & 10.52       & 4.00        & 2.00        & 114.96        \\ \midrule
		\multirow{4}{*}{\textbf{\begin{tabular}[c]{@{}c@{}}SN-conic \\ (refitted V2)\end{tabular}}} & 10         & 0.05          & 0.34          & 0.03         & 0.32        & 0.68        & 0.24        & 0.02        & 4.70        & 1.30        & 0.48          \\
		& 100        & 0.10          & 0.35          & 0.03         & 0.34        & 0.71        & 0.27        & 0.07        & 4.29        & 1.71        & 3.40          \\
		& 400        & 0.10          & 0.37          & 0.03         & 0.36        & 0.76        & 0.28        & 0.01        & 4.08        & 1.92        & 17.75         \\
		& 750        & 0.16          & 0.42          & 0.03         & 0.41        & 0.87        & 0.31        & 0.04        & 3.81        & 2.19        & 50.57         \\ \midrule
		\multirow{4}{*}{\textbf{Conic}}                                                             & 10         & 0.30          & 0.38          & 0.05         & 0.38        & 0.79        & 0.46        & 0.03        & 4.80        & 1.20        & 0.19          \\
		& 100        & 0.35          & 0.41          & 0.06         & 0.41        & 0.88        & 0.53        & 0.16        & 4.53        & 1.47        & 0.30          \\
		& 400        & 0.37          & 0.43          & 0.06         & 0.43        & 0.91        & 0.56        & 0.35        & 4.37        & 1.63        & 1.28          \\
		& 750        & 0.39          & 0.45          & 0.06         & 0.45        & 0.96        & 0.58        & 5.33        & 4.31        & 1.69        & 3.98          \\ \midrule
		\multirow{4}{*}{\textbf{Bias Cor. L.S.}}                                                    & 10         & 0.38          & 0.45          & 0.06         & 0.45        & 0.94        & 0.57        & 0.00        & 4.24        & 1.76        & 0.75          \\
		& 100        & 0.46          & 0.51          & 0.08         & 0.50        & 1.09        & 0.68        & 0.00        & 3.82        & 2.18        & 2.87          \\
		& 400        & 0.49          & 0.55          & 0.07         & 0.54        & 1.17        & 0.74        & 0.00        & 3.68        & 2.32        & 9.40          \\
		& 750        & 0.51          & 0.57          & 0.07         & 0.56        & 1.21        & 0.76        & 0.00        & 3.47        & 2.53        & 18.50         \\ \midrule
		\multirow{4}{*}{\textbf{\begin{tabular}[c]{@{}c@{}}Lasso \\ (biased)\end{tabular}}}         & 10         & 0.38          & 0.41          & 0.05         & 0.40        & 0.77        & 0.49        & 2.50        & 5.96        & 0.04        & 0.04          \\
		& 100        & 0.41          & 0.44          & 0.06         & 0.43        & 0.86        & 0.55        & 13.86       & 5.53        & 0.47        & 0.26          \\
		& 400        & 0.42          & 0.45          & 0.06         & 0.45        & 0.90        & 0.57        & 21.26       & 5.41        & 0.59        & 0.56          \\
		& 750        & 0.43          & 0.46          & 0.06         & 0.46        & 0.95        & 0.58        & 23.07       & 5.20        & 0.80        & 0.90          \\ \midrule
		\multirow{4}{*}{\textbf{\begin{tabular}[c]{@{}c@{}}Lasso \\ (no meas error)\end{tabular}}}  & 10         & 0.20          & 0.26          & 0.03         & 0.25        & 0.54        & 0.31        & 2.43        & 5.96        & 0.04        & 0.05          \\
		& 100        & 0.25          & 0.29          & 0.04         & 0.29        & 0.62        & 0.37        & 5.30        & 5.41        & 0.59        & 0.32          \\
		& 400        & 0.26          & 0.31          & 0.04         & 0.31        & 0.65        & 0.40        & 11.45       & 5.05        & 0.95        & 0.46          \\
		& 750        & 0.27          & 0.32          & 0.04         & 0.32        & 0.69        & 0.41        & 8.00        & 5.03        & 0.97        & 1.11          \\ \bottomrule
	\end{tabular}
\vspace{3mm}
\caption{\label{tab:misunsep300} Simulation B: numerical results at $n=300$ under the covariates missing at random setup, with an estimated bias correction and unseparated regression coefficients.}
\end{table}%

% Please add the following required packages to your document preamble:
% \usepackage{multirow}
\begin{table}[htbp]
	\begin{tabular}{cccccccccccc}
\toprule
		\multirow{2}{*}{$n=400$}                                                                    & \multicolumn{11}{c}{$\beta=(1,1/2,1/3,1/4,1/5,1/10,0,....,0)^T$}                                                                                              \\ \cmidrule{2-12} 
		& \textbf{p} & \textbf{Bias} & \textbf{RMSE} & \textbf{PRb} & \textbf{L2} & \textbf{L1} & \textbf{PR} & \textbf{FP} & \textbf{TP} & \textbf{FN} & \textbf{Time} \\ \midrule
		\multirow{4}{*}{\textbf{SN-conic}}                                                          & 10         & 0.13          & 0.27          & 0.03         & 0.26        & 0.55        & 0.24        & 0.75        & 5.74        & 0.26        & 0.96          \\
		& 100        & 0.16          & 0.29          & 0.03         & 0.28        & 0.65        & 0.27        & 4.23        & 5.49        & 0.51        & 4.61          \\
		& 400        & 0.18          & 0.31          & 0.03         & 0.30        & 0.72        & 0.30        & 8.74        & 5.41        & 0.59        & 22.34         \\
		& 750        & 0.20          & 0.31          & 0.03         & 0.30        & 0.74        & 0.32        & 10.03       & 5.36        & 0.64        & 52.46         \\ \midrule
		\multirow{4}{*}{\textbf{\begin{tabular}[c]{@{}c@{}}SN-conic \\ (thresholded)\end{tabular}}} & 10         & 0.14          & 0.28          & 0.03         & 0.27        & 0.56        & 0.25        & 0.07        & 4.99        & 1.01        & 0.96          \\
		& 100        & 0.17          & 0.30          & 0.03         & 0.29        & 0.61        & 0.28        & 0.10        & 4.69        & 1.31        & 4.61          \\
		& 400        & 0.20          & 0.32          & 0.03         & 0.31        & 0.64        & 0.33        & 0.05        & 4.38        & 1.62        & 22.34         \\
		& 750        & 0.22          & 0.32          & 0.04         & 0.32        & 0.66        & 0.35        & 0.07        & 4.35        & 1.65        & 52.46         \\ \midrule
		\multirow{4}{*}{\textbf{\begin{tabular}[c]{@{}c@{}}SN-conic \\ (refitted V1)\end{tabular}}} & 10         & 0.17          & 0.33          & 0.02         & 0.32        & 0.68        & 0.25        & 0.58        & 5.08        & 0.92        & 0.77          \\
		& 100        & 0.21          & 0.37          & 0.03         & 0.36        & 0.84        & 0.28        & 3.36        & 4.78        & 1.22        & 4.22          \\
		& 400        & 0.22          & 0.41          & 0.03         & 0.39        & 0.95        & 0.32        & 8.01        & 4.57        & 1.43        & 22.31         \\
		& 750        & 0.24          & 0.40          & 0.03         & 0.39        & 0.96        & 0.32        & 8.97        & 4.61        & 1.39        & 53.33         \\ \midrule
		\multirow{4}{*}{\textbf{\begin{tabular}[c]{@{}c@{}}SN-conic \\ (refitted V2)\end{tabular}}} & 10         & 0.04          & 0.29          & 0.02         & 0.28        & 0.59        & 0.21        & 0.07        & 4.99        & 1.01        & 0.57          \\
		& 100        & 0.06          & 0.30          & 0.02         & 0.29        & 0.60        & 0.22        & 0.10        & 4.69        & 1.31        & 3.66          \\
		& 400        & 0.08          & 0.35          & 0.03         & 0.33        & 0.70        & 0.25        & 0.05        & 4.38        & 1.62        & 23.35         \\
		& 750        & 0.09          & 0.33          & 0.03         & 0.32        & 0.67        & 0.25        & 0.07        & 4.35        & 1.65        & 64.62         \\ \midrule
		\multirow{4}{*}{\textbf{Conic}}                                                             & 10         & 0.27          & 0.34          & 0.04         & 0.33        & 0.69        & 0.41        & 0.01        & 5.10        & 0.90        & 0.22          \\
		& 100        & 0.30          & 0.37          & 0.05         & 0.36        & 0.77        & 0.46        & 0.00        & 4.75        & 1.25        & 0.30          \\
		& 400        & 0.33          & 0.39          & 0.05         & 0.38        & 0.82        & 0.50        & 0.82        & 4.58        & 1.42        & 1.26          \\
		& 750        & 0.34          & 0.40          & 0.05         & 0.39        & 0.84        & 0.52        & 0.22        & 4.72        & 1.28        & 4.51          \\ \midrule
		\multirow{4}{*}{\textbf{Bias Cor. L.S.}}                                                    & 10         & 0.33          & 0.40          & 0.05         & 0.39        & 0.83        & 0.50        & 0.00        & 4.61        & 1.39        & 0.82          \\
		& 100        & 0.39          & 0.45          & 0.06         & 0.44        & 0.95        & 0.58        & 0.00        & 4.15        & 1.85        & 3.03          \\
		& 400        & 0.43          & 0.48          & 0.06         & 0.48        & 1.04        & 0.64        & 0.00        & 4.01        & 1.99        & 9.84          \\
		& 750        & 0.45          & 0.50          & 0.07         & 0.50        & 1.08        & 0.68        & 0.00        & 4.03        & 1.97        & 19.32         \\ \midrule
		\multirow{4}{*}{\textbf{\begin{tabular}[c]{@{}c@{}}Lasso \\ (biased)\end{tabular}}}         & 10         & 0.35          & 0.38          & 0.05         & 0.37        & 0.70        & 0.46        & 3.30        & 6.00        & 0.00        & 0.04          \\
		& 100        & 0.38          & 0.40          & 0.05         & 0.40        & 0.77        & 0.49        & 13.32       & 5.69        & 0.31        & 0.45          \\
		& 400        & 0.40          & 0.42          & 0.05         & 0.42        & 0.82        & 0.52        & 22.69       & 5.47        & 0.53        & 0.80          \\
		& 750        & 0.41          & 0.43          & 0.05         & 0.43        & 0.86        & 0.54        & 27.03       & 5.44        & 0.56        & 1.42          \\ \midrule
		\multirow{4}{*}{\textbf{\begin{tabular}[c]{@{}c@{}}Lasso \\ (no meas error)\end{tabular}}}  & 10         & 0.17          & 0.24          & 0.03         & 0.23        & 0.48        & 0.27        & 2.78        & 5.97        & 0.03        & 0.06          \\
		& 100        & 0.20          & 0.25          & 0.03         & 0.25        & 0.53        & 0.31        & 9.17        & 5.48        & 0.52        & 0.51          \\
		& 400        & 0.23          & 0.27          & 0.03         & 0.27        & 0.57        & 0.34        & 16.46       & 5.37        & 0.63        & 0.67          \\
		& 750        & 0.24          & 0.28          & 0.04         & 0.28        & 0.60        & 0.37        & 11.55       & 5.20        & 0.80        & 1.19          \\ \bottomrule
	\end{tabular}
\vspace{3mm}
\caption{\label{tab:misunsep400} Simulation B: numerical results at $n=400$ under the covariates missing at random setup, with an estimated bias correction and unseparated regression coefficients.}
\end{table}%

% Please add the following required packages to your document preamble:
% \usepackage{multirow}

% Please add the following required packages to your document preamble:
% \usepackage{multirow}
\begin{table}[htbp]
	\begin{tabular}{clllllllllll}
\toprule
		\multirow{2}{*}{$n=400$}                                                                & \multicolumn{11}{c}{$\beta=c(1,1,1,1,1,1,0,...,0)^T$}                                                                                                                                                                                                                                                                                                                                     \\ \cmidrule{2-12} 
		& \multicolumn{1}{c}{\textbf{p}} & \multicolumn{1}{c}{\textbf{Bias}} & \multicolumn{1}{c}{\textbf{RMSE}} & \multicolumn{1}{c}{\textbf{PRb}} & \multicolumn{1}{c}{\textbf{L2}} & \multicolumn{1}{c}{\textbf{L1}} & \multicolumn{1}{c}{\textbf{PR}} & \multicolumn{1}{c}{\textbf{FP}} & \multicolumn{1}{c}{\textbf{TP}} & \multicolumn{1}{c}{\textbf{FN}} & \multicolumn{1}{c}{\textbf{Time}} \\ \midrule
		\multirow{4}{*}{\textbf{Bias Cor. L.S.}}                                                & 10                             & 0.17                              & 0.94                              & 0.07                             & 0.86                            & 1.96                            & 0.67                            & 1.59                            & 5.92                            & 0.08                            & 3.82                              \\
		& 100                            & 0.31                              & 1.10                              & 0.08                             & 1.03                            & 2.70                            & 0.83                            & 6.95                            & 5.87                            & 0.13                            & 19.26                             \\
		& 400                            & 0.40                              & 1.10                              & 0.09                             & 1.03                            & 2.76                            & 0.89                            & 9.27                            & 5.79                            & 0.21                            & 71.63                             \\
		& 750                            & 0.50                              & 1.28                              & 0.10                             & 1.18                            & 3.19                            & 1.01                            & 10.28                           & 5.69                            & 0.31                            & 202.80                            \\ \midrule
		\multirow{4}{*}{\textbf{Conic}}                                                         & 10                             & 0.24                              & 0.72                              & 0.06                             & 0.69                            & 1.51                            & 0.59                            & 1.33                            & 6.00                            & 0.00                            & 10.65                             \\
		& 100                            & 0.33                              & 0.79                              & 0.07                             & 0.76                            & 2.09                            & 0.68                            & 8.91                            & 6.00                            & 0.00                            & 48.40                             \\
		& 400                            & 0.47                              & 0.78                              & 0.07                             & 0.77                            & 2.15                            & 0.75                            & 15.02                           & 6.00                            & 0.00                            & 280.43                            \\
		& 750                            & 0.53                              & 0.82                              & 0.08                             & 0.80                            & 2.24                            & 0.80                            & 17.58                           & 5.99                            & 0.01                            & 890.59                            \\ \midrule
		\multirow{4}{*}{\textbf{\begin{tabular}[c]{@{}c@{}}Lasso \\ (biased)\end{tabular}}}     & 10                             & 1.08                              & 1.10                              & 0.16                             & 1.10                            & 2.66                            & 1.59                            & 0.79                            & 6.00                            & 0.00                            & 0.14                              \\
		& 100                            & 1.13                              & 1.15                              & 0.17                             & 1.15                            & 2.78                            & 1.68                            & 1.09                            & 6.00                            & 0.00                            & 0.35                              \\
		& 400                            & 1.19                              & 1.21                              & 0.16                             & 1.20                            & 2.92                            & 1.76                            & 1.37                            & 6.00                            & 0.00                            & 1.69                              \\
		& 750                            & 1.21                              & 1.23                              & 0.17                             & 1.23                            & 2.97                            & 1.80                            & 1.16                            & 6.00                            & 0.00                            & 1.51                              \\ \midrule
		\multirow{4}{*}{\textbf{\begin{tabular}[c]{@{}c@{}}Lasso\\ (no meas err)\end{tabular}}} & 10                             & 0.18                              & 0.25                              & 0.03                             & 0.24                            & 0.51                            & 0.28                            & 0.05                            & 6.00                            & 0.00                            & 0.13                              \\
		& 100                            & 0.20                              & 0.26                              & 0.03                             & 0.25                            & 0.53                            & 0.30                            & 0.33                            & 6.00                            & 0.00                            & 0.32                              \\
		& 400                            & 0.23                              & 0.28                              & 0.03                             & 0.28                            & 0.58                            & 0.34                            & 0.22                            & 6.00                            & 0.00                            & 0.87                              \\
		& 750                            & 0.25                              & 0.30                              & 0.03                             & 0.29                            & 0.61                            & 0.36                            & 0.43                            & 6.00                            & 0.00                            & 0.95                              \\ \bottomrule
	\end{tabular}
\vspace{3mm}
\caption{\label{tab:CVsepn400} Simulation C: numerical results at $n=400$ for comparative methods tuned via cross validation, under the additive EIV setup with separated regression coefficients.}
\end{table}%

% Please add the following required packages to your document preamble:
% \usepackage{multirow}
\begin{table}[htbp]
	\begin{tabular}{clllllllllll}
		\toprule
		\multirow{2}{*}{$n=300$}                                                                & \multicolumn{11}{c}{$\beta=(1,1/2,1/3,1/4,1/5,1/10,0,...,0)^T$}                                                                                                                                                                                                                                                                                                                           \\ \cmidrule{2-12} 
		& \multicolumn{1}{c}{\textbf{p}} & \multicolumn{1}{c}{\textbf{Bias}} & \multicolumn{1}{c}{\textbf{RMSE}} & \multicolumn{1}{c}{\textbf{PRb}} & \multicolumn{1}{c}{\textbf{L2}} & \multicolumn{1}{c}{\textbf{L1}} & \multicolumn{1}{c}{\textbf{PR}} & \multicolumn{1}{c}{\textbf{FP}} & \multicolumn{1}{c}{\textbf{TP}} & \multicolumn{1}{c}{\textbf{FN}} & \multicolumn{1}{c}{\textbf{Time}} \\ \midrule
		\multirow{4}{*}{\textbf{Bias Cor. L.S.}}                                                & 10                             & 0.12                              & 0.55                              & 0.04                             & 0.52                            & 1.17                            & 0.42                            & 1.44                            & 4.89                            & 1.11                            & 2.58                              \\
		& 100                            & 0.21                              & 0.58                              & 0.05                             & 0.56                            & 1.51                            & 0.51                            & 6.59                            & 4.12                            & 1.88                            & 12.56                             \\
		& 400                            & 0.23                              & 0.62                              & 0.05                             & 0.6                             & 1.76                            & 0.55                            & 10.38                           & 3.81                            & 2.19                            & 53.67                             \\
		& 750                            & 0.27                              & 0.65                              & 0.06                             & 0.62                            & 1.82                            & 0.59                            & 11.46                           & 3.7                             & 2.3                             & 134.51                            \\ \midrule
		\multirow{4}{*}{\textbf{Conic}}                                                         & 10                             & 0.18                              & 0.46                              & 0.04                             & 0.43                            & 0.92                            & 0.38                            & 0.96                            & 5.21                            & 0.79                            & 10.8                              \\
		& 100                            & 0.24                              & 0.46                              & 0.04                             & 0.45                            & 1.25                            & 0.44                            & 8.6                             & 5                               & 1                               & 51.51                             \\
		& 400                            & 0.28                              & 0.46                              & 0.05                             & 0.45                            & 1.24                            & 0.48                            & 11.32                           & 4.66                            & 1.34                            & 329.57                            \\
		& 750                            & 0.33                              & 0.49                              & 0.05                             & 0.47                            & 1.28                            & 0.52                            & 13.44                           & 4.7                             & 1.3                             & 893.66                            \\ \midrule
		\multirow{4}{*}{\textbf{\begin{tabular}[c]{@{}c@{}}Lasso \\ (biased)\end{tabular}}}     & 10                             & 0.67                              & 0.69                              & 0.09                             & 0.68                            & 1.29                            & 0.87                            & 0.2                             & 4.98                            & 1.02                            & 0.14                              \\
		& 100                            & 0.7                               & 0.71                              & 0.09                             & 0.71                            & 1.4                             & 0.93                            & 0.62                            & 4.68                            & 1.32                            & 0.32                              \\
		& 400                            & 0.73                              & 0.74                              & 0.1                              & 0.74                            & 1.48                            & 0.98                            & 0.61                            & 4.48                            & 1.52                            & 1.18                              \\
		& 750                            & 0.75                              & 0.77                              & 0.1                              & 0.76                            & 1.53                            & 1.01                            & 0.64                            & 4.39                            & 1.61                            & 1.2                               \\ \midrule
		\multirow{4}{*}{\textbf{\begin{tabular}[c]{@{}c@{}}Lasso\\ (no meas err)\end{tabular}}} & 10                             & 0.19                              & 0.25                              & 0.03                             & 0.25                            & 0.52                            & 0.29                            & 0.1                             & 5.35                            & 0.65                            & 0.12                              \\
		& 100                            & 0.22                              & 0.28                              & 0.03                             & 0.27                            & 0.58                            & 0.34                            & 0.37                            & 5.15                            & 0.85                            & 0.3                               \\
		& 400                            & 0.26                              & 0.31                              & 0.04                             & 0.3                             & 0.66                            & 0.39                            & 0.66                            & 5.1                             & 0.9                             & 0.87                              \\
		& 750                            & 0.27                              & 0.32                              & 0.04                             & 0.32                            & 0.68                            & 0.4                             & 0.39                            & 4.98                            & 1.02                            & 1.01                              \\ \bottomrule
	\end{tabular}
\vspace{3mm}
\caption{\label{tab:CVunsepn300} Simulation C: numerical results at $n=300$ for comparative methods tuned via cross validation, under the additive EIV setup with unseparated regression coefficients.}
\end{table}%

\begin{table}[htbp]
	\begin{tabular}{clllllllllll}
		\toprule
		\multirow{2}{*}{$n=400$}                                                                & \multicolumn{11}{c}{$\beta=(1,1/2,1/3,1/4,1/5,1/10,0,...,0)^T$}                                                                                                                                                                                                                                                                                                                           \\ \cmidrule{2-12} 
		& \multicolumn{1}{c}{\textbf{p}} & \multicolumn{1}{c}{\textbf{Bias}} & \multicolumn{1}{c}{\textbf{RMSE}} & \multicolumn{1}{c}{\textbf{PRb}} & \multicolumn{1}{c}{\textbf{L2}} & \multicolumn{1}{c}{\textbf{L1}} & \multicolumn{1}{c}{\textbf{PR}} & \multicolumn{1}{c}{\textbf{FP}} & \multicolumn{1}{c}{\textbf{TP}} & \multicolumn{1}{c}{\textbf{FN}} & \multicolumn{1}{c}{\textbf{Time}} \\ \midrule
		\multirow{4}{*}{\textbf{Bias Cor. L.S.}}                                                & 10                             & 0.10                              & 0.46                              & 0.04                             & 0.44                            & 1.02                            & 0.36                            & 1.48                            & 5.05                            & 0.95                            & 2.33                              \\
		& 100                            & 0.15                              & 0.52                              & 0.05                             & 0.51                            & 1.43                            & 0.44                            & 8.04                            & 4.58                            & 1.42                            & 12.56                             \\
		& 400                            & 0.22                              & 0.54                              & 0.05                             & 0.52                            & 1.54                            & 0.49                            & 10.18                           & 4.27                            & 1.73                            & 50.26                             \\
		& 750                            & 0.22                              & 0.54                              & 0.05                             & 0.53                            & 1.58                            & 0.50                            & 12.49                           & 4.17                            & 1.83                            & 128.48                            \\ \midrule
		\multirow{4}{*}{\textbf{Conic}}                                                         & 10                             & 0.14                              & 0.38                              & 0.03                             & 0.36                            & 0.79                            & 0.32                            & 1.00                            & 5.36                            & 0.64                            & 10.58                             \\
		& 100                            & 0.21                              & 0.41                              & 0.04                             & 0.39                            & 1.06                            & 0.39                            & 6.48                            & 5.04                            & 0.96                            & 53.08                             \\
		& 400                            & 0.24                              & 0.42                              & 0.04                             & 0.41                            & 1.21                            & 0.43                            & 15.59                           & 5.06                            & 0.94                            & 322.41                            \\
		& 750                            & 0.28                              & 0.42                              & 0.04                             & 0.41                            & 1.17                            & 0.44                            & 18.42                           & 4.92                            & 1.08                            & 919.31                            \\ \midrule
		\multirow{4}{*}{\textbf{\begin{tabular}[c]{@{}c@{}}Lasso \\ (biased)\end{tabular}}}     & 10                             & 0.64                              & 0.65                              & 0.09                             & 0.65                            & 1.23                            & 0.83                            & 0.21                            & 5.20                            & 0.80                            & 0.14                              \\
		& 100                            & 0.68                              & 0.69                              & 0.09                             & 0.69                            & 1.34                            & 0.89                            & 0.29                            & 5.04                            & 0.96                            & 0.35                              \\
		& 400                            & 0.69                              & 0.70                              & 0.09                             & 0.70                            & 1.38                            & 0.91                            & 0.59                            & 4.90                            & 1.10                            & 2.20                              \\
		& 750                            & 0.72                              & 0.72                              & 0.09                             & 0.72                            & 1.43                            & 0.95                            & 0.44                            & 4.74                            & 1.26                            & 1.68                              \\ \midrule
		\multirow{4}{*}{\textbf{\begin{tabular}[c]{@{}c@{}}Lasso\\ (no meas err)\end{tabular}}} & 10                             & 0.18                              & 0.23                              & 0.03                             & 0.23                            & 0.48                            & 0.27                            & 0.04                            & 5.52                            & 0.48                            & 0.13                              \\
		& 100                            & 0.19                              & 0.25                              & 0.03                             & 0.25                            & 0.54                            & 0.30                            & 0.37                            & 5.31                            & 0.69                            & 0.32                              \\
		& 400                            & 0.22                              & 0.27                              & 0.04                             & 0.26                            & 0.57                            & 0.34                            & 0.77                            & 5.21                            & 0.79                            & 1.56                              \\
		& 750                            & 0.23                              & 0.28                              & 0.03                             & 0.27                            & 0.59                            & 0.35                            & 0.48                            & 5.20                            & 0.80                            & 1.29                              \\ \bottomrule
	\end{tabular}
\vspace{3mm}
\caption{\label{tab:CVunsepn400} Simulation C: numerical results at $n=400$ for comparative methods tuned via cross validation, under the additive EIV setup with unseparated regression coefficients.}
\end{table}%

%%%%%%%%%%%%%%%%%NEW ABOVE 6/23/19%%%%%%%%%%%%%%%%%%%%%%%%%%%%%%%%%%%%%%%%%%%%%%%

\end{document}